\documentclass{article}
\usepackage{graphicx} 
\usepackage{waingarten}
\title{Locally Approximating the Top Eigenvector \\
of Bounded Entry Matrices}

    \author { 
      Nicolas Menand\thanks{University of Pennsylvania. Supported by the National Science Foundation (NSF) under Grant No. CCF-2337993, as well as No. CCF-2045128. \texttt{nmenand@cis.upenn.edu}.}
      \and
      Erik Waingarten\thanks{University of Pennsylvania. Supported by the National Science Foundation (NSF) under Grant No. CCF-2337993. \texttt{ewaingar@seas.upenn.edu}. }
    }

\date{}
\begin{document}
\maketitle

\begin{abstract}
We provide a local computation algorithm to approximate the top eigenvector $\bx \in \R^n$ of a symmetric matrix $A \in \R^{n \times n}$ with entries between $-1$ and $1$, building on the work of Swartworth and Woodruff~\cite{SW25} who show how to approximate the eigenvalues up to additive-$\eps n$ error using $\tilde{O}(1/\eps^4)$ queries. 

Our local computation algorithm has a preprocessing complexity of \smash{$\tilde{O}(1/\eps^4)$} and per-coordinate query complexity of \smash{$\tilde{O}(1/\eps^2)$} for an additive-$\eps n$ approximation whenever \smash{$|\lambda_{\min}(A)| = O(\lambda_{\max}(A))$}. When \smash{$\lambda_{\min}(A)$} greatly exceeds \smash{$\lambda_{\max}(A)$}, our complexity degrades to at most \smash{$\tilde{O}(1/\eps^{6.\ol{6}})$} in preprocessing and \smash{$\tilde{O}(1/\eps^{3.\ol{3}})$} per query. Furthermore, we show a lower bound of $\Omega(n/\eps^2)$ on the total number of queries needed to output an approximately top eigenvector (implying that the per-coordinate query complexity of $\Omega(1/\eps^2)$ is necessary).

As an application, we use our algorithm to provide local computation algorithms for the sparsest-cut and max-cut problems in the dense graph model of Goldreich, Goldwasser, Ron~\cite{GGR98}. By accessing the top eigenvectors (of an approximate normalized adjacency), we implement local versions of Cheeger's inequality and Trevisan's algorithm~\cite{L12} to obtain ``square-root-opt'' approximations in polynomial time (as opposed to exponential-in-$\poly(1/\eps)$ time which is incurred in~\cite{GGR98}).
\end{abstract}
\newpage

\begin{singlespace}
\setcounter{tocdepth}{2}
\tableofcontents
\end{singlespace}

\newpage


\section{Introduction}

Computing the top eigenvector of a symmetric matrix is a fundamental algorithmic primitive and often the starting point of spectral algorithms (see~\cite{KV09c} for a survey on spectral algorithms). A canonical such algorithm receives as input a matrix and utilizes a polynomial-time algorithm to find the top (or top-few) eigenvectors and eigenvalues (to high precision)~\cite{GV00}. This computation, while running in polynomial time, must read the entire input matrix and can be prohibitive when the matrix becomes large (or is only implicitly maintained). This has motivated work on sublinear algorithms for approximation problems in numerical linear algebra (e.g., low-rank approximation~\cite{MM17,MW17, BW18, BCW20, MMS26}, property testing of spectral properties~\cite{BLWZ19, BCJ20}, and eigenvalue approximation~\cite{AN13, SW23, BDDMR24, SW25}, among many others).

We work in the \emph{bounded-entry model}~\cite{BLWZ19}, where algorithms can selectively probe individual entries of an $n\times n$ matrix $A$ satisfying $\|A\|_{\infty} \leq 1$.\footnote{Unconditionally and without the bounded-entry assumption, sublinear-time spectral approximation via queries is often impossible, since simply distinguishing an all-zeros matrix from one containing a single large entry requires $\Omega(n^2)$ queries.} The model not only captures the natural settings of graph adjacency matrices and many well-studied kernel matrices (e.g., Gaussian~\cite{CS17,CKNS20} and smooth~\cite{BCIS18,CKW24} kernels), but also hits a ``sweet spot'' for sublinear algorithms, often enabling query complexity which is independent of the matrix size. For example,~\cite{BDDMR24,SW25} show that, by querying the entries of a (randomly chosen) $\tilde{O}(1/\eps^2) \times \tilde{O}(1/\eps^2)$ principal submatrix of a bounded-entry symmetric matrix, one can estimate all eigenvalues of the matrix up to an additive $\eps n$ error.\footnote{The bounded-entry model fixes the ``scale'' of the problem. All eigenvalues of an $n\times n$ symmetric matrix $A$ with $\|A \|_{\infty} \leq 1$ must lie within $[-n,n]$, and at most $1/\eps^2$ of them can be larger than $\eps n$ in magnitude.} 
\begin{quote}
    Do there exist sublinear algorithms for outputting an approximate top eigenvector $x\in \R^n$ of a bounded-entry symmetric matrix $A$?
\end{quote}
The above question requires a careful notion of ``sublinear time.'' A top eigenvector is the result of a global optimization problem over the entire matrix $A$, i.e., a maximizing $\langle x, A x \rangle / \|x\|_2^2$, and an inherently $n$-dimensional object (requiring at least time $n$ to output). Hence, we design constant-time algorithms in the model of local computation algorithms~\cite{RTVX11, ARVX12}, where the goal is to provide localized oracle access to the individual entries of an approximate top eigenvector (see Definition~\ref{def:lca}). Local computation algorithms have had multiple historical successes (e.g., maximum independent set~\cite{G22} and maximum matching~\cite{BRR23}), and this work studies local computation algorithms into realm of numerical linear algebra. 

The perspective of local computation algorithms directly leads to new sketching questions in numerical linear algebra, as well as new applications in algorithms for the dense graph model~\cite{GGR98}. The key algorithmic challenge (which we further elaborate on in the technical overview) can be succinctly described as follows. In~\cite{SW25}, the main algorithm shows that the maximum eigenvalue (i.e., the maximizing \emph{value} of Rayleigh quotient with $A$) is preserved when ``sketching'' to a constant-size principal submatrix (roughly speaking, by carefully combining structural properties of bounded-entry matrices with matrix Bernstein inequalities). However, matrix Bernstein inequalities, which show that a small $\bt \times \bt$ ``sketched'' matrix behaves like a large $n \times n$ ``original'' matrix with respect to its large eigenvalues, does not give any actionable correspondence for the eigenvectors---indeed, the eigenvectors in the ``original'' and ``sketched'' matrices lie in completely different spaces, $\R^n$ and $\R^{\bt}$ respectively. In this work, the local computation algorithm must explicitly create a correspondence between the ``sketched'' and ``original'' matrix, so that it may lift the top eigenvector found in the sketched space into an approximately top eigenvector in the original space. With this primitive for locally approximating the top eigenvector, we give local versions of spectral algorithms which rely on top eigenvectors (namely, Cheeger's inequality~\cite{AM85} for sparse cuts, and Trevisan's algorithm~\cite{L12} for max-cut), resulting in local partitioning algorithms with a polynomial dependence on the accuracy parameter, albeit ``square-root-opt'' guarantees.\footnote{The prior classic works~\cite{GGR98} obtain approximability guarantees with respect to the optimum (not square-root-opt), but run in time exponential in the accuracy parameter.}

\subsection{Our Contributions}

Our main result is a local computation algorithm which receives query access to a symmetric $n\times n$ matrix $A$ with bounded entries, and can provide query access to an approximately top eigenvector. For the theorem below, recall that a symmetric $n\times n$ matrix $A$ has $n$ real eigenvalues which we arrange as $\lambda_1(A) \geq \lambda_2(A) \geq \dots \geq \lambda_n(A)$, and we assume without loss of generality that $\lambda_{1}(A)$ is a large constant factor larger than $\eps n$,\footnote{Otherwise, we'll show that a random vector satisfies the desired approximation guarantees (see Remark~\ref{app:lambda-max-large}).} and we let $\lambda_{\max}(A) = \lambda_1(A)$ and $\lambda_{\min}(A) = |\lambda_n(A)|$. Our goal is to output a (randomized) vector $\bx \in \R^n$ whose Rayleigh quotient with $A$ is at least $\lambda_{\max}(A) - \eps n$. In other words, the output vector $\bx$ is the certificate for an additive $\eps n$ approximation on the top eigenvalue.

\begin{theorem}[Main Upper Bound--Informal (see Theorem~\ref{thm:main})]\label{thm:main-intro}
    There exists a (randomized) local computation algorithm $\calA$ which receives query access to a symmetric $n\times n$ matrix $A$ with $\|A \|_{\infty} \leq 1$ and an accuracy parameter $\eps > 0$. 
    \begin{itemize}
        \item $\calA$ makes $\tilde{O}(1/\eps^4)$ non-adaptive queries to $A$ during preprocessing, and for any $i \in [n]$, $\tilde{O}(1/\eps^2)$ non-adaptive queries to output $\bx_i \in \R$.\footnote{The algorithm's queries depend on the randomness the algorithm uses and the index $i \in [n]$ on the query, but not on the specific contents of $A$.}
        \item Letting $\bx \in \R^n$ where $\bx_i$ is the output of $\calA$ on $i \in [n]$, 
        \[ \Prx_{\calA}\left[ \lambda_{\max}(A) - \dfrac{\langle \bx, A \bx\rangle}{\|\bx\|_2^2} \leq \eps n \right] \geq 0.99, \]
        whenever $\|A\|_2 = O(\lambda_{\max}(A))$.
    \end{itemize}
    Furthermore, $\calA$ runs in $\poly(1/\eps)$ time, and uses $\poly(1/\eps) \cdot \log n$ bits of shared randomness.
\end{theorem}

We elaborate on the assumption that $\|A\|_2 = O(\lambda_{\max}(A))$ in Theorem~\ref{thm:main-intro}. When the matrix $A$ is psd (i.e., all eigenvalues are non-negative), we always satisfy $\lambda_{\max}(A) = \|A\|_2$ so the assumption becomes unnecessary. The algorithm does work for (non-psd) symmetric matrices, but one of the technical challenges is that negative eigenvalues of large magnitude may ``interfere'' with the top (positive) eigenvalue (we elaborate on this aspect in the technical overview). In Theorem~\ref{thm:main}, the more expressive query complexity degrades gracefully as a function of how large $\lambda_{\min}(A)$ is compared to $\lambda_{\max}(A)$ (still independent of $n$, but with weaker $\poly(1/\eps)$-dependence). Theorem~\ref{thm:main-intro} is stated for $\lambda_{\min}(A) = O(\lambda_{\max}(A))$, as this assumption natively holds in many desired applications (for example, in spectral graph partitioning). It is an interesting open problem to locally handle the ``negative eigenvalue interference'' without any complexity degradation (for example, when $\lambda_{\min}(A) = \Theta(n)$ and $\lambda_{\max}(A) = \Theta(\eps n)$); in Appendix~\ref{app:negative-interfere} we give a concrete example where our analysis degrades.

\begin{restatable}{theorem}{mainthmlb}\label{thm:main-lb-intro}
    Fix any $\eps \in (1/\sqrt{n}, 1)$ and suppose $\calA$ is an algorithm with access to an $n\times n$ symmetric matrix $A$ with $\| A \|_{\infty} \leq 1$, and one may also assume $\|A\|_2 = O(\lambda_{\max}(A))$, satisfying
    \begin{align*}
        \Prx_{\calA}\left[ \lambda_{\max}(A) - \dfrac{\langle \bx, A \bx\rangle}{\|\bx\|_2^2} \leq \eps n\right] \geq 0.99,
    \end{align*}
    where $\bx \in \R^n$ is the output of $\calA$. Then, $\calA$ must make $\Omega(n/\eps^2)$ queries to $A$.
\end{restatable}

The above lower bound implies that the per-coordinate complexity of our local computation algorithm, i.e., the query complexity of $\tilde{O}(1/\eps^2)$ for outputting the $i$-th entry $\bx_i$ in Theorem~\ref{thm:main-intro}, perfectly amortizes over the $n$ coordinates of $\bx$, matching the lower bound and optimal (up to poly-logarithmic factors in $1/\eps$). Indeed, the total query complexity for outputting the top eigenvector using Theorem~\ref{thm:main-intro} is~$\tilde{O}(1/\eps^4) + n \cdot \tilde{O}(1/\eps^2)$, and Theorem~\ref{thm:main-lb-intro} says that the total query complexity must be $\Omega(n/\eps^2)$. One may hope for improvements to the $\tilde{O}(1/\eps^4)$-query complexity during the preprocessing time, however, Theorem~\ref{thm:main-intro} runs~\cite{SW25} during preprocessing, which is known to require $\Omega(1/\eps^4)$ non-adaptive queries (see the discussion on page 6 of~\cite{SW25}). 

Theorem~9.5 of~\cite{SW25} also gives an algorithm (not a local computation algorithm) which approximates the top eigenvector of a symmetric psd matrix $A$ by sampling $O(1/\eps)$ columns of $A$. Importantly, the algorithm needs to make $O(n/\eps)$ queries to read the entirety of the sampled columns, and cannot give localized access to the approximate top eigenvector. Crucially, their global optimization maximizes the Rayleigh quotient over the subspace spanned by $O(1/\eps)$ columns, so it must read these columns explicitly. At the same time, however, it does achieve linear in $1/\eps$ query complexity for psd matrices, whereas Theorem~\ref{thm:main-lb-intro} shows that quadratic in $1/\eps^2$ is necessary; this implies a separation between psd matrices and (general) symmetric matrices.

\subsection{Applications to Graph Partitioning in Dense Graphs}

An important application of Theorem~\ref{thm:main-intro} is to approximate graph partitioning in the dense graph model, introduced in the seminal work of Goldreich, Goldwasser, and Ron~\cite{GGR98} in the context of property testing and approximation (see Chapter~8 of~\cite{G17}). Theorem~8.12 in Chapter 8 of~\cite{G17} gives an algorithm for testing general partition problems with $\poly(1/\eps)$ query complexities and $\exp(\poly(1/\eps))$ time complexity. In particular, the algorithms sample a few vertices of the graph, and after an initial exponential-in-$\poly(1/\eps)$ preprocessing time (which searches over all possible partitions of the sample), provide oracle access to an underlying cut in $\poly(1/\eps)$ time per query (as in local computation algorithms). A super-polynomial dependence on $1/\eps$ is necessary (assuming $\mathsf{P} \neq \mathsf{NP}$) because graph partitioning problems are often $\mathsf{NP}$-hard. However, what is possible and was open (until this work), is a Cheeger-like approximation guarantee in polynomial time. As we discuss in Theorem~\ref{thm:cheeger-intro} and Theorem~\ref{thm:trevisan} below, this will be a direct consequence of using Theorem~\ref{thm:main-intro} to give local versions of algorithms from spectral graph theory.

\begin{definition}[Sparsity of a Graph]\label{def:sparsity}
    For a fixed graph $G = ([n], E)$ and a subset $S \subset [n]$, we let
    \begin{itemize}
        \item $\vol_G(S) = \sum_{i \in S} \deg_G(i)$, where $\deg_G(i)$ is the degree of vertex $i$.
        \item $E(S, \ol{S}) = \sum_{(i,j) \in E} \ind\left\{ i \in S, j \in \ol{S} \right\}$ is the number of edges in $E$ which cross $S$.
        \item The sparsity of a cut is $\phi_G(S) = E(S, \ol{S}) / \min\{\vol_G(S), \vol_G(\ol{S}) \}$.
    \end{itemize}
    For a density parameter $\beta \geq 0$, the sparsity of $G$ over $\beta$-dense cuts is defined by 
    \[ \phi_{\beta}(G) = \min \{ \phi_G(S) : S \subset [n] \text{ s.t. }\min\{ \vol_G(S), \vol_G(\ol{S}) \} \geq \beta n^2 \}. \]
\end{definition}

The standard proof of Cheeger's inequality~\cite{AM85} gives a constructive spectral algorithm which performs a ``sweep cut'' over level sets of the second largest eigenvector of the normalized adjacency matrix. The following theorem applies Theorem~\ref{thm:main-intro}, to provide localized access to a sparse cut obtained from this sweep where all running times are polynomial in $1/\eps$. Note, the sparsity of the returned cut will compete with $\phi_{\eps}(G)$; one cannot approximate $\phi_0(G)$ (with no density constraints) because it is too sensitive, and a single disconnected vertex could make $\phi_0(G) = 0$ and be undetectable to a constant-query algorithm.

\begin{restatable}{theorem}{localcheeger}\label{thm:cheeger-intro}
    There exists a (randomized) local computation algorithm $\calA$ with query access to the adjacency matrix of a graph $G = ([n], E)$ and an accuracy parameter $\eps > 0$.
    \begin{itemize}
        \item For any $i \in [n]$, $\calA$ makes $\poly(1/\eps)$ queries and outputs access ``$i \in \bS$'' for an underlying cut $\bS \subset [n]$.
        \item Letting $\bS \subset [n]$ be the underlying cut that $\calA$ generates,
        \begin{align*}
            \Prx_{\calA}\left[ \phi_G(\bS) \leq O(\sqrt{\phi_{\eps}(G)}) + \eps  \right] \geq 0.99.
        \end{align*}
    \end{itemize}
    Furthermore, $\calA$ always runs in $\poly(1/\eps)$ time and uses $\poly(1/\eps)$ space.
\end{restatable}

We use a similar methodology to provide local access to an approximate max-cut (with analogous ``square-root-opt'' guarantees), by giving a local implementation of Trevisan's algorithm~\cite{L12}. As in Cheeger's inequality, Trevisan's algorithm performs a sweep cut over the level sets of the smallest eigenvector of the normalized adjancency matrix in order to find an approximately bipartite subgraph (and recurses on the remainder). As above, we will use Theorem~\ref{thm:main-intro} to find the top eigenvector to $-A$, and then show that we may (i) provide local access to the approximate bipartite subgraph, and (ii) recurse for at most $\poly(1/\eps)$ iterations on the remainder (which suffices for dense graphs).

\begin{restatable}{theorem}{localtrevisan}\label{thm:trevisan}
    There exists a (randomized) local computation algorithm $\calA$ with query access to the adjacency matrix of a graph $G = ([n], E)$ and accuracy parameter $\eps > 0$.
    \begin{itemize}
        \item For any $i \in [n]$, $\calA$ makes $\poly(1/\eps)$ queries and outputs access ``$i \in \bS$'' for an underlying cut $\bS \subset [n]$.
        \item If the max-cut of $G$ cuts $(1-\varphi)|E|$ edges, and $\bS \subset [n]$ is the underlying cut that $\calA$ generates,
        \[ \Prx_{\calA}\left[ E(\bS, \ol{\bS}) \geq (1 - \sqrt{\varphi}) |E| - \eps n^2 \right] \geq 0.99.\]
    \end{itemize}
    Furthermore, $\calA$ always runs in $\poly(1/\eps)$ time and uses $\poly(1/\eps)$ space.
\end{restatable}

\paragraph{Concurrent Work.} In concurrent and independent work, Bhattacharjee, Musco, and Rutkowski~\cite{BMR26} study sublinear time eigenvector approximation for
bounded-entry symmetric matrices and obtain results related to ours. The algorithms described are essentially the same (they too, sample a random principal submatrix and ``lift'' to an eigenvector but normalize slightly differently), with some qualitative differences:
\begin{itemize}
    \item The error guarantee in~\cite{BMR26} is in terms of the residual. The output vector $\bx$ satisfies $\| A \bx - \lambda \bx\|_2 \leq \eps n$, and they obtain such bounds for all eigenvalues. Residual error bounds are stronger than Rayleigh quotient bounds and imply our Rayleigh quotient bounds.
    \item When $\lambda_{\max}$ is the maximum magnitude eigenvector, their analysis matches our $\eps$-dependence but incurs a dependence on $n$: $\tilde{O}(\log^2 n / \eps^4)$ queries in preprocessing and $\tilde{O}(\log n /\eps^2)$ queries per entry. When $\lambda_{\max}$ is not the maximum magnitude eigenvector, the $\eps$-dependence degrades to $1/\eps^8$ on preprocessing and $1/\eps^4$ per entry (whereas our algorithms are always independent of $n$ and obtain a better $\eps$-dependence for $\lambda_{\max}$), but their algorithm computes all outlying eigenvectors.
\end{itemize}
Both papers prove a query complexity lower bound of $\Omega(n/\eps^2)$. We note that~\cite{BMR26} give other results in related models in addition to the bounded-entry model. 

\paragraph{Additional Related Work.} Local computation algorithms have also been applied to linear-algebraic
problems by David, Goldenberg, and Krauthgamer~\cite{DGK17}, who study local
reconstruction~\cite{SS10} of low-rank matrices. While the problem studied in that work and approximation guarantees different, the underlying techniques resemble those used in this work. As with our, there will also be an initial sample of random submatrix which is ``lifted'' to the entire matrix. 

\ignore{given $M \in \mathbb{F}^{n
\times m}$ promised to be $\eps$-close in normalized Hamming distance (denoted $\Delta$) to
a rank-$d$ matrix, they read $\tilde{O}(\sqrt{d/\varepsilon^3})$ entries in
preprocessing and then provide query access to the entries of a fixed
rank-$d$ matrix $M'$ with $\Delta(M, M') = O(\sqrt{d\varepsilon})$ reading $2d$ additional entries per query.  Their
algorithm proceeds using a similar pre-processing phase and a query phase: the algorithm first finds a random submatrix $M_{\tilde{R},\tilde{C}}$ which satisfies $M' = M_{[n],\tilde{C}} M_{\tilde{R},\tilde{C}}^{-1}M_{\tilde{R},[m]}$,
then the individual coordinates $M'_{i,j}$ can be recovered by reading $d$ entries from a the corresponding row and
column and calculating $M' = M_{i,\tilde{C}} M_{\tilde{R},\tilde{C}}^{-1}M_{\tilde{R},j}$.  Our setting differs in that we assume
bounded entries instead of a low rank promise, and in the notion of approximation: the
guarantee of~\cite{DGK17} is entrywise agreement in normalized Hamming distance, versus our additive $\eps n$ Rayleigh quotient guarantee.}

\subsection{Technical Overview}

\subsubsection{Upper Bound}

We give an overview of the algorithm and analysis behind Theorem~\ref{thm:main-intro}, where the goal is to design a local computation algorithm with query access to a bounded-entry matrix $A$ that can provide query access to entries of an approximately top eigenvector $\bx$, whose Rayleigh quotient with $A$ is at least $\lambda_{\max}(A) - \eps n$. The starting point is an algorithm of~\cite{SW25} who prove that, if $\bT$ is an $\bt \times n$ row-sampling matrix (i.e., each of the $n$ rows is included in $\bT A$ with probability $q$ and re-scaled by $1/\sqrt{q}$), then $\lambda_{\max}(\bT A \bT^{\intercal})$ and $\lambda_{\max}(A)$ differ by at most $\eps n$ with high probability whenever $q \gsim 1/(\eps^2 n)$ (up to poly-logarithmic in $1/\eps$ factors).\footnote{In the technical overview, we will ignore $\polylog(1/\eps)$-factors.} Hence, we may assume an algorithm has found the top (unit) eigenvector $\by \in \R^{\bt}$ satisfying
\[ \hat{\blambda}_{\max} = \dfrac{\langle \by, \bT A \bT^{\intercal} \by \rangle}{\| \by\|_2^2} \qquad \text{where} \qquad \left|\hat{\blambda}_{\max} - \lambda_{\max}(A) \right| \leq \eps n. \]
Note, however, that while $\by \in \R^{\bt}$, our task is to locally produce a vector $\bx \in \R^n$ where $\langle \bx, A \bx \rangle / \|\bx\|_2^2$ is at least $\lambda_{\max}(A) - O(\eps n)$; ideally, using the information learned from $\by$ and $\bT A \bT^{\intercal}$. 

\paragraph{The Lift.} We aim to convince the reader in the next few paragraphs that there is a natural choice of vector $\bx \in \R^{n}$, a ``lift'' of the vector $\by \in \R^{\bt}$ from the sketch-space to the original $n$-dimensional space, for which there is a local computation algorithm for individual entries, and one may hope satisfies
\begin{align} 
\left| \dfrac{\langle \by, \bT A \bT^{\intercal} \by \rangle}{\|\by\|_2^2} - \dfrac{\langle \bx, A \bx\rangle}{\|\bx\|_2^2}\right| \lsim \eps n. \label{eq:rayleigh-diff-int}
\end{align}
We consider the vector $\bx = A \bT^{\intercal} \by / \hat{\blambda}_{\max} \in \R^{n}$, which we call the \emph{lift} of $\by$. Note, producing the $i$-th entry of $\bx$ requires reading the $i$-th row of $A \bT^{\intercal}$ (which contains $\bt$ numbers), and we claim is a natural candidate for (\ref{eq:rayleigh-diff-int}) for the following reason. First, we set up the vector $\bx$ such that $\bT \bx \in \R^{\bt}$ is exactly $\bT A \bT^{\intercal} \by / \hat{\blambda}_{\max} = \by$ (since $\by$ is an eigenvector of $\bT A \bT^{\intercal}$ with eigenvalue $\hat{\blambda}_{\max}$). Re-expressing the inner products $\langle \bx, A \bx\rangle$ and $\langle \by, \bT A \bT^{\intercal} \by\rangle$ as the individual contributions from the eigenvectors of $A$, $u_1, \dots, u_n \in \R^n$ with eigenvalues $\lambda_1, \dots, \lambda_n \in \R$, it becomes natural to compare the ``numerators'' in (\ref{eq:rayleigh-diff-int})
\begin{align*}
\langle \by, \bT A \bT^{\intercal} \by \rangle = \sum_{i=1}^n \lambda_i \cdot \langle \bT u_i, \bT \bx\rangle^2 \qquad\text{with} \qquad \langle \bx, A \bx \rangle = \sum_{i=1}^n \lambda_i \cdot \langle u_i, \bx \rangle^2, 
\end{align*}
as well as the ``denominators'' $\| \bT \bx\|_2^2$ and $\|\bx\|_2^2$ in (\ref{eq:rayleigh-diff-int}). The intuition is that, if the vector $\bx \in \R^n$ lies mostly in the subspace of eigenvectors of $A$ with large magnitude eigenvalues, then, the subspace embedding condition of $\bT$ (used in~\cite{SW25}, that $\bT$ has low distortion on large magnitude eigenvectors) would apply to $\bx$, and one would conclude $\| \bT \bx \|_2^2 \approx \| \bx \|_2^2$, and that for large $|\lambda_i|$, $\langle \bT u_i, \bT \bx\rangle^2$ is close enough to $\langle u_i, \bx\rangle^2$. Roughly speaking, the differences in the ``numerator'' and ``denominator'' in the left-hand side of (\ref{eq:rayleigh-diff-int}) would suggest the upper bound. The reason this does not work, of course, is that the vector $\bx$ depends on $\by$ and $\bT$, and $\by$ also heavily depends on $\bT$, so one cannot a posteriori consider a vector $\bx$ where $\bT \bx = \by$ and apply the properties of $\bT$ (which are distributional in nature). 

\paragraph{Analyzing the Lift Vector $\bx$.} The analysis in Section~\ref{sec:analysis} proceeds via a two-step process. First, we show that, because $\by$ is an eigenvector of $\bT A \bT^{\intercal}$ with eigenvalue $\hat{\blambda}_{\max}$ which is positive and large, it must necessarily impose a coarse structure on $\bx$; then, we will use the coarse structure to obtain more refined estimates. Below, we give a high-level sketch for how this is used to compare $\| \bx \|_2^2$ and $\| \bT \bx\|_2^2$ (see proof of Theorem~\ref{thm:main} and Lemma~\ref{lem:2}). Writing $\bx = U\Lambda^{1/2} \tilde{\bw} / \hat{\blambda}_{\max}$ where $U$ is matrix whose columns are the eigenvectors of $A$, $\Lambda$ is the diagonal matrix of eigenvalues, and $\tilde{\bw} = \sign(\Lambda) \Lambda^{1/2} (\bT U)^{\intercal} \by$, we must upper-bound
\begin{align}
\left|\|\bx\|_2^2 - \|\bT \bx\|_2^2 \right| &= \frac{1}{\hat{\blambda}_{\max}^2} \left| \| U\Lambda^{1/2} \tilde{\bw} \|_2^2 - \|\bT U\Lambda^{1/2} \tilde{\bw}\|_2^2 \right| \nonumber \\
							&= \frac{1}{\hat{\blambda}_{\max}^2} \left| \langle \tilde{\bw}, \Lambda^{1/2} \left( U^{\intercal} U - (\bT U)^{\intercal} (\bT U) \right) \Lambda^{1/2} \tilde{\bw} \rangle \right|. \label{eq:inner-p-int}
\end{align}
Notice that, the inner-most matrix $\bM = \Lambda^{1/2} (U^{\intercal} U - (\bT U)^{\intercal} (\bT U))\Lambda^{1/2}$ is an $n\times n$ matrix where rows and columns correspond to eigenvectors, and which will give a weighted measure of the distortion of $\bT$ on the subspaces spanned by eigenvectors of $A$. Importantly, a spectral norm bound on $\bM$ would be very weak (since $\bT$ necessarily distorts some vectors in $\R^n$ significantly), however, it suffices for us to bound the inner-product (\ref{eq:inner-p-int}) which only cares about how $\bM$ acts on $\tilde{\bw}$. In a coarse analysis in Subsection~\ref{sec:structural-vecs}, we establish two facts: Lemma~\ref{lem:neg-small-pos-large} first shows that most of the squared-$\ell_2$-mass of $\tilde{\bw}$ is along the high-positive eigenvectors, and in Lemma~\ref{lem:small-bottom} that this estimate can be refined (using Lemma~\ref{lem:neg-small-pos-large}) so that most of the squared-$\ell_2$-mass is along the eigenvectors whose eigenvalue is $\Theta(\lambda_{\max}(A))$. Hence, the inner product in (\ref{eq:inner-p-int}) ``focuses-in" on the principal sub-matrix of $\bM$ corresponding to eigenvectors which are positive and have magnitude $\Omega(\lambda_{\max}(A))$; using the subspace embedding conditions of~\cite{SW25}, $\bT$ will have sufficiently small distortion on this subspace.

More generally, the proofs of Lemma~\ref{lem:2} and Lemma~\ref{lem:3} utilize this approach. On the one hand, the subspace embedding condition of~\cite{SW25} shows that $\bT$ acts as a low-distortion subspace embedding on subspaces spanned by high-magnitude eigenvectors (the higher the magnitude, the lower the distortion). This allows us to upper bound the spectral norm of sub-matrices of $\bM$. On the other hand, the distribution of squared-$\ell_2$-mass on $\tilde{\bw}$ will necessarily be biased toward the positive eigenvectors of large magnitude, so bounds on inner-products will ``focus-in'' on regions of $\bM$ which we can more tightly control.

\paragraph{Negative Eigenvalue Interference.} The complexity degradation from negative eigenvalues of large magnitude happens because the subspace embedding condition $\bT$ inevitably ``mixes'' some of large negative eigenspace $U_N$ and the large positive eigenspace $U_P$. In particular, even though every vector in $U_N$ and every vector in $U_P$ are orthogonal, $\bT U_N$ and $\bT U_P$ will not be. Hence, when the algorithm computes the top eigenvector of $\by$ of $\bT A \bT^{\intercal}$, $\by$ will aim to be aligned with the very top eigenvectors of $\bT U_P$; however, doing so inevitably also aligns $\by$ with $\bT U_N$ as well. When we lift $\bx = A \bT^{\intercal} \by = U \Lambda (\bT U)^{\intercal} \by$, the alignment between $\by$ and the very top of $\bT U_N$ is further multiplied by $-\lambda_{\min}(A)$. In Appendix~\ref{app:negative-interfere}, we give a concrete example of this phenomenon, and why an analysis relying solely on the subspace embedding condition of~\cite{SW25} cannot achieve query complexity $\tilde{O}(1/\eps^2)$ when $\lambda_{\min}(A)$ greatly exceeds $\lambda_{\max}(A)$.

\subsubsection{Lower Bound}

The lower bound proceeds via a reduction to the classic ``coin problem.'' For a uniformly random draw $\bb \sim \{-1,1\}$, an algorithm is provided $s$ independent draws $\bc_1, \dots, \bc_s$ from a distribution supported on $\{-1,1\}$ whose mean is $\eps \bb$. The algorithm only sees the draws $\bc_1,\dots, \bc_s$ and must learn $\bb$, and it is a standard fact this requires $s = \Omega(1/\eps^2)$ draws. Theorem~\ref{thm:main-lb-intro} will show that a $q$-query algorithm for the top eigenvector could be used to solve the coin problem with only $q/n$ draws, giving the desired lower bound.

By Yao's minimax principle, it suffices to assume that a (deterministic) $q$-query algorithm can approximate the top eigenvector of a random symmetric matrix $\bA \in \{-1,1\}^{n\times n}$ with ones on the diagonals and whose off-diagonal entries have expectation is $2\eps \bx \bx^{\intercal}$ where $\bx \sim \{-1,1\}^n$. Perhaps unsurprisingly, with high probability over the draw of $\bA$, any approximately top eigenvector $\by$ of $\bA$ must be heavily correlated with the hidden vector $\bx$ (up to a global sign). The key in our argument, which must boost the $\Omega(1/\eps^2)$ lower bound on the coin problem to an $\Omega(n/\eps^2)$ lower bound on the eigenvector problem, is to ``embed'' the coin problem into a single random coordinate of $\bx$. In particular, we consider the following:
\begin{itemize}
    \item We draw a target index $\bell \sim [n]$, and we generate a random vector $\bx \sim \{-1,1\}^n$ where every entry $\bx_{i} \sim \{-1,1\}$ for $i \neq \ell$ is known to the algorithm, and where the algorithm will pretend $\bx_{\ell} = \bb$ (even though it does not know $\bb$ and can only access the coins $\bc_1, \dots, \bc_s$). 
    \item We may now generate a matrix $\bA$ according to our distribution. For any entry $(i,j)$ not in the $\ell$-th row or column, we may explicitly generate $\bA_{i,j}$ using $\bx_i$ and $\bx_j$. For any off-diagonal entry in the $\ell$-th row or column, even though we do not know $\bx_{\ell} = \bb$, a single coin $\bc_{j}$ of expectation $\eps \bb$ may be used to determine $\bA_{\ell,j}$ (see Definition~\ref{def:coin-distribution}).
\end{itemize}
Importantly, the top-eigenvector algorithm is completely unaware of the index $\bell \sim [n]$ so among the $q$-queries to the matrix $O(q/n)$ will fall in the $\ell$-th row or column in expectation, and as a result, only $O(q/n)$ draws of coins are needed to generate the probes to the matrix $\bA$. Once the algorithm finds the approximately top eigenvector $\by$ of $\bA$, we know that $\by$ must correlate with $\bx$ (and hence fix a global sign such that most coordinates of $\by$ and $\bx$ must agree). Once more, the algorithm which produced $\by$ was unaware of $\bell$ and must have $\sign(\by_i) = \sign(\bx_i)$ for most $i \in [n]$; therefore, what we formalize in Section~\ref{sec:lb} is that this implies the value $\by_{\ell}$ will agree with $\bb$ with high constant probability (thereby giving the lower bound).

\textbf{Lower Bound via Direct-Sum Theorems for Information Complexity.} An anonymous reviewer of this work suggested that Theorem~\ref{thm:main-lb-intro} be proved via direct-sum theorems for information complexity (with respect to the uniform distribution~\cite{CKW12}). The direct-sum approach of ``embedding'' a problem within a collection of problems to boost the lower bound is exactly what happens here. We agree with the reviewer, but keep the current presentation since it is elementary and entirely self-contained.

\ignore{\subsubsection{Applications to Local Graph Partitioning}

The proof of Theorem~\ref{thm:cheeger} and Theorem~\ref{thm:trevisan} proceeds by utilizing Theorem~\ref{thm:main-intro} to give local versions of Cheeger's inequality~\cite{AM85}, using the second-largest eigenvector, and Trevisan's algorithm~\cite{L12}, using smallest eigenvector, of the normalized adjacency matrix $M$ of an undirected graph $G = ([n], E)$,
\[ M_{ij} = \dfrac{\ind\{(i,j) \in E\}}{\sqrt{\deg_G(i) \cdot \deg_{G}(j)}}, \]
which have all eigenvalues between $-1$ and $1$. For example, the proof of Cheeger's inequality gives an approximately sparsest cut via a ``coordinate-wise'' rounding procedure which takes a second largest eigenvector and decides which side of a cut a vertex is on by solely looking at the coordinate corresponding to that vertex. Our local algorithm for approximating the top eigenvector is well-suited for providing the needed coordinate-wise local access, as long as the ``scale'' of the error incurred in Theorem~\ref{thm:main} does not significantly degrade the guarantees of the cut. 

The simplest case to consider is when all degrees are $\Omega(n)$ and $m = \Omega(n^2)$. In this case, we use Theorem~\ref{thm:main} on a noisy version (where we use sampling to estimate degrees) of $M' = M - uu^{\intercal}/(2m)$, where $u \in \R^n$ is the top eigenvector and is given by $u_i = \sqrt{\deg_G(i)}$ for each $i \in [n]$.\footnote{Here, the top eigenvector of $M'$ is the second top eigenvector of $M$.} The matrix $M'$ has every entry $O(1/n)$ so, after re-scaling by $n$, the error of Theorem~\ref{thm:main} would incur an additive $\eps$ error on the Rayleigh quotient in $M'$, which becomes an additive $O(\eps)$ error on the sparsity.\footnote{One aspect which shows up in our analysis is that }

\begin{itemize}
    \item The normalized adjacency matrix $M$ depends on the degrees, $\deg_G(i)$, which we cannot hope to know exactly with $\poly(1/\eps)$ query complexity.
    \item Furthermore, the guarantees of Theorem~\ref{thm:main} are for obtaining an additive-$\eps n$ approximation to the Rayleigh quotient of a bounded-entry matrix.
\end{itemize}
In the case that all the degrees were $\Omega(n)$, then by sampling $\poly(1/\eps)$ entries, we could estimate the degree of all but $(1-\xi)$-fraction of vertices to $(1+\xi)$-factor, for any $\xi = \poly(\eps)$. In addition, the entries of $M$ would all be $O(1/n)$, so that after scaling by $\Theta(n)$, we'd get a bounded-entry matrix and additive-$\poly(\eps) n$ errors would correspond to the right ``scale'' for Rayleigh quotients with $M$. The last technical hurdle which arises is that one iteration of rounding the (second top or smallest) eigenvector results in cuts which may 

in the case of max-cut, Theorem~\ref{thm:main} would provide coordinate-wise access to the top eigenvector of $-M$ (which would correspond to the smallest eigenvector), and for sparsest cut, Theorem~\ref{thm:main} would provide coordinate-wise access to the top eigenvector of $M - uu^{\intercal}/ (2m)$, where $u = \sqrt{\deg_G} \in \R^n$ is the top eigenvector of $M$ (which removes the top eigenvector). In Section~\ref{sec:sparse-cut} and Section~\ref{sec:max-cut}, the above approach will proceed as follows:
\begin{itemize}
    \item First, the normalized adjacency matrix $M$ (as well as the deflated $M - uu^{\intercal} / (2m)$) depends on the degrees $\deg_G(i)$, which would be too expensive to exactly simulate. Instead, we approximate the number of edges $\hat{m}$, as well as the degrees, $d(i)$, via sampling $\poly(1/\eps)$; recall, in the dense graph model, we will obtain degree estimates up to $\poly(\eps) n$, and we may assume $\poly($
    \item Second, we ensure that the proof of Cheeger's inequality and Trevisan's algorithm is robust to the errors incurred from Theorem~\ref{thm:main}. 
\end{itemize}}


\section{Preliminaries}

\subsection{Local Computation Algorithms} 

We begin by describing the model of local computation algorithms adapted from the definition of~\cite{RTVX11}, and tailored to the problems which we will study in this paper. In all cases, the problem can be phrased as that of receiving query access to the entries of an $n\times n$ matrix (either the bounded entry matrix directly, or the adjancency matrix of an underlying graph), and the goal is to provide local access to the entries of a vector $\bx \in \R^n$ (which will either correspond to a top eigenvector, or encode a cut of the underlying vertices of the graph).

\begin{definition}[Local Computation Algorithm for Matrices]\label{def:lca}
A local computation algorithm $\calA$ is a (randomized) algorithm which has:
\begin{itemize}
    \item Query access to a large object, which in this work, will be the entries of a symmetric $n\times n$ matrix $A$ with $\|A\|_{\infty} \leq 1$,
    \item Access to a shared randomness tape $\brho \in \{0,1\}^{*}$, and
    \item A query index $i \in [n]$, an accuracy parameter $\eps \in (0, 1)$, and a failure probability $\delta \in (0,1)$.
\end{itemize}
Letting $\bx \in \R^n$ be the random vector where
\[ \bx_i = \text{ output of $\calA$ on input $i \in [n]$, $\eps,\delta \in (0,1)$, with access to $A$ and $\brho$,}\]
we will prove properties of the vector $\bx$ while ensuring that $\calA$ makes few queries to the underlying matrix $A$. 
\end{definition}
The query complexity $q(\eps,\delta)$ of the algorithm $\calA$ is the maximum number of entries of the matrix $A$ that $\calA$ may inspect on any particular input $i \in [n]$, and the time complexity $t(\eps,\delta)$ is the maximum amount of time $\calA$ needs to produce an output in the word RAM model, where we have dropped the parameter $n$ since our algorithms will have query and time complexity independent of $n$. Our algorithm $\calA$, on input $i \in [n]$, $\eps,\delta \in (0,1)$ will proceed by first executing a preprocessing phase on the matrix $A$ which does not depend on the index $i$, and is shared for all possible queries $i \in [n]$. Hence, it is conceptually useful to divide the algorithm $\calA$ into two phases, the preprocessing phase, and the query phase. As a result, we let $q_{p}(\eps,\delta), q_{q}(\eps,\delta)$ denote the query complexity of preprocessing and querying, respectively, and $t_{p}(\eps,\delta), t_q(\eps,\delta)$ for the time complexity accordingly. 

\subsection{Randomized Linear Algebra}

For fixed $n \in \N$ and a parameter $q$, a row-sampling matrix distribution $\calR(n,q)$ is given by the following procedure to generate $\bT \sim \calR(n,q)$: starting with an empty set of rows, each $i \in [n]$ is included independently with probability $q$, and if it is included, the row $e_i / \sqrt{q} \in \R^n$ is added as a row of $\bT$ (here, $e_i$ denotes the $i$-th basis vector). Notice that the resulting matrix $\bT$ is a $\bt \times n$ matrix, where $\bt$ is a binomial with parameters $n$ and $q$.

\begin{definition}[Subspace Embedding]
For any $\eps \in (0,1/2)$ and an $n \times d$ matrix $X$, we say that a $t \times n$ matrix $T$ is an $\eps$-distortion subspace embedding of $X$ if for all $v \in \R^d$,
\[ \left| \| X v \|_2^2 - \| T X v \|_2^2 \right| \leq \eps \| X v\|_2^2. \]
\end{definition}
\begin{definition}[Leverage Scores]\label{def:lev-scores}
For an $n\times d$ matrix $X$, the leverage score for row $i$ of $X$ is $\tau_i = e_i^{\intercal} X (X^{\intercal} X)^{\dagger} X^{\intercal} e_i$, where $(X^{\intercal} X)^{\dagger}$ denotes the Moore-Penrose pseudo-inverse of $X^{\intercal} X$, and we will make use of the following equivalent characterization,\footnote{Writing $X = U \Sigma V^{\intercal}$ as its singular value decomposition, $V$ is an orthonormal basis of $\R^d$ and $U$ an orthonormal basis of the column span of $X$. The Moore-Penrose pseudo-inverse of $X^{\intercal} X$ becomes $V \Sigma^{-2} V^{\intercal}$, and hence $\tau_i = \| U^{\intercal} e_i\|_2^2$. The fact that $U$ is an orthonormal basis of the column span of $X$ implies the maximum $(Xv)_i / \| Xv\|_2^2$ over $v \in \R^d$ is equivalent to maximizing $\langle e_i, Uz\rangle^2 / \| Uz\|_2^2$ over $z$, which is at most $\| U^{\intercal} e_i\|_2^2$ by Cauchy-Schwarz and also attains it by $z = U^{\intercal} e_i$.}
\[ \tau_i = \max_{v \in \R^d} \frac{(Xv)_i^2}{\| Xv\|_2^2}. \]
\end{definition}

The following theorem appears in~\cite{SW25} as Theorem 3.4 in the preliminaries, and follows from the matrix Chernoff bound~\cite{K18b, T12}. 

\begin{theorem}\label{thm:sampling-subspace}
Fix any $n\times d$ matrix $X$. For any $\eps,\delta \in (0,1)$ and $p \in (0,1)$ which satisfies $p \geq \min(2\tau_i \log(d/\delta) / \eps^2, 1)$ for all $i \in [n]$, a draw $\bT \sim \calR(n,p)$ is an $\eps$-distortion subspace embedding of $X$ with probability at least $1-\delta$.
\end{theorem}


\section{Local Computation Algorithm for the Top Eigenvector}

We receive as input query access to a symmetric $n \times n$ matrix $A$ with $\|A \|_{\infty} \leq 1$, and accuracy parameter $\eps \in (0, 1)$, as well as an index $i \in [n]$. The algorithm uses a public source of randomness (which is shared for all queries $i \in [n]$), makes queries to the matrix $A$, and will output a number $\bx_i \in \R$. Using the public source of randomness, the algorithm samples (for some parameter $q \in (0,1)$ which we specify later), a matrix $\bT \sim \calR(n,q)$ of size $\bt \times n$.

The algorithm is divided into a preprocessing phase which is completely independent of $i$, and the query phase, which produces the output coordinate $\bx_i$ for the value of $i$. In a local implementation of the algorithm, one may execute the preprocessing phase prior to each query with the same source of randomness.

\begin{figure}[h]
\begin{framed}
\textbf{Preprocessing Phase:} In the preprocessing phase independent of the input $i \in [n]$, the algorithm proceeds by:
\begin{enumerate}
\item Query the $\bt^2$ entries of $\bA_1 = \bT A \bT^{\intercal}$.
\item\label{ln:3} Find the top eigenvector $\by \in \R^{\bt}$ of $\bA_1$ with eigenvalue $\hat{\blambda}_{\max}$, and store $\by$ and $\hat{\blambda}_{\max}$.
\end{enumerate}
\end{framed}
\caption{Preprocessing Algorithm}\label{fig:preprocess}
\end{figure}

Notice that the query complexity of the preprocessing phase is at most $\bt^2$ which is at most $O((nq)^2)$ with high constant probability. Furthermore, the running time of the preprocessing phase consists of the time needed to compute a top eigenvector $\by$ which can be done in $\poly(\bt)$ time up to high precision in the Word-RAM model.\footnote{We will ignore bit-complexity dependencies involves in rounding eigenvectors and eigenvalues.} Having stored the vector $\by \in \R^{\bt}$, the (full) output vector is $\bx = A\bT^{\intercal} \by / \hat{\blambda}_{\max}$, which means that on query $i \in [n]$, it suffices for the algorithm to know the $i$-th row of $A \bT^{\intercal}$.

\begin{figure}[h]
\begin{framed}
\textbf{Query Phase:} On the query for index $i \in [n]$, the algorithm will:
\begin{enumerate}
\item Query the $\bt$ entries at the $i$-th row of $A$ for the columns selected by $\bT$, to compute $(A\bT^{\intercal})_i$. 
\item Output $\bx_i = \langle (A\bT^{\intercal})_i, \by\rangle / \hat{\blambda}_{\max}$.
\end{enumerate}
\end{framed}
\caption{Query Algorithm}\label{fig:query}
\end{figure}

The above completes the description of the algorithm, and the remainder of the section is devoted to proving the following theorem, which gives the parameter setting and guarantees for approximating the top eigenvector. In what follows, we refer to $\lambda_{\max}(A)$ as the maximum eigenvalue and $\lambda_{\min}(A)$ as the magnitude of the minimum eigenvalue (which is smaller than $0$, or $0$ if the matrix is psd). Since we seek additive $\eps n$-approximations, we will assume throughout that $\lambda_{\max}(A) \geq c \cdot \eps n$ for a large constant factor $c$ (see Remark~\ref{app:lambda-max-large}). Furthermore, even though the parameter settings below depend on $\lambda_{\max}(A)$ and $\lambda_{\min}(A)$, it suffices to know an additive $\eps n$ approximation to these, which one can obtain without any asymptotic overhead using the eigenvalue estimation algorithm of~\cite{SW25}.
\begin{theorem}\label{thm:main}
For any $n \times n$ symmetric matrix $A$ with $\|A \|_{\infty} \leq 1$ and any $\eps > 0$, let $\bx \in \R^n$ be the output of the algorithm on $A$, initialized with
\[ q = \max\left\{ \frac{1}{\eps^2 n}, \frac{n^{1/3} \cdot \lambda_{\min}(A)^{4/3}}{\eps^{2/3} \cdot \lambda_{\max}(A)^{8/3}} \right\} \cdot \polylog(1/\eps), \]
Then,
\begin{align*}
\Prx\left[ \max_{x_0 \in \R^n} \dfrac{\langle x_0, Ax_0\rangle}{\|x_0\|_2^2} - \dfrac{\langle \bx, A \bx \rangle}{\|\bx\|_2^2} \leq \eps n\right] \geq  0.99 
\end{align*}
\end{theorem}

We note that any setting of $\lambda_{\min}(A) = O(\lambda_{\max})$ results in a bound on $q$ which is at most $1/(\eps^2 n) \cdot \polylog(1/\eps)$ since $\lambda_{\max}(A) \geq \eps n$. Furthermore, the worst-case occurs when $\lambda_{\max}(A) = \Theta(\eps n)$ and $\lambda_{\min}(A) = \Theta(n)$, in which case the above bound on $q$ becomes $1/(\eps^{3.\ol{3}} n) \cdot \polylog(1/\eps)$. It is an interesting problem to handle this ``negative-eigenvalue interference'' without any complexity degradation, since this would require new ideas. We refer the reader to Appendix~\ref{app:negative-interfere}, where we provide a concrete example where our analysis does incur some ``negative-eigenvalue interference.''

\begin{remark}[Assuming $\lambda_{\max}(A) \geq c \cdot \eps n$ for a large $c > 0$]\label{app:lambda-max-large} First, suppose $\lambda_{\max}(A)$ is positive, and let $c_0 \in \R_{\geq 0}$ such that $\lambda_{\max}(A) = c_0 \cdot \eps n$. If $c_0$ is larger than an absolute constant (say, $10^{-10}$), then we consider executing the algorithm with $\eps_0 = (c_0 / c) \eps$, and hence we are in the situation where $\lambda_{\max}(A) = c_0 \cdot \eps n \geq c \cdot \eps_0 n$. Otherwise, in order to meet our desired guarantee, it suffices to output query access to a vector $\bx$ which satisfies $\langle \bx, A \bx\rangle / \|\bx\|_2^2 \geq - (\eps/2) n$. We let $\bx$ be a uniform random vector from the unit sphere, and by rotational invariance, we may assume that our Rayleigh quotient is $\sum_{i=1}^n \lambda_i \bx_i^2$. Note, the expected Rayleigh quotient is exactly $\Tr(A) / n \in [-1, 1]$. For the variance, we use the fact that the fourth moment $\Ex[\bx_i^4]$, as well as $\Ex[\bx_i^2 \bx_j^2]$ for a uniform random unit vector $\bx$ are both $O(1/n^2)$. As a result, the expected square of the Rayleigh quotient is at most $\sum_{i=1}^n \lambda_i^2 + \sum_{i \neq j} \lambda_{i} \lambda_j \leq n^2 + \Tr(A)^2$ divided by $n^2$, giving us a constant second moment. By Chebyshev's inequality, $\bx$ is very likely to have Rayleigh quotient which is at least $-O(1) \geq -\eps n$ with high constant probability.
\end{remark}

In our applications to the sparsest cut and max cut, it will become important to ``boost'' the error probability to $1-\delta$. Even though Theorem~\ref{thm:main} is stated for success probability $0.99$, a benefit of the fact that our algorithm outputs query access to an eigenvector means that we can verify a lower bound on $\langle \bx, A \bx \rangle / \|\bx\|_2^2$ via sampling. This allows us to ``boost'' our error probability to $1-\delta$ via multiple independent executions of Theorem~\ref{thm:main} and verifying the outputs via sampling. Specifically, let $k = O(\log(1/\delta)$ and let $\bx^{(1)},\dots \bx^{(k)}$ be the result of $k$ independent calls to Theorem \ref{thm:main}. By independence, the probability that $i\in [k]$ outputs have Rayleigh quotient with $A$ less than $\lambda_{\max}(A) - \eps n$ is at most $0.01^k \leq \delta$.  Furthermore, we can estimate the numerator and denominator of the Rayleigh quotient up to an additive $\eps n$ and $\eps$ error respectively, using $\poly(\log(1/\delta)/\eps)$ queries to $\bx$ and $A$, since we know $\|A\|_{\infty} \leq 1$ and $\|\bx\|_{\infty}^2 \leq 1/(\eps^2 n)$.

\section{Analysis of the Algorithm}\label{sec:analysis}

\subsection{Theorem~\ref{thm:main} assuming Lemma~\ref{lem:1},~\ref{lem:2} and~\ref{lem:3}}

For the proof of Theorem~\ref{thm:main}, we consider a fixed $n \times n$ symmetric matrix $A$ with $\|A\|_{\infty} \leq 1$, as well as a fixed error parameter $\eps > 0$. We consider the following deterministic events, which we will show occur with high constant probability, and we will show these ensure the conclusion of Theorem~\ref{thm:main}. In particular, the events depend on parameters $\alpha \in (0, 1)$ which should be thought about being $1/\polylog(1/\eps)$, and $L \geq 0$ (which we set shortly) and should be thought of as being smaller than the allowable error (i.e., $o(\eps n)$). Both $n \in \N$ and $\eps > 0$ should be viewed as asymptotic parameters with $n \to \infty$ and $\eps \to 0$.
\begin{definition}[The Events]\label{def:events}
For $\alpha \in (0, 1)$ and $0 \leq L = o(\eps n)$, let $\calbE$ denote the indicator random variable that the following events (over the randomness in the draw of $\bT$) hold. 
\begin{itemize}
\item The matrix $\bT$ is an $\alpha L / \lambda$-distortion subspace embeddings for the span of eigenvectors of $A$ with eigenvalue of magnitude at least $\lambda$, for all $\lambda \geq L$.
\item The matrix $\bT$ satisfies $|\lambda_{\max}(\bT A \bT^{\intercal}) - \lambda_{\max}(A)| \leq \alpha L$ and $|\lambda_{\min}(\bT A \bT^{\intercal}) - \lambda_{\min}(A)| \leq \alpha L$.
\item Letting $A_{L}$ denote the ``portion'' of $A$ with eigenvalues of magnitude at most $L$ (see Definition~\ref{def:high-low-decomp} in Subsection~\ref{sec:decomp}), there is a fixed constant $c_0 > 0$ where $\| A_L \bT^{\intercal} \|_2$ and $\| \bT A_{L} \bT^{\intercal} \|_2$ are at most $c_0 L$.\footnote{We should note that both prior works~\cite{BDDMR24,SW25} utilize an argument of~\cite{T12} to upper bound $\| \bT A_L \bT^{\intercal}\|_2$ and lose a $O(\log n)$-factor. This logarithmic factor propagates throughout the analysis, even though~\cite{SW25} later manage to remove it via a recursive analysis. It turns out, a closer look at~\cite{T12} (as well as the precursor~\cite{RV07}) reveals that the dependence is $O(\log(nq))$ which removes the logarithmic factor in our application of $q = 1/(\poly(\eps) n)$. This seemed to have been missed in the prior work, and we give the argument in Appendix~\ref{app:spectral-norm-decay}.}
\end{itemize}
\end{definition}
For the remainder of the analysis, we will assume that the conditions of $\calbE$ hold deterministically for a parameter $L$ which we will precisely set later (see Lemma~\ref{lem:1} below). It will be useful to introduce the parameter $\gamma \geq 0$ which will measure our amount of ``negative-eigenvalue interference,'' and which depends on $\lambda_{\max}(A), \lambda_{\min}(A)$ as well as $L$; we set $\gamma$ to 
\begin{align}
\gamma \eqdef \dfrac{\sqrt{L \lambda_{\min}(A)}}{\lambda_{\max}(A)} .\label{eq:gamma-bound}
\end{align}
In our analysis, we will need to set the parameter $L$ such that the following two lemmas (Lemma~\ref{lem:2} and Lemma~\ref{lem:3}) give bounds which are at most $\eps n / \lambda_{\max}$ and $\eps n$.

\begin{lemma}\label{lem:2}
Whenever $\calbE$ holds,
\begin{align*}
\left| \| \bx \|_2^2 - \| \by \|_2^2\right| \lsim \frac{L \cdot (1 + \gamma^2)}{\lambda_{\max}(A)} \cdot \polylog(1/\eps).
\end{align*} 
\end{lemma}

\begin{lemma}\label{lem:3}
Whenever $\calbE$ holds,
\begin{align*}
\left| \langle \by, \bA_1 \by\rangle - \langle \bx, A \bx \rangle\right| \lsim  L \cdot (1 + \gamma^4) \cdot \polylog(1/\eps).
\end{align*}
\end{lemma}

Given that the above two lemmas should be $\eps n /\lambda_{\max}$ and $\eps n$, respectively, we set $L$ according to
\begin{align}
L = \min\left\{ \eps n,\left( \eps n \cdot \left(\frac{\lambda_{\max}(A)^2}{\lambda_{\min}(A)}\right)^{2}\right)^{1/3} \right\} \cdot \dfrac{1}{\polylog(1/\eps)}. \label{eq:L-setting}
\end{align}
Note, the above setting is such that Lemma~\ref{lem:2} is at most $\eps n / \lambda_{\max}(A)$ and Lemma~\ref{lem:3} is at most $\eps n$. As we will see (in the next lemma), the setting of $L$ is directly tied to the query complexity, and smaller values of $L$ lead to larger query complexities. From (\ref{eq:L-setting}), $L$ is set to $\eps n / \polylog(n)$ whenever $\lambda_{\min}(A)$ is not too much larger than $\lambda_{\max}(A)$---in general, when $\lambda_{\min}(A)$ is much larger than $\lambda_{\max}(A)$, $L$ will be a ``geometric mean'' of $\eps n$ and $\lambda_{\max}(A)^2/\lambda_{\min}(A)$ (twice), which leads to a worse query complexity whenever $\lambda_{\max}(A)^2 / \lambda_{\min}(A)$ is smaller than $\eps n$.

\begin{lemma}\label{lem:1}
Suppose we let $q$ such that
\[ q = \frac{n}{L^2} \cdot \polylog(n/L). \] 
the event $\calbE$ holds with probability $0.99$.
\end{lemma}

We note that the above setting of $q$ leads to a query complexity of $(nq)^2$ for the preprocessing stage and $nq$ for the query phase. This evaluates to a preprocessing complexity of $1/\eps^4 \cdot \polylog(1/\eps)$ and query complexity of $1 / \eps^2 \cdot \polylog(1/\eps)$ whenever $\lambda_{\min}(A) = O(\lambda_{\max}(A))$ as desired. Note, an extremal case for the worst-possible query complexity above happens when $\lambda_{\min}(A)$ is as large as possibe and $\lambda_{\max}(A)$ as small as possible (since this makes the geometric mean in $L$ as small as possible, and leads to higher query complexity). For example, when $\lambda_{\min}(A) = \Theta(n)$ and $\lambda_{\max}(A) = \Theta(\eps n)$, $L$ gets set to $\eps^{5/3} n \cdot \polylog(1/\eps)$, and leads to preprocessing complexity $1/\eps^{6.\ol{6}}$ and query complexity $1/\eps^{3.\ol{3}}$ up to poly-logarithmic factors.

\begin{proof}[Proof of Theorem~\ref{thm:main} assuming Lemmas~\ref{lem:1}~\ref{lem:2},~\ref{lem:3}]
We consider the setting of $L$ in~(\ref{eq:L-setting}) so as to assume the bound of Lemma~\ref{lem:2} is $\eps n / \lambda_{\max}$ and the bound of Lemma~\ref{lem:3} is $\eps n$, and $q = n / L^2 \cdot \polylog(n/L)$ results in a bound of:
\[ q = \max\left\{ \frac{1}{\eps^2 n}, \frac{n^{1/3} \lambda_{\min}(A)^{4/3}}{\eps^{2/3} \lambda_{\max}(A)^{8/3}} \right\} \cdot \polylog(1/\eps), \]
and assume that Lemma~\ref{lem:1} holds, so we may assume that event $\calbE$ holds. We note that Lemmas~\ref{lem:2},~\ref{lem:3} allows us to compare the Rayleigh quotients of $\by$ in $\bA_1$ and $\bx$ in $A$. Namely, we upper bound (using the fact $\|\by\|_2 = 1$),
\begin{align*}
\left| \dfrac{\langle \by, \bA_1 \by\rangle}{\|\by\|_2^2} - \dfrac{\langle \bx, A \bx\rangle}{\|\bx\|_2^2}\right| &\leq \left| \langle \by, \bA_1 \by\rangle - \langle \bx, A \bx\rangle \right| + \left| \frac{\langle \bx, A \bx \rangle}{\| \bx \|_2^2} \right| \left| \|\bx\|_2^2 - \| \by\|_2^2  \right|\\
		&\lsim \eps n + \left|\dfrac{\langle \bx, A \bx\rangle}{\| \bx\|_2^2}\right| \cdot \frac{\eps n}{\lambda_{\max}(A)} \lsim \eps n + \left|\frac{\langle \by, \bA_1 \by\rangle}{\|\by\|_2^2} - \frac{\langle \bx, A \bx\rangle}{\|\bx\|_2^2} \right| \cdot \frac{\eps n}{\lambda_{\max}(A)},
\end{align*}
where the final inequality added and subtracted $\langle \by, \bA_1 \by \rangle$, and we know $|\langle \by, \bA_1 \by \rangle| = O(\lambda_{\max}(A))$, from the third condition of Definition~\ref{def:events}. Re-arranging the above terms, we obtain
\[ \left| \dfrac{\langle \by, \bA_1 \by \rangle}{\| \by\|_2^2} - \dfrac{\langle \bx, A \bx\rangle}{\|\bx\|_2^2} \right| \cdot \left(1 - O(\eps n/\lambda_{\max}(A)) \right) \lsim \eps n \qquad \Longrightarrow \qquad \left| \dfrac{\langle \by, \bA_1 \by \rangle}{\| \by\|_2^2} - \dfrac{\langle \bx, A \bx\rangle}{\|\bx\|_2^2} \right| \lsim \eps n, \]
since $\eps$ can be set so that $\eps n$ is a small enough constant factor of $\lambda_{\max}(A)$. The fact that $\by$ is an eigenvector of $\bA_1$ which obtains eigenvalue at least $\lambda_{\max}(A) - \alpha L$, by the third condition of Definition~\ref{def:events}, implies
\begin{align*}
\max_{x_0 \in \R^n} \dfrac{\langle x_0, A x_0 \rangle}{\|x_0\|_2^2} - \dfrac{\langle \bx, A \bx\rangle}{\|\bx\|_2^2} \leq \alpha L + \dfrac{\langle \by, \bA_1 \by \rangle}{\| \by\|_2^2} - \dfrac{\langle \bx, A \bx \rangle}{\| \bx\|_2^2} \lsim \alpha L + \eps n.
\end{align*}
Notice that, the above argument concludes that the difference in Rayleigh quotients is $O(\eps n)$, and the claimed $\eps n$ bound can be done by instantiating $\eps$ to be a fixed constant factor smaller.
\end{proof}

It remains to prove Lemmas~\ref{lem:1},~\ref{lem:2} and~\ref{lem:3}, which we do in Subsections~\ref{sec:proof-lem-2},~\ref{sec:proof-lem:3} below. In order to prove these, we introduce the following structural properties and auxiliary definitions.

\subsection{Proof of Lemma~\ref{lem:1}}

The first bullet point of Definition~\ref{def:events} will follow directly from~\cite{SW25}; we re-create their proof (making minor modifications to parameter settings). For any $\lambda \geq L$, let $U_{\geq \lambda}$ denote the orthonormal matrix whose columns are eigenvectors of $A$ with eigenvalue having magnitude at least $\lambda$. 
\begin{claim}\label{cl:1}
The number of columns in $U_{\geq \lambda}$ is at most $n^2/L^2$, and for every $i \in [n]$, $\| (U_{\geq \lambda})_i \|_2^2 \leq n / \lambda^2$.
\end{claim}

The first part of Claim~\ref{cl:1} follows from the fact $\| A \|_F^2 \leq n^2$ is also the sum of squares of all eigenvalues. Hence, the number of eigenvalues which are larger than $\lambda \geq L$ is at most $n^2/\lambda^2 \leq n^2/L^2$. For the second part, we use Lemma~3.3 of~\cite{SW25}, the squared Euclidean norm of the rows is bounded by $n/\lambda^2$, i.e., for any $i \in [n]$, $\| (U_{\geq \lambda})_i\|_2^2 \leq n/\lambda^2$. 
\begin{claim}\label{cl:2}
For any $i \in [n]$, the $i$-th leverage score of $U_{\geq \lambda}$ (recall, Definition~\ref{def:lev-scores}) is at most
\[ \max_{v \in \R^d} \dfrac{\langle (U_{\geq \lambda})_i, v\rangle^2}{\| U_{\geq \lambda} v\|_2^2} \leq \max_{v \in \R^d} \dfrac{\| (U_{\geq \lambda})_i \|_2^2 \cdot \| v\|_2^2}{\| v\|_2^2} \leq \frac{n}{\lambda^2}, \]
where the numerator bound follows from Cauchy-Schwarz and the denominator from the fact $U_{\geq \lambda}$ is orthonormal.
\end{claim}
 By Theorem~\ref{thm:sampling-subspace} with $\eps = \alpha L / \lambda$, a setting of 
 \[ q \gsim (n/L^2) \log(n/(L\delta)) / \alpha^2 \geq \min\left\{ \tau_i \log(|U_{\geq \lambda}|/\delta') / (\alpha L)^2 \cdot \lambda^2, 1  \right\}, \]
suffices to ensure that for any fixed $\lambda \geq L$, a draw of $\bT$ gives a $\alpha L / \lambda$-distortion subspace embedding with probability $1-\delta$. Considering the collection of $\lambda$ of the form $L \cdot 2^r$ for $r = 0, \dots, \lceil \log_2 n/L \rceil$, and setting $\delta = 1/ (1000\log_2(n/L))$ suffices to union bound over all $\lambda$ as well as $\bT$ and obtain success probability $1-1 /1000$ for the first bullet point. This completes the proof that the first bullet point occurs with high constant probability.

For the second bullet point, we directly apply Theorem 2.2 of~\cite{SW25} (reproduced below) with an error parameter $\eps'$ so that $\eps' n$ is at most $\alpha L$ (and boosting the success probability from the claimed $2/3$ to $1 - 1 / 1000$). Note, this occurs so long as $q$ is at least some constant factor of $1/(n \cdot (\eps')^2) \log^2(1/\eps') \geq (n/L^2) \log^2(n/(\alpha L)) /\alpha^2$, and this completes the third bullet point.
\begin{theorem}[Theorem~2.2 of~\cite{SW25}]
For any symmetric $n\times n$ matrix $A$ with $\|A\|_{\infty} \leq 1$, any $\eps \in (0,1)$, and any $q \gsim \log^2(1/\eps) / (\eps^2 n)$, the eigenvalues of $\bT A \bT^{\intercal}$ (and additional $0$'s) give additive $\pm \eps n$ approximation to the eigenvalues of $A$ with probability at least $2/3$.
\end{theorem} 

For the fourth bullet point, we apply Theorem~\ref{thm:rv} and Theorem~\ref{thm:tropp}, as well as Markov's inequality. Namely,
\begin{align*}
\Ex_{\bT}\left[ \| A_L \bT^{\intercal} \|_{2}\right] &\lsim \| A_L \|_2 + \sqrt{\log(nq)/q} \cdot \|A_{L} \|_{1\to2}, \\
\Ex_{\bT}\left[ \| \bT  A_L \bT^{\intercal} \|_2 \right] &\lsim \|A_L\|_2 + \sqrt{\log(nq)/q} \cdot \|A_L \|_{1\to2} + \log(nq)/q \cdot \|A_{L}\|_{\infty},
\end{align*}
and we claim both bounds are $O(L)$ with appropriate choice of $q$ by bounding each of the above terms.
\begin{claim}
The value $\| A_L \|_{1\to2}$, i.e., maximum column norm of $A_L$, is at most $\sqrt{n}$.
\end{claim} 
This follows from the fact columns of $A_L$ are projections of columns of $A$ onto span of eigenvectors with eigenvalue of magnitude at most $L$. This must only decrease the norm, so the maximum column norm of $A_L$ is at most that the maximum column norm of $A$, which is at most $\sqrt{n}$ since $\|A \|_{\infty} \leq 1$.
\begin{claim}
The value $\|A_{L} \|_{\infty} \leq 1 + n/L$.
\end{claim}
The argument appears in Lemma~4 of~\cite{BDDMR24}, using the fact $\|A_L\|_{\infty} \leq \|A \|_{\infty} + \|A_H \|_{\infty}$ (here, $A_H = A - A_L = U_{\geq L} \Lambda_{\geq L} U_{\geq L}^{\intercal}$ is the ``portion'' of $A$ whose eigenvectors have eigenvalue of magnitude at least $L$---see Definition~\ref{def:high-low-decomp}). It remains to bound $\|A_H \|_{\infty}$, so consider any fixed $i,j\in [n]$, and notice
\begin{align*}
|(A_H)_{i,j}| = \langle (U_{\geq L})_i, (U_{\geq L} \Lambda_{\geq L})_j \rangle &= \langle (U_{\geq L})_i, (A U_{\geq L})_j \rangle \\
	&\leq \| (U_{\geq L})_i \|_2 \cdot \| (A U_{\geq L})_j \|_2 \leq \sqrt{n} / L \cdot \|A_j\|_2 = n/L,
\end{align*}
where the first inequality is Cauchy-Schwarz, and the second uses Claim~\ref{cl:1} and the fact $U_{\geq L}$ is orthonormal. Ensuring that $q \geq (n / L^2) \cdot \log(nq)$ guarantees both expectations are at most $O(L)$, which allows us to apply Markov's inequality.

\subsection{Spectral Decompositions}\label{sec:decomp}

\newcommand{\bSigma}{\boldsymbol{\Sigma}}
\newcommand{\bLambda}{\boldsymbol{\Lambda}}
\newcommand{\Sbb}{\mathbb{S}}

We first define a few useful decompositions of the matrix $A$ which we will use throughout the analysis. These involve decomposing the spectrum of $A$ into parts. 

\begin{definition}[Spectral Decomposition]\label{def:spectral-decomp}
We express $A = U \Lambda U^{\intercal}$ as its singular value decomposition.
\begin{itemize}
\item $U$ is a matrix whose columns form the eigenvectors of $A$, and 
\item $\Lambda$ is an $n\times n$ diagonal matrix of eigenvalues. 
\end{itemize}
Let $\Lambda^{1/2}$ is the $n \times n$ diagonal matrix whose entries are $\sqrt{|\lambda_i|}$ , and $\sign(\Lambda)$ is the $n\times n$ diagonal matrix whose entries are $\sign(\lambda_i) \in \{-1,1\}$ for each of the eigenvalues $\lambda_i$ of $A$.
\end{definition}

In what follows, we define some decompositions of $A$. These will always involve selecting a subset of the eigenvectors and corresponding eigenvalues of $A$. 

\begin{definition}[Positive-Negative-Low Decomposition]\label{def:high-low-decomp}
We let $A = A_{P} - A_{N} + A_{L}$, where $A_{P} = U_{P} \Lambda_{P} U_{P}^{\intercal}$ and $A_{N} = U_N \Lambda_{N} U_{N}^{\intercal}$ and $A_L = U_L \Lambda_{L} U_{L}^{\intercal}$, and let $A_H = A_P - A_N$.
\begin{itemize}
\item The matrix $U_{P}$ has columns given by the subset of eigenvectors from $U$ whose corresponding eigenvalue is positive and larger than $L$.
\item The matrix $U_N$ has columns given by the subset of eigenvectors from $U$ whose corresponding eigenvalue is negative and smaller than $-L$, and let $\Lambda_{N}$ be the diagonal matrix of magnitudes of negative eigenvalues.
\item The matrix $U_{L}$ has columns given by the subset of eigenvectors from $U$ whose corresponding eigenvalue has magnitude at most $L$. 
\end{itemize}
We let $\lambda_{\max} = \lambda_{\max}(A)$, and $r = \lceil \log(\lambda_{\max} / L)\rceil$,\footnote{In the parameter settings, $\lambda_{\max}$ is at most $n$ and $L$ will be on the order of $\poly(\eps) n$; hence, one should consider $r$ as being logarithmic in $1/\eps$, with no dependence on $n$.} and write
\[ A_P = \sum_{h=0}^{r} U_{h}^{+} \Lambda_h^{+} (U_h^+)^{\intercal}, \]
where the matrix $U_h^+$ contains columns of $U_P$ whose corresponding eigenvalue has magnitude between $\lambda_{\max} / 2^h$ and $\lambda_{\max} / 2^{h+1}$. The matrix $U_{-0}^+$ contains columns for all eigenvectors of $U_P$ which are not in $U_0^+$, and $\Lambda_{-0}^+$ is the corresponding diagonal matrix of eigenvalues.

Similarly, we let $\lambda_{\min} = \lambda_{\min}(A_N)$, and $k = \lceil \log(\lambda_{\min}/L)\rceil$, and write 
\[ A_N = \sum_{\ell=0}^k U_{\ell}^- \Lambda_{\ell}^- ( U_{\ell}^- )^{\intercal}, \]
where the matrix $U_{\ell}^-$ contains columns of $U_N$ whose corresponding eigenvalue has magnitude between $\lambda_{\min} / 2^{\ell}$ and $\lambda_{\min} / 2^{\ell+1}$, and $\Lambda_{\ell}^-$ is the corresponding diagonal matrix of eigenvalues.
\end{definition}

The following estimates will be used throughout the paper, and are particularly important for Lemma~\ref{lem:small-bottom}. We prove them here by applying the conditions of $\calbE$. It is convenient in the lemma below to think of the easier case of $\lambda_{\min} = O(\lambda_{\max})$, as this would simplify some of the expressions below.

\begin{lemma}\label{lem:spectral-norm-bounds}
There exists a fixed constant $c_2 > 0$ such that whenever $\calbE$ holds
\begin{itemize}
\item $\| (\Lambda_0^+)^{1/2} (((U_0^+)^{\intercal} U_{0}^+) - (\bT U_0^+)^{\intercal} (\bT U_0^+)) (\Lambda_0^+)^{1/2} \|_2 \leq c_2 L$.
\item $\| (\Lambda_{-0}^+)^{1/2} (((U_{-0}^+)^{\intercal} U_{-0}^+) - (\bT U_{-0}^+)^{\intercal} (\bT U_{-0}^+)) (\Lambda_{-0}^+)^{1/2} \|_2 \leq c_2 \sqrt{r\lambda_{\max}L}$.
\item $\| \Lambda_N^{1/2} ((U_N^{\intercal} U_N) - (\bT U_N)^{\intercal} (\bT U_N)) \Lambda_N^{1/2} \|_2 \leq c_2 \sqrt{kL \lambda_{\min}}$.
\item $\| (\Lambda_0^+)^{1/2} (\bT U_0^+)^{\intercal} (\bT U_{-0}^+) (\Lambda_{-0}^+)^{1/2} \|_2 \leq c_2 \sqrt{\lambda_{\max} L}$.
\item $\| (\Lambda_0^+)^{1/2} (\bT U_0^+)^{\intercal} (\bT U_{N}) \Lambda_{N}^{1/2} \|_2 \leq c_2 \sqrt{\lambda_{\max} L} \left(1 + \gamma \right)$.
\item $\| (\Lambda_{-0}^+)^{1/2} (\bT U_{-0}^+)^{\intercal} (\bT U_{N}) \Lambda_{N}^{1/2} \|_2 \leq c_2 \left(\sqrt{\lambda_{\max} L} + \sqrt{\lambda_{\min} L}\right)$.
\end{itemize}
\end{lemma}

\begin{proof}
The first bound follows from the fact $\bT$ gives an $\alpha L / \lambda_{\max}$-distortion subspace embedding of $U_{0}^+$. The subsequent bounds follow the same structure as~\cite{SW25}. For the second bullet point, the spectral norm is upper bounded by the spectral norm of the following $r \times r$ matrix $M$, where the $(h_1, h_2)$ entry (when $h_1 \neq h_2$) is given by
\[ M_{h_1, h_2} = \| (\Lambda_{h_1}^+)^{1/2} (\bT U_{h_1}^+)^{\intercal} (\bT U_{h_2}^+) (\Lambda_{h_2}^+)^{1/2}  \|_2 \leq \frac{2 \cdot \lambda_{\max}}{2^{h_1/2} \cdot 2^{h_2/2}} \cdot \frac{\alpha L \cdot \max\{ 2^{h_1}, 2^{h_2}\}}{\lambda_{\max}} \lsim L \cdot \sqrt{2^{|h_1 - h_2|}},\]
since $\bT$ is a $\alpha L\cdot \max\{2^{h_1}, 2^{h_2}\} / \lambda_{\max}$-distortion subspace embedding of $[U_{h_1}^+|U_{h_2}^+]$. On the other hand, the diagonal entries for $h$ are 
\[ M_{h,h} = \| (\Lambda_{h}^+)^{1/2} \left( (U_{h}^+)^{\intercal} (U_h^+) - (\bT U_{h_1}^+)^{\intercal} (\bT U_{h_2}^+)\right) (\Lambda_{h_2}^+)^{1/2}  \|_2  \leq \frac{2 \lambda_{\max}}{2^{h}} \cdot \frac{\alpha L \cdot 2^h}{\lambda_{\max}} \lsim L.\]
The squared Frobenius norm of $M$ (which upper bounds the square of the spectral norm) is at most, up to constant factor, $L^2 \sum_{h_1=1}^r \sum_{h_2=1}^r 2^{|h_1 - h_2|} \leq rL  \lambda_{\max}$. The analogous argument implies the third spectral norm bound. For the fourth bullet, consider the $1 \times r$ matrix $M$ where $M_{1, h}$ is given by
\[ M_{1,h} = \| (\Lambda_{0}^{+})^{1/2} (\bT U_{0}^+)^{\intercal} (\bT U_{h}^+) (\Lambda_{h}^+)^{1/2} \|_2 \leq \frac{2\lambda_{\max}}{2^{h/2}} \cdot \frac{\alpha L \cdot 2^h}{\lambda_{\max}} \lsim L 2^{h/2},   \]
so that $\| M \|_F^2 \lsim L^2 \sum_{h=1}^r 2^h \lsim L \lambda_{\max}$. Similarly, the fifth bullet sets the $1 \times k$ matrix $M$ where $M_{1,\ell}$ is given by $\| (\Lambda_{0}^+)^{1/2} (\bT U_0^+)^{\intercal} (\bT U_{\ell}^-) (\Lambda_{\ell}^-)^{1/2} \|_2$, which is at most $\sqrt{\lambda_{\max} \lambda_{\min}} / 2^{\ell/2} \cdot \alpha L \left( 1/\lambda_{\max} + 2^{\ell} / \lambda_{\min}\right)$, so that $\| M \|_F^2 \lsim L^2 (\lambda_{\min} / \lambda_{\max}) + L \lambda_{\max}$; we obtain the desired bound by taking the square-root and plugging in the value of $\gamma$ from (\ref{eq:gamma-bound}). For the final spectral norm bound, the analogous argument bounds the squared Frobenius norm of the $r \times k$ matrix $M$ where 
\begin{align*} 
M_{h, \ell} = \| (\Lambda_{h}^+)^{1/2} (\bT U_{h}^+)^{\intercal} (\bT U_{\ell}^-) (\Lambda_{\ell}^-)^{1/2}\|_2 &\lsim \frac{\sqrt{\lambda_{\max} \lambda_{\min}}}{2^{h/2} \cdot 2^{\ell/2}} \cdot \alpha L \left( \frac{2^h}{\lambda_{\max}} + \frac{2^{\ell}}{\lambda_{\min}}\right) \\
	&\lsim L \left( \sqrt{\frac{\lambda_{\min}}{\lambda_{\max}}} \cdot \frac{2^{h/2}}{2^{\ell/2}} + \sqrt{\frac{\lambda_{\max}}{\lambda_{\min}}} \cdot \frac{2^{\ell/2}}{2^{h/2}} \right),
\end{align*}
so the squared Frobenius norm is at most, up to constant factors, $L \lambda_{\min} + L \lambda_{\max}$.
\end{proof}

\subsection{The Lift}\label{sec:analytic-lift}

We now define the key (random) vectors which form part of our analysis, and which allow us to go from the top eigenvector $\by \in \Sbb^t$ computed in the ``subsampled'' or ``sketched'' space, and lift it to the vector $\bx \in \R^n$ the algorithm outputs. We define a vector $\bx \in \R^n$ which we call the ``lift'' of $\by$. Note, randomness is with respect to the draws of $\bT$.

\begin{definition}[The Lift]\label{def:lift}
As specified in Figure~\ref{fig:preprocess}, $\bA_1 = \bT A \bT^{\intercal}$ and $\by \in \Sbb^{\bt}$ is the top eigenvector of $\bA_1$ with eigenvalue $\lambda_{\max}(\bT A \bT^{\intercal}) = \hat{\blambda}_{\max}$. The lift $\bx \in \R^n$ is given by
\[ \bx = \frac{1}{\hat{\blambda}_{\max}} \cdot A \bT^{\intercal} \by,\]
where we defined $\bx$ such that $\bT \bx = \bT A \bT^{\intercal} \by / \hat{\blambda}_{\max} = \by$. It will be helpful in the analysis to define the following vectors:
\begin{itemize}
\item $\tilde{\bw} \in \R^n$ is given by $\tilde{\bw} = \sign(\Lambda) \Lambda^{1/2} (\bT U)^{\intercal} \by$, and note that $\bT U \Lambda^{1/2} \tilde{\bw} = \bA_1 \by = \hat{\blambda}_{\max}\by$. 
\item $\tilde{\bx} \in \R^n$ is the vector given by $\tilde{\bx} = U \Lambda^{1/2} \tilde{\bw}$, and again, note $\bT \tilde{\bx} = \hat{\blambda}_{\max} \by$.
\item Let $\bx = \tilde{\bx} / \hat{\blambda}_{\max}$, which implies $\bT \bx = \by$.
\end{itemize}
Note that coordinates of $\tilde{\bw} \in \R^n$ correspond to eigenvectors of $U$, so we may decompose the vector $\tilde{\bw}$ following the decompositions of Subsection~\ref{sec:decomp}. In particular,
\begin{itemize}
\item $\tilde{\bw} = \tilde{\bw}_{H} + \tilde{\bw}_L$, where $\tilde{\bw}_{H} = \sign(\Lambda_H) \Lambda_{H}^{1/2} (\bT U_{H})^{\intercal} \by$, and $\tilde{\bw}_{L} = \sign(\Lambda_L) \Lambda_{L}^{1/2} (\bT U_{L})^{\intercal} \by$.\footnote{As written, there is a minor type-check issue, since $\tilde{\bw}_{H}$ and $\tilde{\bw}_{L}$ will not be in $\R^n$ (the sum of their dimensionalities will be $n$), since they are non-zero in a disjoint set of coordinates; the sum is interpreted as placing the non-zero entries in their corresponding positions.}
\item $\tilde{\bw}_H = \tilde{\bw}_P + \tilde{\bw}_{N}$, where $\tilde{\bw}_{P} = \Lambda_{P}^{1/2} (\bT U_{P})^{\intercal} \by$, $\tilde{\bw}_{N} = - \Lambda_{N}^{1/2} (\bT U_N)^{\intercal} \by$.
\item $\tilde{\bw}_P = \sum_{h=0}^{r} \tilde{\bw}_h^+$ where $\tilde{\bw}_h^+ = (\Lambda_{h}^+)^{1/2} (\bT U_{h}^+)^{\intercal} \by$.
\item $\tilde{\bw}_N = \sum_{\ell=0}^{k} \tilde{\bw}_{\ell}^{-}$ where $\tilde{\bw}_{\ell}^- = -(\Lambda_{\ell}^-)^{1/2} (\bT U_{\ell}^-)^{\intercal} \by$.
\end{itemize}
We define the normalization $\bw_H = \tilde{\bw}_H / \| \tilde{\bw}_H\|_2$, and analogously define the decompositions $\bw_H = \bw_P + \bw_N$ and $\bw_P = \sum_{h=0}^r \bw_h^+$ and $\bw_N = \sum_{\ell=0}^k \bw_{\ell}^-$, and $\bw_{-0} = \sum_{h=1}^r \bw_{h}^+$.
\end{definition}

\subsubsection{Coarse Structure of the Lift}\label{sec:structural-vecs}

The subsequent lemmas show that, the top eigenvector $\by \in \Sbb^{\bt}$ found by our algorithm in Line~\ref{ln:3} of Figure~\ref{fig:preprocess} will necessarily impose coarse structure on the lift $\bx$ and $\bw$. This coarse structure will be necessary when relating the Rayleigh quotients of the analytic lift in $A$ with $\hat{\blambda}_{\max}$ (which is the Rayleigh quotient of $\by$ in the ``sketch'' matrix $\bT A \bT^{\intercal}$). The first lemma below shows that the lift $\bx$ must have significantly more ``$\ell_2$-mass'' along eigenvectors of $U_P$, as opposed to eigenvectors of $U_N$. The second lemma is a more refined statement, that the lift $\bx$ must not only have most of the ``$\ell_2$-mass'' along eigenvectors of $U_P$, but also most of the ``$\ell_2$-mass'' along eigenvectors in $U_{0}^+$.

\begin{lemma}\label{lem:neg-small-pos-large}
There is a fixed constant $c_1 > 0$, such that whenever $\calbE$ holds, 
\begin{itemize}
\item $\lambda_{\max} - c_1 L \leq \| \tilde{\bw}_{P} \|_2^2 \leq \lambda_{\max} + c_1 Lr$
\item $\|\tilde{\bw}_{N}\|_2^2 \leq c_1 Lr$.
\item $\| \tilde{\bw}_H\|_2^2 \leq \lambda_{\max} + c_1 Lr$.
\end{itemize}
\end{lemma}

\begin{proof}
The fact $\by \in \Sbb^{\bt}$ is an eigenvector of $\bT A \bT^{\intercal}$ with eigenvalue $\hat{\blambda}_{\max}$ implies that
\begin{align}
\hat{\blambda}_{\max} = \langle \by, \bT A \bT^{\intercal} \by \rangle &\leq \langle \by, \bT A_P \bT^{\intercal} \by \rangle - \langle \by, \bT A_N \bT^{\intercal} \by \rangle + \| \bT A_L \bT^{\intercal} \|_2 \nonumber \\
&\leq \left\| \Lambda^{1/2}_{P} (\bT U_P)^{\intercal} \by \right\|_2^2 - \left\| \Lambda^{1/2}_{N} (\bT U_{N})^{\intercal} \by \right\|_2^2 + c_0L = \| \tilde{\bw}_{P} \|_2^2 - \| \tilde{\bw}_N \|_2^2 + c_0L. \label{eq:ter}
\end{align}
Note that this directly gives the lower bound for the first inequality, since the second condition of $\calbE$ in Definition~\ref{def:events} implies $\hat{\blambda}_{\max} \geq \lambda_{\max} - \alpha L$; we obtain the desired bound so long as $c_1 \geq c_0 + \alpha$. Furthermore, $\| \tilde{\bw}_P \|_2^2 \leq \lambda_{\max}(\bT A_P \bT^{\intercal})$, and the latter is at most $\lambda_{\max} + c_0 Lr$ by directly applying Lemma 5.6 of~\cite{SW25}. This establishes the upper bound of the first item when $c_1 \geq c_0$. Re-arranging terms in (\ref{eq:ter}), we obtain the third inequality whenever $c_1 \geq 2c_0 + \alpha$. The final inequality is a direct consequence of $\| \tilde{\bw}_H \|_2^2 = \| \tilde{\bw}_P \|_2^2 + \| \tilde{\bw}_{N} \|_2^2$, whenever $c_1 \geq 3c_0 + \alpha$. 
\end{proof}

\begin{lemma}\label{lem:small-bottom}
There exists a fixed constant $c_3 > 0$ such that whenever $\calbE$ holds, $\|\bw_{-0}^+\|_2^2 \leq c_3Lr / \lambda_{\max} (1 +\gamma)^2$.
\end{lemma}

\begin{proof}
First, we note that the vector $\tilde{\bw}_H = \sign(\Lambda_H) \Lambda_H^{1/2} (\bT U_H)^{\intercal} \by$, which means that the following sequence of (in)equalities is satisfied:
\begin{align}
\langle \bw_P,  \Lambda_{P}^{1/2} (\bT U_P)^{\intercal} \bT U_H \Lambda_H^{1/2} \tilde{\bw}_H \rangle &\geq  \langle \bw_P, \Lambda_{P}^{1/2} (\bT U_{P})^{\intercal} \bT U \Lambda^{1/2} \tilde{\bw} \rangle - \left\| \Lambda_P^{1/2} (\bT U_P)^{\intercal} \bT U_L \Lambda_L^{1/2} \tilde{\bw}_L \right\|_2 \nonumber \\
&= \hat{\blambda}_{\max} \langle \bw_P,  \Lambda_P^{1/2} (\bT U_P)^{\intercal} \by  \rangle - \left\|\Lambda_{P}^{1/2} (\bT U_P)^{\intercal} \right\|_2 \cdot \left\|\bT A_L \bT^{\intercal} \by \right\|_2 \nonumber \\
&\geq \hat{\blambda}_{\max} \cdot \frac{\| \tilde{\bw}_P \|_2^2}{\| \tilde{\bw}_H \|_2} - 2c_0 L\sqrt{\lambda_{\max}}, \label{eq:growth}
\end{align}
where we used $\bT U \Lambda^{1/2} \tilde{\bw} = \hat{\blambda}_{\max} \by$, as well as the fact $\| \Lambda_P^{1/2} U_P^{\intercal} \bT^{\intercal} \|_2 \leq \sqrt{\lambda_{\max}} \| \bT U_P \|_2 \leq 2 \sqrt{\lambda_{\max}}$, since $\bT$ is a $\alpha$-distortion subspace embedding of $U_{P}$. Then, Lemma~\ref{lem:neg-small-pos-large} implies that the right-hand side of (\ref{eq:growth}), since $\sqrt{\lambda_{\max}} \leq 2 \| \tilde{\bw}_P\|_2$ because $c_1 L = o(\eps n) \leq \lambda_{\max}/2$, is at least
\begin{align*}
\hat{\blambda}_{\max} \cdot \frac{\| \tilde{\bw}_P \|_2^2}{\| \tilde{\bw}_H \|_2} - 4c_0L \cdot \| \tilde{\bw}_P \|_2 &\geq \| \tilde{\bw}_H \|_2 \left(\hat{\blambda}_{\max} (1 - \|\tilde{\bw}_N\|_2^2 / \|\tilde{\bw}_H\|_2^2) - 4c_0 L\right) \\
	&\geq  \| \tilde{\bw}_H \|_2 \left( \lambda_{\max} - cL r \right),
\end{align*}
for some constant $c$, since $\|\tilde{\bw}_N \|_2^2 / \| \tilde{\bw}_H \|_2^2 \leq O(Lr/\lambda)$. Normalizing (since $\bw_H$ is a unit vector), we obtain the inequality:
\begin{align}
\lambda_{\max} - c L r \leq \langle \bw_P, \Lambda_P^{1/2} (\bT U_P)^{\intercal} \bT U_H \Lambda_H^{1/2} \bw_H \rangle. \label{eq:desired-eq-1}
\end{align}
On the other hand, we may also upper bound the right-hand side above by, 
\begin{align*}
\langle \bw_P,  \Lambda_P^{1/2} (\bT U_P)^{\intercal} \bT U_H \Lambda_H^{1/2} \bw_H \rangle &\leq \| \Lambda_{P}^{1/2} \bw_P \|_2^2 \\
	&\qquad +  \| (\Lambda_{0}^+)^{1/2} (( U_{0}^+)^{\intercal} U_0^+ - (\bT U_{0}^+)^{\intercal} (\bT U_{0}^+)) (\Lambda_{0}^+)^{1/2} \|_2 \| \bw_{0}^+\|_2^2\\
	&\qquad + \| (\Lambda_{-0}^+)^{1/2} (( U_{-0}^+)^{\intercal} U_{-0}^+ - (\bT U_{-0}^+)^{\intercal} (\bT U_{-0}^+)) (\Lambda_{-0}^+)^{1/2} \|_2 \| \bw_{-0}^+\|_2^2 \\
	&\qquad + 2\| (\Lambda_0^+)^{1/2} (\bT U_0^+)^{\intercal} (\bT U_{-0}^+) (\Lambda_{-0}^+)^{1/2} \|_2 \cdot \| \bw_{0}^+\|_2 \| \bw_{-0}^+\|_2 \\
	&\qquad + 2\| (\Lambda_0^+)^{1/2} (\bT U_0^+)^{\intercal} (\bT U_{N}) \Lambda_{N}^{1/2} \|_2 \cdot \| \bw_{0}^+\|_2 \| \bw_{N}\|_2 \\
	&\qquad + 2\| (\Lambda_{-0}^+)^{1/2} (\bT U_{-0}^+)^{\intercal} (\bT U_{N}) \Lambda_{N}^{1/2} \|_2 \cdot \| \bw_{-0}^+\|_2 \| \bw_{N}\|_2,
\end{align*}
which we upper bound from Lemma~\ref{lem:spectral-norm-bounds} by,
\begin{align*}
& \left( \lambda_{\max} + c_2 L\right)\cdot \|\bw_0^+\|_2^2 + \left( \frac{\lambda_{\max}}{2} + c_2 \sqrt{r\lambda_{\max} L}\right) \cdot \| \bw_{-0}^+ \|_2^2 + 2c_2 \sqrt{\lambda_{\max} L}\| \bw_0^+\|_2 \| \bw_{-0}^+\|_2 \\
&\qquad + 2 c_2 \sqrt{\lambda_{\max} L} \left( 1 + \gamma \right)\cdot \| \bw_{0}^+\|_2 \| \bw_N \|_2 + 2 c_2 \left( \sqrt{\lambda_{\max} L} + \sqrt{\lambda_{\min} L}\right) \| \bw_{-0}^+\|_2 \|\bw_N \|_2.
\end{align*}
and using that $\|\bw_{0}^+\|_2^2 \leq 1 - \| \bw_{-0}^+\|_2^2$, and that $\|\bw_N \|_2 \leq  \sqrt{2c_1 Lr / \lambda_{\max}}$, the above expression simplifies to the following:
\[ \lambda_{\max} + c' \cdot L \sqrt{r} (1+\gamma)  - \left( \frac{\lambda_{\max}}{2} - c_2 \sqrt{r\lambda_{\max} L} \right) \| \bw_{-0}^+\|_2^2 + c' \sqrt{\lambda_{\max} Lr} (1+\gamma) \cdot \| \bw_{-0}^+ \|_2, \]
for some constant $c > 0$. Combined with (\ref{eq:desired-eq-1}) for the lower bound, we obtain:
\begin{align*}
\frac{\lambda_{\max}}{3} \cdot \| \bw_{-0}^+\|_2^2 \leq (c+c') Lr (1+\gamma) + c' \sqrt{\lambda_{\max} L r}\cdot (1+\gamma)\cdot \| \bw_{-0}^+\|_2,
\end{align*}
which implies $\|\bw_{-0}^+\|_2^2 \leq c_3 Lr / \lambda_{\max} (1+\gamma)^2$, for some constant $c_3 > 0$, which gives the desired bound plugging in $\gamma$ from (\ref{eq:gamma-bound}).
\end{proof}


\subsection{Proof of Lemma~\ref{lem:2}}\label{sec:proof-lem-2}

The proof of the lemma will proceed from the following two intermediate claims.

\begin{claim}\label{cl:intermediate}
Whenever $\calbE$ from Definition~\ref{def:events} holds, 
\begin{align*}
\left| \| \bx \|_2^2 - \| \by \|_2^2 \right| \lsim \frac{1}{\lambda_{\max}} \left| \left\| U_H \Lambda_H^{1/2} \bw_H \right\|_2^2 - \left\| \bT U_H \Lambda_H^{1/2} \bw_H \right\|_2^2 \right| + O\left(\frac{L}{\lambda_{\max}}\right).
\end{align*}
\end{claim}

\begin{proof} 
First, $\bT \bx = \by$ and use Definition~\ref{def:lift} to re-write in terms of $\tilde{\bw}$ and $\hat{\blambda}_{\max}$ as 
\begin{align}
\left| \| \bx \|_2^2 - \| \by \|_2^2 \right| = \left| \| \bx \|_2^2 - \| \bT \bx \|_2^2 \right| &= \frac{1}{\hat{\blambda}_{\max}^2} \left| \| U \Lambda^{1/2} \tilde{\bw} \|_2^2 - \| \bT U \Lambda^{1/2} \tilde{\bw} \|_2^2 \right| \label{eq:ha2}.
\end{align}
We decompose the right-hand side of (\ref{eq:ha2}) in terms of $U_H$ and $U_L$. The term $U \Lambda^{1/2} \tilde{\bw}$ is simple to handle, since $U_H$ and $U_L$ are orthogonal. Namely, $\| U\Lambda^{1/2} \tilde{\bw}\|_2^2$ is exactly $\| U_H \Lambda_H^{1/2} \tilde{\bw}_H \|_2^2$ plus $\| U_L \Lambda_L^{1/2} \tilde{\bw}_L\|_2^2$. Furthermore, the definition of $\tilde{\bw}_L$ from Definition~\ref{def:lift} implies $\| U_L \Lambda_L^{1/2} \tilde{\bw}_L\|_2 = \| A_L \bT^{\intercal} \by\|_2$ which is at most $O(L)$ whenever $\calbE$ holds. On the other hand, the term $\bT U \Lambda^{1/2} \tilde{\bw}$ cannot proceed via the same argument since $\bT U_H$ and $\bT U_L$ are not orthogonal; thus, we resort to the triangle inequality:
\begin{align*}
\left| \| \bT U \Lambda^{1/2} \tilde{\bw} \|_2 - \| \bT U_H \Lambda^{1/2}_H \tilde{\bw}_H \|_2 \right| \leq \| \bT U_L \Lambda_L^{1/2} \tilde{\bw}_L \|_2 = \| \bT A_L \bT^{\intercal} \by \|_2 \leq O(L)
\end{align*}
whenever $\calbE$ holds. Furthermore, event $\calbE$ holding implies $\| \bT U_H \Lambda_H^{1/2} \tilde{\bw}_H \|_2 \leq \| \bT A \bT^{\intercal} \by \|_2 + \| \bT A_L \bT^{\intercal} \|_2$ which is at most $\hat{\blambda}_{\max} + O(L)$. In summary, the fact $\hat{\blambda}_{\max} = \Theta(\lambda_{\max})$, we obtain a bound on the right-hand side of (\ref{eq:ha2}), 
\begin{align*}
\frac{1}{\lambda_{\max}^2} \left| \| U_H \Lambda^{1/2}_H \tilde{\bw}_H \|_2^2 - \| \bT U_H \Lambda_H^{1/2} \tilde{\bw}_H \|_2^2 \right| + O\left( L^2 / \lambda_{\max}^2 \right) + O\left( L / \lambda_{\max}\right),
\end{align*}
and implies the desired bound, since $\bw_H = \tilde{\bw}_H / \| \tilde{\bw}_H\|_2$ and $\| \tilde{\bw}_H \|_2^2 = O(\lambda_{\max})$.
\end{proof}

\begin{lemma}\label{lem:H-part}
There exists a constant $c_3 > 0$ such that, whenever $\calbE$ holds,
\begin{align*}
\left| \| U_H \Lambda_H^{1/2} \bw_H \|_2^2 - \| \bT U_H \Lambda_H^{1/2} \bw_H \|_2^2 \right| \leq c_3 Lrk \left( 1 + \gamma +\gamma^2\right).
\end{align*}
\end{lemma}

\begin{proof}
We expand the above squared norms as inner products, and divide $\bw_H$ into three vectors $\bw_0^+$, $\bw_{-0}^+$ and $\bw_N$ and apply the bounds from Lemma~\ref{lem:spectral-norm-bounds}, Lemma~\ref{lem:neg-small-pos-large}, and Lemma~\ref{lem:small-bottom}. The resulting expression consists of the following the following terms:
\begin{itemize}
\item The term $\| \bw_{0}^+ \|_2^2 \leq 1$ times $\| (\Lambda_0^+)^{1/2} \left( (U_0^+)^{\intercal} (U_0^+) - (\bT U_0^+)^{\intercal} (\bT U_0^+) \right) (\Lambda_0^+)^{1/2} \|_2 \leq O(L)$.
\item The term $\| \bw_{-0}^+ \|_2^2 \lsim Lr / \lambda_{\max} (1+\gamma)^2$ (from Lemma~\ref{lem:small-bottom}) is multiplied by the spectral norm bound $\| (\Lambda_{-0}^+)^{1/2} \left( (U_{-0}^+)^{\intercal} (U_{-0}^+) - (\bT U_{-0}^+)^{\intercal} (\bT U_{-0}^+)\right) (\Lambda_{-0}^+)^{1/2} \|_2 = O(\sqrt{r\lambda_{\max} L})$. This results in the upper bound of $O((Lr)^{3/2} / \sqrt{\lambda_{\max}} (1+\gamma)^2)$.
\item The term $\| \bw_N \|_2^2 \lsim Lr / \lambda_{\max}$ (from Lemma~\ref{lem:neg-small-pos-large}) times $\| \Lambda_{N}^{1/2}\left(U_N^{\intercal} U_N - (\bT U_N)^{\intercal} (\bT U_N) \right) \Lambda_N^{1/2} \|_2 = O(\sqrt{k\lambda_{\min} L})$. This evaluates to $O(Lr\sqrt{k} \cdot \gamma)$ by the setting of $\gamma$ in (\ref{eq:gamma-bound}).
\item The term $2\| \bw_{0}^+ \|_2 \| \bw_{-0}^+\|_2$ which is at most $O(Lr / \lambda_{\max})^{1/2}(1+\gamma)$ (once, more from Lemma~\ref{lem:small-bottom}) times $\| (\Lambda_{0}^+)^{1/2} (\bT U_0^+)^{\intercal} (\bT U_{-0}^+) (\Lambda_{-0}^+)^{1/2}\|_2 = O(\sqrt{\lambda_{\max} L})$ (from Lemma~\ref{lem:spectral-norm-bounds}), which contributes $O(L\sqrt{r} \cdot (1+\gamma))$.
\item The term $2\| \bw_{0}^+ \|_2 \| \bw_{N}\|_2$ which is at most $O(Lr / \lambda_{\max})^{1/2}$ (from Lemma~\ref{lem:neg-small-pos-large}) times the spectral norm bound $\| (\Lambda_{0}^+)^{1/2} (\bT U_0^+)^{\intercal} (\bT U_N) \Lambda_N^{1/2} \|_2$, which is at most $O(\sqrt{\lambda_{\max} L} (1 + \gamma))$ from Lemma~\ref{lem:spectral-norm-bounds}. This contributes $O(L \sqrt{r}\cdot (1+\gamma))$.
\item The final term $2 \| \bw_{-0}^+\|_2 \| \bw_N\|_2$ is at most $O(Lr / \lambda_{\max}) (1+\gamma)$ (using both Lemma~\ref{lem:neg-small-pos-large} and Lemma~\ref{lem:small-bottom} once more) which is multiplied by $\| (\Lambda_{-0}^+)^{1/2} (\bT U_{-0}^+)^{\intercal} (\bT U_{N}) \Lambda_N^{1/2} \|_2$, which is at most $O(\sqrt{\lambda_{\max} L} + \sqrt{\lambda_{\min} L})$, which is at most $O(Lr(1+\gamma))$.
\end{itemize}
Summing these up, we can upper bound all terms by $Lrk (1 + \gamma + \gamma^2)$.
\end{proof}
Putting Claim~\ref{cl:intermediate} and Lemma~\ref{lem:H-part} together, we obtain a bound of
\begin{align*}
\left| \|\bx\|_2^2 - \|\by\|_2^2 \right| &\lsim \frac{1}{\lambda_{\max}} \left| \left\| U_H \Lambda_H^{1/2} \bw_{H} \right\|_2^2 - \left\| \bT U_H \Lambda_H^{1/2} \bw_H \right\|_2^2 \right| + \frac{L}{\lambda_{\max}} \\
			&\lsim \frac{Lrk}{\lambda_{\max}} \left( 1 + \gamma + \gamma^2\right) + \frac{L}{\lambda_{\max}}.
\end{align*}

\subsection{Proof of Lemma~\ref{lem:3}}\label{sec:proof-lem:3}

Similarly to the case of Lemma~\ref{lem:2}, we aim to upper bound:
\begin{align}
\left| \langle \by, \bA_1 \by \rangle - \langle \bx, A \bx \rangle \right| &\leq \left| \left\| \Lambda^{1/2}_P (\bT U_P)^{\intercal} \by \right\|_2^2 - \left\| \Lambda^{1/2}_P U_P^{\intercal} \bx \right\|_2^2 \right| + \left\| \Lambda_N^{1/2} U_N^{\intercal} \bx \right\|_2^2 + \| \Lambda_{N}^{1/2} (\bT U_N)^{\intercal} \by \|_2^2 \label{eq:rayleigh-diff} \\
    &\qquad \qquad + \| \bT A_L \bT^{\intercal}\|_2 + \| A_L \|_2 \cdot \|\bx\|_2^2 \nonumber 
\end{align}
where we drop the final term $\| \Lambda_N^{1/2} (\bT U_N)^{\intercal} \by \|_2^2$ as it is exactly $\| \tilde{\bw}_N \|_2^2 \leq c_1 Lr$, and we drop $\|\bT A_L \bT^{\intercal}\|_2$ and $\|A_L \|_2 \cdot \|\bx\|_2^2$, since these are at most $O(L)$. We prove upper bounds on the remaining two terms above in the next lemmas.
\begin{lemma}\label{lem:p-bound}
Whenever $\calbE$ holds
\begin{align*}
\left\| \Lambda_P^{1/2} (\bT U_P)^{\intercal} \by\right\|_2^2 - \left\| \Lambda_P^{1/2} U_P^{\intercal} \bx \right\|_2^2 = O\left(Lr (1+\gamma + \gamma^2) \right).
\end{align*}
\end{lemma}

\begin{proof}
Substitute $\by = \bT \bx$ where $\bx = U \Lambda^{1/2} \tilde{\bw} / \hat{\blambda}_{\max}$. Then, by the triangle inequality, 
\begin{align*}
\left\| \Lambda_P^{1/2} (\bT U_P)^{\intercal} \bT \bx \right\|_2^2 &\leq \frac{1}{\hat{\blambda}_{\max}^2} \left( \| \Lambda_{P}^{1/2} (\bT U_P)^{\intercal} \bT U_H \Lambda_H^{1/2} \tilde{\bw}_H \|_2 + \| \Lambda_{P}^{1/2} (\bT U_P)^{\intercal} \bT U_L \Lambda_L^{1/2} \tilde{\bw}_L \|_2\right)^2,
\end{align*}
and we can simplify the right-most term above using $\| \bT U_L \Lambda_L^{1/2} \tilde{\bw}_L \|_2 = \| \bT A_L \bT^{\intercal} \by \|_2 = O(L)$ to obtain the bound
\begin{align*}
	&\frac{1}{\hat{\blambda}_{\max}^2} \left(\| \Lambda_P^{1/2} (\bT U_P)^{\intercal} \bT U_H \Lambda_H^{1/2} \tilde{\bw}_H\|_2 + O(\sqrt{\lambda_{\max}} L) \right)^2 \\
    &\qquad \qquad \leq \frac{1}{\hat{\blambda}_{\max}^2} \left\| \Lambda_P^{1/2} (\bT U_P)^{\intercal} \bT U_H \Lambda_H^{1/2}\tilde{\bw}_H \right\|_2^2 + \frac{1}{\hat{\blambda}_{\max}^2} \cdot O\left(\lambda_{\max}^{3/2} + \sqrt{\lambda_{\max}} L \right) \cdot \sqrt{\lambda_{\max}} L \\
	&\qquad\qquad\leq \frac{1}{\hat{\blambda}_{\max}^2} \left\| \Lambda_P^{1/2} (\bT U_P)^{\intercal} \bT U_H \Lambda_H^{1/2} \tilde{\bw}_H \right\|_2^2 + O(L).
\end{align*}
where we use the fact $\| \Lambda_P^{1/2} (\bT U_P)^{\intercal} \bT U_H \Lambda_H^{1/2} \tilde{\bw}_H \|_2 = O(\lambda_{\max}^{3/2})$. Hence, up to an additive $O(L)$, and using the fact $\bw_H = \tilde{\bw}_H / \| \tilde{\bw}_H \|_2$ with $\| \tilde{\bw}_H \|_2 = O(\sqrt{\lambda_{\max}})$,  it suffices to upper bound
\begin{align*}
&\frac{1}{\lambda_{\max}} \left( \left\| \Lambda_P^{1/2} (\bT U_P)^{\intercal} \bT U_H \Lambda_H^{1/2} \bw_H \right\|_2^2 - \left\| \Lambda_P \bw_P \right\|_2^2\right),
\end{align*}
which is, in turn, at most the following sum of two terms (ignoring the scaling by $1/\lambda_{\max}$),
\begin{align}
\left( \left\| (\Lambda_0^+)^{1/2} (\bT U_0^+)^{\intercal} \bT U_H \Lambda_H^{1/2} \bw_H \right\|_2^2 - \left\| \Lambda_0^+ \bw_0^+ \right\|_2^2 \right) &+ \left\| (\Lambda_{-0}^+)^{1/2} (\bT U_{-0}^+)^{\intercal} \bT U_{H} \Lambda_H^{1/2} \bw_H \right\|_2^2. \label{eq:two-terms}
\end{align}
The second term above is simpler to bound, since we may use the triangle inequality and divide $\bw_H = \bw_0^+ + \bw_{-0}^+ + \bw_N$, resulting in the bound (once more, up to constant factors) of the terms:
\begin{itemize}
\item The contribution of $\bw_0^+$, given by $\left\| (\Lambda_{-0}^+)^{1/2} (\bT U_{-0}^+)^{\intercal} (\bT U_0^+) (\Lambda_0^+)^{1/2} \right\|_2^2 \| \bw_{0}^+ \|_2^2$, which is at most $O(\lambda_{\max} L)$ from Lemma~\ref{lem:spectral-norm-bounds}.
\item The contribution of $\bw_{-0}^+$, given by $\left\| (\Lambda_{-0}^+)^{1/2} (\bT U_{-0}^+)^{\intercal} (\bT U_{-0}^+) (\Lambda_{-0}^+)^{1/2} \right\|_2^2 \| \bw_{-0}^+\|_2^2$, which is at most $O(\lambda_{\max}^2)$ (using the fact that $\bT$ is a $O(1)$-subspace embedding of $U_{-0}^+$) times $O(Lr/ \lambda_{\max})(1+\gamma)^2$ from Lemma~\ref{lem:small-bottom}, giving us $O(\lambda_{\max} Lr \cdot (1+\gamma)^2)$. 
\item Finally, the contribution of $\bw_N$, given by $\| \Lambda_{-0}^{1/2} (\bT U_{-0}^+)^{\intercal} (\bT U_N) \Lambda_N^{1/2} \|_2^2 \| \bw_N \|_2^2$, which is at most $O(\lambda_{\max}L + \lambda_{\min}L)$ from Lemma~\ref{lem:spectral-norm-bounds} once more, times $O(Lr/\lambda_{\max})$ from Lemma~\ref{lem:neg-small-pos-large}, and giving a bound of $O(L^2r + \lambda_{\min} L^2r / \lambda_{\max})$.
\end{itemize}
This bounds the second term of (\ref{eq:two-terms}) by $O\left(\lambda_{\max} Lr (1+\gamma)^2\right)$. Turning back to the first term of (\ref{eq:two-terms}), we also apply the triangle inequality and decompose $\bw_H = \bw_{0}^+ + \bw_{-0}^+ + \bw_N$, but we expand the square carefully. Namely, we upper bound
\begin{align*}
 \left\| (\Lambda_0^+)^{1/2} (\bT U_0^+)^{\intercal} \bT U_H \Lambda_H^{1/2} \bw_H \right\|_2 &\leq  \left\| \Lambda_0^+\bw_0^+ \right\|_2 \\
 	&\qquad + \left\| (\Lambda_0^+)^{1/2} \left( (U_0^+)^{\intercal} U_{0}^+ - (\bT U_{0}^+)^{\intercal} (\bT U_0^+) \right)(\Lambda_0^+)^{1/2} \right\| \| \bw_0^+\|_2 \\
	&\qquad + \left\| (\Lambda_0^+)^{1/2} (\bT U_{0}^+)^{\intercal} (\bT U_{-0}^{+}) (\Lambda_{-0}^+)^{1/2} \right\|_2 \| \bw_{-0}^+ \|_2 \\
	&\qquad + \left\| (\Lambda_0^+)^{1/2} (\bT U_{0}^+)^{\intercal} (\bT U_{N}) \Lambda_{N}^{1/2} \right\|_2 \|\bw_{N} \|_2 \\
	&\leq \| \Lambda_{0}^+ \bw_{0}^+\|_2 + O(L) + O(\sqrt{\lambda_{\max} L}) \cdot \sqrt{Lr/\lambda_{\max}} (1+\gamma) \\
	&\qquad\qquad\qquad + O(\sqrt{\lambda_{\max} L}(1+\gamma)) \cdot \sqrt{L r / \lambda_{\max}}
\end{align*}
where we use Lemma~\ref{lem:spectral-norm-bounds}, along with the bounds on $\| \bw_N \|_2, \| \bw_{-0}^+\|_2$ from Lemma~\ref{lem:neg-small-pos-large} and Lemma~\ref{lem:small-bottom}. This gives a bound on the difference (note that (\ref{eq:two-terms}) considers square differences) of
\[ \left| \left\| (\Lambda_{0}^+)^{1/2} (\bT U_0^+)^{\intercal} \bT U_H \Lambda_H^{1/2} \bw_H \right\|_2 - \left\| \Lambda_0^+ \bw_0^+ \right\|_2 \right| \lsim L \sqrt{r}(1+\gamma). \]
Since $\| \Lambda_{0}^+ \bw_{0}^+\|_2 \leq \lambda_{\max}$, squaring the above and subtracting $\| \Lambda_{0}^+ \bw_{0}^+\|_2$ gives us a bound of $O(\lambda_{\max} L\sqrt{r})$. Hence, (\ref{eq:two-terms}) is at most $O(\lambda_{\max} L \sqrt{r} (1+\gamma) + L^2 r (1 + \gamma)^2)$, and incorporating the scaling of $1/\lambda_{\max}$ gives the upper bound (after simplifying, using the fact $\lambda_{\max}$ is always at least $L$).
\end{proof}

\begin{lemma}\label{lem:n-bound}
Whenever $\calbE$ holds,
\begin{align*}
\left\| \Lambda_N^{1/2} U_N^{\intercal} \bx \right\|_2^2 = O(Lrk \cdot (1 + \gamma + \gamma^2)^2).
\end{align*}
\end{lemma}

\begin{proof}[Proof of Lemma~\ref{lem:n-bound}]
We first substitute $\bx = U \Lambda^{1/2} \tilde{\bw} / \hat{\blambda}_{\max}$ which allows us to upper bound $\| \Lambda_N^{1/2} U_N^{\intercal} \bx \|_2$ by $\left\| \Lambda_N \tilde{\bw}_N \right\|_2 / \hat{\blambda}_{\max}$. By expanding $\bT U \Lambda^{1/2} \tilde{\bw}$ into a sum of four terms, $\bT U_0^+ (\Lambda_0^{+})^{1/2} \tilde{\bw}_0^+$, $\bT U_{-0}^+ (\Lambda_{-0}^+)^{1/2} \tilde{\bw}_{-0}^+$, $\bT U_{N} \Lambda_N^{1/2} \tilde{\bw}_N$ and $\bT U_{L} \Lambda_L^{1/2} \tilde{\bw}_L$, as well as the identity $U_N^{\intercal} U_N = I$, we can isolate $\Lambda_N \tilde{\bw}_N$, and adding and subtracting cross terms. In particular,
\begin{align*}
\left\| \Lambda_N \tilde{\bw}_N \right\|_2 &\leq \left\| \Lambda_N^{1/2} (\bT U_N)^{\intercal} \bT U \Lambda^{1/2} \tilde{\bw} \right\|_2 \\
				&\qquad + \left\| \Lambda_N^{1/2} (\bT U_N)^{\intercal} \bT U_{0}^+ (\Lambda_0^{+})^{1/2} \tilde{\bw}_{0}^+ \right\|_2 \\
				&\qquad + \left\| \Lambda_N^{1/2} (\bT U_N)^{\intercal} \bT U_{-0}^+ (\Lambda_{-0}^+)^{1/2} \tilde{\bw}_{-0}^+\right\|_2 \\
				&\qquad + \left\| \Lambda_{N}^{1/2} (\bT U_N)^{\intercal} \bT U_L \Lambda_L^{1/2} \tilde{\bw}_L \right\|_2 \\
				&\qquad + \left\| \Lambda_N^{1/2} \left( U_N^{\intercal} U_N - (\bT U_N)^{\intercal} (\bT U_N) \right) \Lambda_N^{1/2} \tilde{\bw}_N\right\|_2.
\end{align*}
Note, the first term $\| \Lambda_N^{1/2} (\bT U_N)^{\intercal} \bT U \Lambda^{1/2} \tilde{\bw} \|_2$, multiplied by $1/\hat{\blambda}_{\max}$ is exactly $\| \Lambda_N^{1/2} (\bT U_N)^{\intercal} \by \|_2$, which we know is at most $O(\sqrt{Lr})$ from Lemma~\ref{lem:neg-small-pos-large}. The second term, using Lemma~\ref{lem:spectral-norm-bounds}, becomes at most
\begin{align*}
O(\sqrt{\lambda_{\max}}) \cdot \left\| \Lambda_N^{1/2} (\bT U_N)^{\intercal} \bT U_{0}^+ (\Lambda_0^{+})^{1/2} \right\|_2 = O(\sqrt{L} \cdot \lambda_{\max} (1+\gamma)),
\end{align*}
so the scaling by $1/\hat{\blambda}_{\max}$ makes its contribution $O(\sqrt{L} (1+\gamma))$. The third term is similarly bounded by the spectral norm $\| \Lambda_N^{1/2} (\bT U_N)^{\intercal} \bT U_{-0}^+ (\Lambda_{-0}^+)^{1/2}\|_2$, which is bounded by $O(\sqrt{\lambda_{\max} L} + \sqrt{\lambda_{\min} L})$ from Lemma~\ref{lem:spectral-norm-bounds} times $O(\sqrt{\lambda_{\max}})$ (for the scaling of $\| \tilde{\bw}_H\|_2$), times $\| \bw_{-0}^+\|_2$, which is at most $O(\sqrt{Lr/\lambda_{\max}} (1+\gamma))$ from Lemma~\ref{lem:small-bottom}. Once more, after re-scaling by $1/\hat{\blambda}_{\max}$, we obtain a bound on the third term (up to constant factors) of
\[ \frac{1}{\lambda_{\max}} \cdot O(\sqrt{\lambda_{\max} L} + \sqrt{\lambda_{\min} L}) \cdot O(\sqrt{\lambda_{\max}}) \cdot O\left( \sqrt{\frac{Lr}{\lambda_{\max}}} (1+\gamma)\right) \lsim L \sqrt{r} (1 + \gamma) + \sqrt{L r} \cdot \gamma^2. \]
The fourth term becomes $O(\sqrt{\lambda_{\min}} \cdot L)$, and after scaling by $1/\hat{\blambda}_{\max}$ is $O(\sqrt{L} \cdot \gamma)$, and finally, the last term is $O(\sqrt{\lambda_{\max}})$ times $O(\sqrt{k\lambda_{\min} L})$ times $\| \bw_N \|_2 = O(\sqrt{Lr/\lambda_{\max}})$, which is $O(L\sqrt{rk\lambda_{\min}})$, which after the re-scaling by $1/\hat{\blambda}_{\max}$ is similarly $O(\sqrt{Lrk} \cdot \gamma)$.
\end{proof}

Putting Lemma~\ref{lem:p-bound} and Lemma~\ref{lem:n-bound} together, we conclude that (\ref{eq:rayleigh-diff}) is at most
\[ O\left(Lrk \cdot (1 + \gamma + \gamma^2 + \gamma^3 + \gamma^4) \right), \]
which means we must set $L$ so as to ensure the above bound is $O(\eps n)$.

\section{Lower Bound}\label{sec:lb}

In this section, we prove a lower bound on the query complexity required to output an approximately top eigenvector of an $n \times n$ symmetric matrix $A$ with $\|A\|_{\infty} \leq 1$. As we will see, the inputs we use will satisfy the additional assumption that $\|A\|_2 = O(\lambda_{\max}(A))$.

\mainthmlb*

The proof of Theorem~\ref{thm:main-lb-intro} follows by showing a reduction to the problem of estimating the bias of a biased coin. We first begin by formally defining the coin problem, which is known to require $\Omega(1/\eps^2)$ samples, and then we show how an algorithm satisfying the guarantees of Theorem~\ref{thm:main-lb-intro} with query complexity $q$ implies an algorithm for the coin problem which uses $O(q/n)$ samples.

\begin{problem}[Coin Problem]\label{prob:coin}
    For a fixed parameter $\eps \in (0, 1/2)$ and $b \in \{-1,1\}$, let $\calC(\eps b)$ denote the distribution supported on $\{-1,1\}$ with $\Ex_{\bc \sim \calC(\eps b)}[\bc] = 2\eps b.$ The coin problem refers to the following statistical task:
    \begin{itemize}
        \item A hidden variable $\bb\sim \{-1,1\}$ is drawn uniformly at random, and $s$ independent draws are generated $\bc_1 ,\dots, \bc_s \sim \calC(\eps \bb)$.
        \item An algorithm for the coin problem is specified by a function $f\colon \{-1,1\}^s \to \{-1,1\}$, and the algorithm outputs $\hat{\bb} = f(\bc_1, \dots, \bc_n)$. 
    \end{itemize}
    It is well-known that any distribution over functions $\boldf \colon \{-1,1\}^n \to \{-1,1\}$ (which may utilize randomness independent of $\bb, \bc_1,\dots, \bc_s$) which satisfies
    \[ \Prx_{\substack{\bb \sim \{-1,1\} \\ \bc_1,\dots, \bc_s \sim \calC(\eps \bb)}}\left[ \boldf(\bc_1,\dots, \bc_s) = \bb \right] \geq 0.9\]
    must have $s = \Omega(1/\eps^2)$.\footnote{This fact follows from Le Cam's method and the fact $\dtv(\calC(\eps)^s, \calC(-\eps)^s)$ is an arbitrarily small constant when $s$ is a small enough constant factor of $1/\eps^2$, see~\cite{C20}.}
\end{problem}

Hence, our task will be to extract, from a randomized algorithm for approximating the top eigenvector, a function $f$ which estimates the bias of a coin. To do so, we will first define a distribution over bounded-entry symmetric $n\times n$ matrices $A$ in order to use Yao's principle and consider a deterministic algorithm. 

\begin{definition}\label{def:hard-distribution} 
For a fixed $\eps \in (0,1/2)$ and $x \in \{-1,1\}^n$, let $\calD_1(\eps,x)$ be the following distribution over $n\times n$ symmetric matrices $\bA$, defined according to the following procedure:
\begin{itemize}
    \item Draw a symmetric $n\times n$ matrix $\bB \in \{-1,1\}^{n\times n}$ with diagonal entries $\bB_{ii} = 1$ uniformly at random. 
    \item The output $\bA \in \{-1,1\}^{n\times n}$ is the symmetric matrix given by letting, for each $i, j$, 
    \[\bA_{i,j} = \bA_{j,i} \sim \begin{cases}
    x_i \cdot x_j & \text{with probability }2 \eps \\
    \bB_{i,j} & \text{with probability } 1-2\eps
\end{cases} .\]
\end{itemize}
\end{definition}
Note for any $\eps > 0$ and $x \in \{-1,1\}^n$, a draw $\bA \sim \calD_1(\eps,x)$ always satisfies $\|\bA\|_{\infty} \leq 1$ and has spectral norm $\| \bA \|_2$ which is at most $2\eps n + O(\sqrt{n})$ with high probability, and $\lambda_{\max}(\bA) \geq 2\eps n - O(\sqrt{n})$.\footnote{This fact follows from Corollary~4.4.7 of~\cite{V18}, see, also the proof of Claim~\ref{cl:spectral-norm}.} Hence, we consider $\eps$ which is at least a large constant factor of $1/\sqrt{n}$---this best possible since there are at most $n^2$ entries---and we also satisfy the desired promise $\|\bA\|_2 = O(\lambda_{\max}(\bA)$. We assume the existence of a randomized algorithm $\calA$ satisfying the requirements of Theorem~\ref{thm:main-lb-intro} which may be run with accuracy parameter $\eps_0 = \eps / 100$. This implies, via Yao's principle, a deterministic algorithm $\calA$ which succeeds at finding an approximately top eigenvector from the distribution over matrices given by first sampling $\bx \sim \{-1,1\}^n$ and then $\bA \sim \calD_1(\eps,\bx)$. Since we will deal with general (possibly) adaptive query algorithms $\calA$, it will be useful to formally specify an algorithm, and note that it suffices to define algorithms operating on $n\times n$ symmetric and bounded entry matrices whose entries are $\{-1,1\}$.
\begin{definition}[Deterministic Query Algorithm for the Hard Distribution]\label{def:det-alg}
Fix $\eps_0 = \eps / 100$, a deterministic algorithm $\calA$ for approximate eigenvectors of symmetric matrices in $\{-1,1\}^{n\times n}$ is a binary decision tree;
\begin{itemize}
    \item Every non-leaf vertex is labeled with a pair of indices $(i,j) \in [n]\times [n]$, and has two outgoing edges which are labeled with $\{-1,1\}$.
    \item Every leaf vertex is labeled with an output vector $y \in \R^n$.
\end{itemize} 
An execution on input matrix $A \in \{-1,1\}^{n\times n}$ is a root-to-leaf which (recursively) proceeds by following the edge labeled $A_{ij} \in \{-1,1\}$ for index pair $(i,j)$ at a node, and upon reaching a leaf, outputting the stored vector $y$. We consider a deterministic algorithm $\calA$ which satisfies:
\begin{align} 
\Prx_{\substack{\bx \sim \{-1,1\}^n \\ \bA \sim \calD_1(\eps,\bx)}}\left[ \lambda_{\max}(\bA) - \dfrac{\langle \by, \bA \by \rangle}{\|\by\|_2^2} \leq \eps_0 n \right] \geq 0.99,  \label{eq:guarantee-lb}
\end{align}
where $\by = \calA(\bA) \in \R^n$ is the output vector. Let $q(A) \in \N$ denote the query complexity of the algorithm on input $A \in \{-1,1\}^{n\times n}$ (i.e., the depth of the root-to-leaf path on $A$), and for any $\ell \in [n]$, $q_{\ell}(A) \in \N$ the number of queries to entries in row/column $\ell$. We seek to lower bound $q = \max_{A \in \{-1,1\}^{n\times n}} q(A)$.
\end{definition}

\ignore{\begin{definition}[Randomized Algorithm]  Let a randomized algorithm $\textsc{ALG}_R$ be a distribution $\calR$ over a set of decision trees $\calT$. An execution of randomized algorithm consists of sampling a decision tree $T\sim \calR$ and then running $T$ on an input $A$.
\end{definition}

\begin{lemma}[Yao's Lemma] \label{lem:yaos-lemma}  Given a set of matrices $\calM$ and a randomized decision tree $\calR$ (distribution over deterministic trees $\calT$ of depth $D$) which for any input in $\calM$ outputs an approximate top eigenvector with probability $0.99$, then for any distribution $\calD$ of matrices $\calM$, there exists a deterministic decision tree $T \in \calT$ of depth at most $D$ which outputs an approximate top eigenvector with probability $0.99$ over inputs sampled from $\calD$.
\end{lemma}

For any symmetric $A \in \{-1,1\}^{n\times n}$ and $\eps \in (\frac{1200}{\sqrt{n}},1)$, let $\textsc{ALG}_R(A,\eps)$ be a randomized algorithm which outputs a vector $\by$  such that 
$ \Pr_{\by}\left [ \frac{\la \by, A\by \ra}{\|\by\|_2^2}  - \lambda_{\max}(A) \leq \eps n/100 \right ] \geq 0.99 $ and makes no more than $n/1000\eps^2$ queries to $A$. 
We start by applying Yao's Lemma on $\textsc{ALG}_R$ to draw a deterministic decision tree $\textsc{ALG}$ for the distribution $\calD_1(\bx,\eps)$, where $\bx \sim \{-1,1\}^n$.}

\subsection{Equivalent Method for Sampling from Hard Distribution}

Recall that, in our reduction, the desired goal is to extract a function $f$ from the deterministic decision tree $\calA$ satisfying (\ref{eq:guarantee-lb}), and for this purpose, it will be useful to define another (as we will later prove) equivalent mechanism for generating a sample $\bx \sim \{-1,1\}^n$ and $\bA \sim \calD_1(\eps,\bx)$.

\begin{definition}\label{def:coin-distribution} 
For a fixed $\eps \in (0,1/2)$, $\ell \in [n]$, $z\in \{-1,1\}^{[n]\setminus \{\ell\}}$ and a sequence $c_1,c_2,\dots c_n \in \{-1,1\}$, let $\calD_2(\eps,\ell,z,c_1,c_2,\dots c_n)$ be the following distribution over symmetric matrices $\bA$, defined according to the following procedure:
\begin{enumerate}
    \item Draw a symmetric $n\times n$ matrix $\bB \in \{-1,1\}^{n\times n}$ with diagonal entries $\bB_{ii} = 1$ uniformly at random.
    \item The output $\bA \in \{-1,1\}^{n\times n}$ is the symmetric matrix given by letting, for each $i, j$,
    \[\bA_{i,j} = \bA_{j,i} \sim \begin{cases}
    1 & \text{if $i=j=\ell$}\\
    c_{j} z_j & \text{if $i=\ell$ or $j=\ell$}\\
    z_i z_j &\text{else, with probability } 2\eps\\
    \bB_{i,j} &\text{else, with probability } 1 - 2\eps\\
    \end{cases} .
    \]
\end{enumerate}
\end{definition}

The equivalence, which we will heavily rely on in our reduction, is captured by the following lemma.

\begin{lemma} \label{clm:random-process}For any fixed $\eps \in (0,1/2)$, consider the following random process.
\begin{enumerate}
    \item Sample $\bb \sim \{-1,1\}$, then sample $\bc_1,\dots,\bc_n \sim \calC(\eps \bb)$.
    \item Draw $\bell \sim [n]$ and $\bz \sim \{-1,1\}^{[n]\setminus\{\bell \} }$. Now draw $\bA \sim \calD_2(\eps,\bell,\bz,\bc_1,\dots,\bc_n)$.
    \item Let $\bx= (\bz_1,\dots, \bz_{\bell-1},\bb,\bz_{\bell+1},\dots,\bz_{n})$.
    \item Output: $\left (\bb, \bA, \bell,\bz, \bc_1,\dots, \bc_n, \bx\right )$.
\end{enumerate}

The marginal distribution of  $(\bb,\bx,\bell,\bA)$ has the following properties 
\begin{itemize}
    \item $\bx \sim \{-1,1\}^n$ and $\bA \sim \calD_1(\eps,\bx)$
    \item $\bell \sim [n]$ and independent of $\bx$ and $\bA$ \item $\bx_{\bell} = \bb$
\end{itemize}  
\end{lemma}

\begin{proof}
Note that by construction $\bx_{\bell} = \bb$. We start by proving $\bx \sim \{-1,1\}^n$. Consider any fixed setting of $\bell=\ell$, now we compute the probability that $\bx = x$ for any $x \in \{-1,1\}^n$. Note that in the process each $\bz_i$ is chosen independently and uniformly in $\{-1,1\}$, and $\bb$ is chosen uniformly from $\{-1,1\}$ and independently from $\bz$. 
\begin{align*}
    \Prx_{\bb,\bz}[\bx = x \vert \bell=\ell]  = \Prx_{\bb}[\bx_{\ell} = x_{\ell}] \cdot \prod_{\substack{i=1\\i\neq\ell}}^n \Pr_{\bz}[\bx_i = x_i] = \Prx_{\bb}[\bb = x_{\ell}] \cdot \prod_{\substack{i=1\\i\neq\ell}}^n \Prx_{\bz}[\bz_i = x_i] =\frac{1}{2^n},
\end{align*}
and this means $\bx$ is drawn uniformly from $\{-1,1\}^n$. To show that $\bx$ and $\bell$ are independent, we evaluate the joint probability of $\bx = x$ for any $x \in \{-1,1\}^n$ and $\bell = \ell$ for any $\ell \in [n]$:
\begin{align*}
    \Prx_{\bb,\bz,\bell}[\bx = x \wedge \bell = \ell] &= \Prx_{\bb,\bz}[\bx = x \vert  \bell = \ell] \cdot \Prx_{\bell}[\bell = \ell] = \frac{1}{2^n} \cdot \Prx_{\bell}[\bell = \ell] = \Prx_{\bx\sim \{-1,1\}^n}[\bx = x] \cdot \Prx_{\bell}[\bell = \ell],
\end{align*}
which means $\bx\sim\{-1,1\}^n$ and $\bell$ are independent. It remains to show that $\bx\sim \{-1,1\}^n$ and $\bA \sim \calD_1(\eps,\bx)$. 

Fix any setting of $\bell = \ell$, and recall we always sample triples $(\bx,\bz,\bb)$ which set $\bx= (\bz_1,\dots, \bz_{\ell-1},\bb,\bz_{\ell+1},\dots,\bz_{n})$. If $i=j$, the entries $\bA_{i,i}$  are always set to one, which is also the case in $\calD_1(\eps,\bx)$. Consider the distribution of entries in $\bA_{i,j}$ for any $i,j$ where neither is equal to $\ell$. By Definition~\ref{def:coin-distribution},  $\bA_{i,j} = \bz_i\bz_j$, which is equal to $\bx_i \bx_j$, with probability $2\eps$ and $\bB_{i,j}$ with probability $1-2\eps$. Since $\bB_{i,j}$ is drawn from the same distribution in both $\calD_1(\eps,\bx)$ and $\calD_2(\eps,\bell,\bz,\bc_1,\dots \bc_n)$. This completes the case of $i \neq j$, and neither equals $\ell$. The final case considers $i=\ell$ and $j\neq \ell$ (since $\bA$ is always symmetric). In this case, Definition~\ref{def:coin-distribution} sets $\bA_{\ell,j}$ to $\bc_{j}\bz_j$, so the probability that $\bA_{\ell,j} = \bb \bz_j = \bx_{\ell} \bx_{j}$ is
\begin{align*}
    \Prx_{\bb,\bc_1,\dots,\bc_n,\bA}[\bA_{\ell,j} = \bb \bz_j] = \Prx_{\bc_1,\dots,\bc_n,\bA}[\bc_{j} 
    \bz_j = \bb \bz_j] = \Prx_{\bb, \bc_j}\left[\bb = \bc_j \right] = 1/2 + \eps.
\end{align*}
In a similar vein, we using Definition~\ref{def:hard-distribution}, we directly have $\bA_{\ell,j} = \bx_{\ell} \bx_j$ with probability $2\eps$ and uniform (from the fact its set to $\bB_{\ell,j}$) with probability $1-2\eps$. Hence, $\bA_{\ell,j} = \bx_{\ell} \bx_{j}$ with probability $1/2 + \eps$ as well. We can  conclude that $\bA \sim \calD_1(\eps,\bx)$ for $\bx\sim\{-1,1\}^n$ and since $\bx$  independent of $\bell$, then $\bA$ must also be independent of $\bell$. 
\end{proof}

\subsection{Using the Deterministic Algorithm $\calA$ to solve the Coin Problem}\label{sec:reduction}

The following algorithm is given a sequence of i.i.d. draws $\bc_1,\dots \bc_n$ from a distribution $\calC(\eps \bb)$ with known parameter $\eps$ which is a large enough constant factor of $1/\sqrt{n}$ and at most $1/2$ and uniform $\bb\sim \{-1,1\}$, and has access to a deterministic algorithm $\calA$ satisfying (\ref{eq:guarantee-lb}); we show how to output a random variable $\smash{\hat{\bb}}$ which agrees with $\bb$ with probability at least $0.9$.

\begin{figure}[H]
\begin{framed}
\textbf{Reduction.}\label{alg:lb-reduction}
Given an accuracy parameter $\eps$ which is a large enough constant factor of $1/\sqrt{n}$ and random variables $\bc_1,\dots \bc_n\sim \calC(\eps \bb)$, as well as a deterministic algorithm $\calA$ satisfying (\ref{eq:guarantee-lb}), we proceed as follows:
\begin{enumerate}
    \item Draw $\bell \sim [n]$ and $\bz\sim\{-1,1\}^{[n]\setminus \{\ell\}}$.
    \item Draw $\bA\sim \calD_2(\eps,\bell,\bz,\bc_1,\dots \bc_n)$ as described by Definition \ref{def:hard-distribution}.

    \item Execute $\calA$ on $\bA$ while ensuring that the number of queries made to the $\bell$-th row/column is no greater than $200q/n$. If the number of queries exceeds $200q/n$, output \smash{$\hat{\bb} \sim\{-1,1\}$}. Otherwise, let the vector $\by \in \R^n$ be the output of $\calA$.

    \item Select $\sigma \in \{-1,1\}$ such that $\sum_{i\in [n]\setminus \{\bell\}} \sigma \by_i \cdot \bz_i \geq 0$.
    
    \item  Output \smash{$\hat{\bb} = \sgn(\sigma \by_{\bell})$} as the estimate of $\bb$.
\end{enumerate}
\end{framed}
\caption{Reduction from $\calA$ satisfying (\ref{eq:guarantee-lb}) to the Coin Problem.} \label{fig:reduction}
\end{figure}

\ignore{Given the above claims, we prove that running Algorithm \ref{alg:lb-reduction} on $\bA \sim \calD_2(\eps,\bell,\bz,\bc_1,\dots \bc_n)$ outputs $\hat{\bb}\in\{-1,1\}$ such that $\Pr_{\substack{\bb, \bz, \bell, \bA \\ \bc_1, \dots, \bc_n}}[\hat{\bb} = \bb] \geq 0.9$ using fewer than $1/10\eps^2$ samples of $\bc_1,\dots\bc_n \sim \bC(2\eps\bb)$. We start by analyzing the probability of failure.
\begin{align*}
    \Pr_{\substack{\bb, \bz, \bell, 
    \bA \\ \bc_1, \dots, \bc_n}}[\hat{\bb} = \bb ] \geq 1 - \Pr_{\substack{\bb, \bz, \bell, 
    \bA \\ \bc_1, \dots, \bc_n}}[\hat{\bb} \neq \bb]
\end{align*}}
We first upper bound the probability that the reduction in Figure~\ref{fig:reduction} outputs a random variable $\smash{\hat{\bb}}$ which disagrees with $\bb$. The reduction can output an incorrect $\hat{\bb}$ for a few different reasons, one being that the query complexity to the $\bell$-th row/column is too high causing the algorithm to output a random sign as $\hat{\bb}$, another being that $\calA$ outputs a vector $\by$ with small Raleigh quotient, and the last being that $\calA$ output a $\sigma\by$ which was simply had the wrong sign at $\bell$. The subsequent analysis will use structural properties of $\bA$ and $\bx$, as well as the condition (\ref{eq:guarantee-lb}) to ensure $\smash{\hat{\bb}}$ is equal to $\bb$ except with probability $0.9$.

\begin{definition}
The following events are defined with respect to the random variables $\bb$, $\bc_1,\dots, \bc_n$, $\bell$, $\bz$, $\bA$, $\by$, and \smash{$\hat{\bb}$} which are generated as part of an execution of Figure~\ref{fig:reduction}. We let $\bx \in \{-1,1\}^n$ be set to
\[ \bx = (\bz_1,\dots, \bz_{\bell-1}, \bb, \bz_{\ell+1}, \dots, \bz_n) \in \{-1,1\}^n,\]
as in Lemma~\ref{clm:random-process}. The following should be thought of as failure events:
\begin{itemize}
    \item Let  $\calbE_0$ be the event where $\|\bA - 2\eps \bx \bx^{\intercal}\|_2 \geq c \sqrt{n}$, for a large enough constant $c >0$.
    \item Let $\calbE_1$ be the event where $q_{\bell}(\bA) > 200q/n$.  
    \item Let $\calbE_2$ be the event where $\la \by , \bA \by\ra / \|\by\|_2^2 \leq \lambda_{\max}(\bA) - \eps_0 n$ for $\eps_0 = \eps / 100$ (Definition~\ref{def:det-alg}). \footnote{\label{ftn:1} Notice, the vector $\by$ is only generated in Figure~\ref{fig:reduction} whenever $q_{\bell}(\bA) \leq 200q/n$, i.e., event $\calbE_1$ does not occur.}
    \item Let  $\calbE_3$ be the event where $ \sgn(\sigma \by_{\bell})\neq \bb$. \footnote{Once more (similarly to Footnote~\ref{ftn:1}), $\calbE_3$ is only defined when $\by$ is generated by Figure~\ref{fig:reduction} and occurs whenever $\calbE_1$ does not occur.}
\end{itemize}
\end{definition}

The analysis of Figure~\ref{fig:preprocess} will proceed by upper-bounding the probability that $\hat{\bb} \neq \bb$ by relating to the events $\calbE_0, \calbE_1, \calbE_2, \calbE_3$. Notice that it suffices to upper bound
\begin{align}
        \Prx[\hat{\bb} \neq \bb ] \leq &\Prx[\calbE_0] 
        + \Prx[\calbE_1] 
        + \Prx[\calbE_2  \wedge  \overline{\calbE_1} ]
        + \Prx\left[ \calbE_3   \wedge    \overline{\calbE_2} \wedge  \overline{\calbE_1} \wedge \ol{\calbE_0} \right], \label{eq:lb-faliure-prob}
\end{align}
where all of the probabilities are over the randomness generated in Figure~\ref{fig:reduction}, because whenever none of the above failure events occur, the final output $\hat{\bb}$ in Figure~\ref{fig:reduction} is equal to $\sgn(\sigma \by_{\bell})$, which is by the fact $\calbE_3$ does not occur equal to $\bb$. This leaves the four probabilistic events on the right-hand side of (\ref{eq:lb-faliure-prob}) to upper bound, and we show each is at most $0.01$ in the following claims.

\begin{claim}\label{cl:spectral-norm} $\Prx[\calbE_0] \leq 0.01$.    
\end{claim}

\begin{proof}
By Lemma~\ref{clm:random-process}, it suffices to consider the matrix $\bA$ generated by letting $\bx\sim \{-1,1\}^n$ and then $\bA \sim \calD_1(\eps,\bx)$. Then, we may re-write $\bA = (1-2\eps)I + 2 \eps \bx \bx^{\intercal} + \bZ$, where $\bZ$ is a zero-centered symmetric $n\times n$ matrix whose entries are independent and bounded in magnitude by $2$. Corollary~4.4.6 of~\cite{V18} implies $\|\bZ\|_{2} = O(\sqrt{n})$ with high probability, which means that $\| \bA - 2\eps \bx\bx^{\intercal}\|_2$ is at most $1 + O(\sqrt{n})$. 
\ignore{\begin{align*}
    \Pr_{\substack{\bb, \bz, \bell, \bA \\ \bc_1, \dots, \bc_n}}[\calE_0] =  \Pr_{\bx, \bA}\left [\lambda_{i}(\bA - \E_{\bA}[\bA]) \in [-6\sqrt{n},   6 \sqrt{n}] \right]
\end{align*}
Define $\bG = \bA - \E[\bA]$, each entry $\bG_{i,j}$ is zero centered $\E[\bG_{i,j}] = 0$ and is bounded by $\vert \bG_{i,j} \vert \leq (1+2\eps)$. Hence the variance $\Var[\bG_{i,j}] \leq (1+2\eps)^2$ and by Wigner's Semi-circle law, with probability $0.99$ it's eigenvalues are bounded by $\pm 2(1+2\eps)\sqrt{n}$. Therefore 
\begin{align*}
    &\Ex_{\bx}\Pr_{\bA}\left [ \lambda_{i}(\bA - \E_{\bA}[\bA]) \notin [-6\sqrt{n},   6 \sqrt{n}] \right] \leq 
    0.01
\end{align*}
concluding the proof. }
\end{proof}

\begin{claim}\label{cl:query-complexity} $\Prx[\calbE_1] \leq 0.01$.  
\end{claim}
\begin{proof}
Using Lemma~\ref{clm:random-process} once more, the matrix $\bA$ is independent from $\bell \sim [n]$. The fact that $\calA$ has query complexity at most $q$ means that for any draw of $\bA$, $q_1(\bA) + q_2(\bA) + \dots + q_n(\bA) \leq 2q$ (the factor of $2$ is from the fact each query pair $(i,j)$ can contribute to at most $2$ row/columns). Since $\bell$ is independent, the expectation over $\bell$ of $q_{\bell}(\bA)$ is at most $2q/n$, so by Markov's inequality, it is at most $200q/n$ except with probability at most $0.01$.
\ignore{Recall by lemma \ref{lem:yaos-lemma} that the worst-case query complexity (depth of the tree $T$) of $\Alg$ on any $\bA\sim \calD_1(\eps,\bx)$ for $\bx \sim \{-1,1\}^n$ is at most $\frac{n}{1000\eps^2}$. Hence we can bound the probability of making at most $1/10\eps^2$ queries to the $\bell$-th row / column by $100/n$ times the worst case query complexity of the algorithm. Let $q(\Alg(\eps,\bA)$ be the number of queries $\Alg$ makes to $A$ when run with accuracy $\eps$. Let $q_{\bell}(\Alg(\eps,\bA))$ be the number of queries $\Alg$ with accuracy $\eps$ makes to the $\bell$-th row or column of A. 
\begin{align*}
\Pr_{\substack{\bb, \bz, \bell, \bA \\ \bc_1, \dots, \bc_n}}[\calE_1] &= \Pr_{\substack{\bb, \bz, \bell, \bA \\ \bc_1, \dots, \bc_n}}[q_{\bell}(\Alg(\eps,\bA)) \geq \frac{1}{10\eps^2}] = \Pr_{\substack{\bb, \bz, \bell, \bA \\ \bc_1, \dots, \bc_n}}[q_{\bell}(\Alg(\eps,\bA)) \geq \frac{100}{n} q(\Alg(\eps,\bA)) ]
\end{align*}
Now observe that by definition, the expected query complexity over a randomly sampled column $\E_{\bell' \sim [n]}\left [ q_{\bell'}(\Alg(\eps,\bA))\right ]$ is at most $1/n$ times the worst case query complexity, where $\sum_{\ell}q_{\ell}(\Alg) = q(\Alg)$.
\begin{align*}
\Pr_{\substack{\bb, \bz, \bell, \bA \\ \bc_1, \dots, \bc_n}}[q_{\bell}(\Alg(\eps,\bA)) \geq \frac{100}{n} q(\Alg(\eps,\bA)) ] = \Pr_{\substack{\bb, \bz, \bell, \bA \\ \bc_1, \dots, \bc_n}}[q_{\bell}(\Alg(\eps,\bA)) \geq 100\E_{\bell' \in [n]}\left [ q_{\bell'}(\Alg(\eps,\bA))\right ] ]
\end{align*}
Now we apply claim \ref{clm:random-process}, to instead analyze $\bx\sim \{-1,1\}^n$ and $\bA \sim \calD_1(\eps,\bx)$. Note that the probability over $\bell$ is unaffected by the change to taking the probability over $\bx\sim \{-1,1\}^n$ and $\bA \sim \calD_1(\eps,\bx)$, as it is independent from $\bx$ and $\bA$. The proof is concluded by applying Markov's inequality. 
\begin{align*}
\Pr_{\substack{\bb, \bz, \bell, \bA \\ \bc_1, \dots, \bc_n}}[q_{\bell}(\Alg(\eps,\bA)) \geq 100\E_{\bell' \in [n]}\left [ q_{\bell'}(\Alg(\eps,\bA))\right ] ]
&= \Pr_{\substack{\bx \bA,\bell}} \left [\left [q_{\bell}(\Alg(\eps,\bA)) \geq 100\E_{\bell' \in [n]}\left [ q_{\bell'}(\Alg(\eps,\bA)) \right ] \right ] \right ]\\
&= \Ex_{\substack{\bx \bA}} \left [ \Pr_{\bell} \left [q_{\bell}(\Alg(\eps,\bA)) \geq 100\E_{\bell' \in [n]}\left [ q_{\bell'}(\Alg(\eps,\bA)) \right ] \right ] \right ]\\
&\leq 0.01
\end{align*}}
\end{proof}

\begin{claim}\label{cl:good-y} $\Prx[\calbE_2 \wedge \ol{\calbE_1} ] \leq 0.01$. 
\end{claim}
\begin{proof}
Whenever $\calbE_1$ does not occur, Figure~\ref{fig:reduction} produces the vector $\by$ which is the output of $\calA$ on input $\bA$. By Lemma~\ref{clm:random-process} once more, we may equivalently consider the matrix $\bA$ as being generated from $\bx \sim \{-1,1\}^n$ and $\bA \sim \calD_1(\eps, \bx)$, and therefore, the algorithm $\calA$ assumption (\ref{eq:guarantee-lb}) implies that $\lambda_{\max}(\bA) - \langle \by, \bA \by \rangle / \| \by\|_2^2$ is at most $\eps_0 n$, except with probability at most $0.01$, as desired.
\ignore{
The proof is fairly straightforward. We apply claim \ref{clm:random-process}, to instead analyze the probability of failure over $\bx\sim \{-1,1\}^n$ and $\bA \sim \calD_1(\eps,\bx)$. Then by lemma \ref{lem:yaos-lemma}, we have that with probability $0.01$ the vector $\by$ found by $\Alg(\eps,\bA)$ for $\bA\sim \calD(\eps,\bx)$ and $\bx\sim\{-1,1\}^{n \times n}$ has an insufficient Raleigh quotient. 
\begin{align*}
    \Pr_{\substack{\bb, \bz, \bell, \bA \\ \bc_1, \dots, \bc_n}}[\calE_2 \wedge \ol{\calE_1} ] &= 
    \Pr_{\substack{\bb, \bz, \bell, \bA \\ \bc_1, \dots, \bc_n}}[\frac{\la \by, \bA \by\ra}{\|\by\|_2^2} \leq \lambda_{\max}(\bA)-\eps n/100  \wedge \ol{\calE_1} ]\\
    &= \Pr_{\substack{\bx,\bA,\bell}} [\frac{\la \by, \bA \by\ra}{\|\by\|_2^2} \leq \lambda_{\max}(\bA)-\eps n/100 \wedge \ol{\calE_1}]\\
    &\leq \Pr_{\substack{\bx,\bA}} [\frac{\la \by, \bA \by\ra}{\|\by\|_2^2} \leq \lambda_{\max}(\bA)-\eps n/100]\\
    &\leq 0.01
\end{align*}}
\end{proof}

\begin{claim}\label{cl:few-disagree} Whenever none of the events $\calbE_0, \calbE_1, \calbE_2$ occur, $\left \vert \{i \in [n] : \sgn(\bx_i) = \sgn(\sigma \by_i) \right \vert \geq 0.99n$.
\end{claim} 

\begin{proof} We bound $\sgn(\bx_i) = \sgn(\sigma \by_i)$ by first bounding $
\la \by,\bx\ra$. Note,
\begin{align*}
    \frac{\la \by ,\bA \by\ra}{\|\by \|_2^2} &= \frac{\la \by , 2\eps \bx \bx^{\intercal} + (\bA - 2\eps \bx \bx^{\intercal}) \by\ra}{\|\by \|_2^2}
    \leq \frac{2\eps \cdot \la \by , \bx \ra^2}{\|\by \|_2^2} + \| \bA - 2\eps \bx \bx^{\intercal} \|_2,
\end{align*}
and the right-most term above is at most $c\sqrt{n}$ since $\calbE_0$ does not occur. Since neither $\calbE_1$ nor $\calbE_2$ occur, the left-hand side above is also lower bounded by $\lambda_{\max}(\bA) - \eps_0 n$, which implies
\begin{align*}
\langle \by, \bx \rangle^2 \geq \frac{(\lambda_{\max}(\bA) - \eps_0 n - c\sqrt{n}) }{2\eps} \cdot  \| \by\|_2^2 \geq \left( (1 - \eps_0 / (2\eps)) n - (c/\eps) \sqrt{n} \right) \|\by\|_2^2 \geq .99n \cdot \| \by\|_2^2
\end{align*}
since $\lambda_{\max}(\bA) \geq 2\eps n - \| \bA - 2\eps \bx \bx^{\intercal} \|_2 \geq 2\eps n - c\sqrt{n}$, and we use the fact $\eps_0 = \eps / 100$ and $\eps \geq 200c/\sqrt{n}$.
\ignore{$\ol{\calE_2}\wedge \ol{\calE_1}$, $\by$ satisfies: 
\[\frac{\la \by ,\bA \by\ra}{\|\by \|_2^2} \geq \lambda_{\max}(\bA) - \eps n/100\] We again
split $\bA$ into $\E[\bA] + (\bA - \E[\bA])$ and use $\ol{\calE_0}$ to bound \[\lambda_{\max}(\bA)\geq \lambda_{\max}(\E[\bA]) - \lambda_{\max}(\bA - \E[\bA]) \geq 2\eps n - 6\sqrt{n}\] 
Combining both bounds and simplifying leaves
\begin{align*}
    \la \by , \bx \ra^2 \geq \frac{\left (2\eps n - 12\sqrt{n} -\eps n/100\right ) \|\by\|_2^2}{2\eps}
\end{align*}
which for sufficiently large $n$ and $\eps \geq \frac{1200}{\sqrt{n}}$ can be reduced to 
\begin{align*}
    \la \by , \bx \ra^2 \geq \frac{99 n}{100} \cdot \|\by\|_2^2
\end{align*}}
As a result, the magnitude $|\langle \by, \bx\rangle|$ is at least $\sqrt{.99n} \cdot\|\by\|_2$, and a very loose bound implies it is larger than $10 |\by_{\ell}|$ for large enough $n$. Since $\langle \by, \bx \rangle$ differs from $\sum_{i \neq \ell} \by_i \bz_{i}$ by at most $|\by_{\ell}|$, the choice of $\sigma \in \{-1,1\}$ in Figure~\ref{fig:reduction} will ensure $\sigma \langle \by, \bx\rangle$ is positive. Therefore, letting $\bS = \{ i \in [n] : \sgn(\sigma \by_i) = \bx_i\}$, we have that an application of Cauchy-Schwarz, as well as the fact $\bx \in \{-1,1\}^n$ implies 
\begin{align*}
\| \by\|_2 \sqrt{|\bS|} \geq \sum_{i \in \bS} |\by_i| \geq \sum_{i \in \bS} |\by_i| - \sum_{i \notin \bS} |\by_i| = \sigma \langle \by, \bx\rangle,
\end{align*}
and since the right-hand side is positive, we may square both sides without changing the direction of the inequality. This implies that $\|\by\|_2^2 |\bS|$ is at least $\langle \by, \bx\rangle^2$, which is at least $.99n \cdot \|\by\|_2^2$. This implies the lower bound on $|\bS|$ as desired.
\ignore{
Now, since $\sigma \in \{-1,1\}$ is set such that $\sum_{i\in [n]\setminus \{\bell\}} \sigma \by_i \cdot \bz_i >0 $, this also implies that $\la \sigma \by , \bx\ra \geq 0$.  (Otherwise $\vert \sigma\by_{\bell}\bx_{\bell}\vert > \frac{99 \sqrt{n}}{100} \cdot \|\by\|_2$). Now observe that the inner product is at most the sum of coordinates $i$ where $ \sgn(\sigma \by_i) = \sgn(\bx_i)$. Let $S \{i \in [n] : \sgn(\bx_i) = \sgn(\sigma \by_i)\}$  \[\la \sigma \by , \bx\ra = \sum_i \sigma \by_i \bx_i \leq \sum_{i \in S} \vert \by_i \vert \leq \sqrt{\vert S\vert}\|\by\|_2^2 \]
Applying the lower bound and squaring proves $\vert S\vert \geq \frac{0.98 n}{100}$. }
\end{proof}

\begin{claim} $\Prx[\calbE_3 \wedge \ol{\calbE_2} \wedge \ol{\calbE_1}\wedge \ol{\calbE_0}] \leq 0.01$.
\end{claim}

\begin{proof}
Whenever none of $\calbE_0, \calbE_1, \calbE_2$ occur, Claim~\ref{cl:few-disagree} implies that for $0.99$-fraction of coordinates $i \in [n]$, $\sgn(\sigma \by_i) = \bx_i$, and since $\bx_{\bell} = \bb$ and $\bell$ is independent and $\bx$ and $\bA$ (using Claim~\ref{clm:random-process} once more), $\bell \sim [n]$ must fall  among the $0.01$-fraction of coordinates that $\sgn(\sigma \by_i) \neq \bx_i$ for $\calbE_3$ to occur while $\calbE_0, \calbE_1, \calbE_2$ do not occur, and this concludes the proof.
\ignore{
The proof is again straightforward. We first note that $\bx_{\bell} = \bb$ so the event $\calE_3$ can be rewritten as $\sgn(\sigma \by_{\bell}) \neq \bx_{\bell}$. We then apply claim \ref{clm:random-process}, to instead analyze $\bx\sim \{-1,1\}^n$ and $\bA \sim \calD_1(\eps,\bx)$. Note that the probability over $\bell$ is unaffected by the change to taking the probability over $\bx\sim \{-1,1\}^n$ and $\bA \sim \calD_1(\eps,\bx)$, as it is independent from $\bx$ and $\bA$. 
\begin{align*}
\Pr_{\substack{\bb, \bz, \bell, \bA \\ \bc_1, \dots, \bc_n}}[\calE_3 \wedge \ol{\calE_2} \wedge \ol{\calE_1}\wedge \ol{\calE_0}] 
&= \Pr_{\substack{\bb, \bz, \bell, \bA \\ \bc_1, \dots, \bc_n}}[\calE_3 \wedge \ol{\calE_2}\wedge \ol{\calE_1}\wedge \ol{\calE_0}] \\
&= \Pr_{\substack{\bb, \bz, \bell, \bA \\ \bc_1, \dots, \bc_n}}[\sgn(\sigma \by_{\bell}) \neq \bx_{\bell}\wedge \ol{\calE_2} \wedge \ol{\calE_0} ] \\
&= \Pr_{\substack{\bA,\bx,\bell}}[\sgn(\sigma \by_{\bell}) \neq \bx_{\bell}\wedge \ol{\calE_2}\wedge \ol{\calE_0} ]
\end{align*}
Now we can apply claim 1.5 using $\ol{\calE_0}$ and $\ol{\calE_2}$ to argue that $\vert \{i \in [n] \ : \  \sgn(\sigma \by_{i}) = \bx_{i}\} \vert \geq 0.98 n$. Let $\bD$ be the set of coordinates where $\{i \in [n] \ : \  \sgn(\sigma \by_{i}) \neq \bx_{i}\}$. Clearly $\vert\bD \vert = n -  \vert \{i \in [n] \ : \  \sgn(\sigma \by_{i}) = \bx_{i}\}\vert \leq 0.02n$. The conclusion then follows from noting that $\bell$ is drawn uniformly and independently from $\bA$ and $\bx$. 
\begin{align*}
\Pr_{\substack{\bA,\bx,\bell}}[\sgn(\sigma \by_{\bell}) \neq \bx_{\bell}\wedge \ol{\calE_2}\wedge \ol{\calE_0} ] 
&= \Pr_{\substack{\bA,\bx,\bell}}[\bell \in \bD  \wedge \ol{\calE_2} \wedge \ol{\calE_0}] \\
&\leq \E_{\bA, \bx}\Pr_{\substack{\bell}}[\bell \in \bD  \cdot \mathbbm{1} \left \{ \ol{\calE_2} \wedge \ol{\calE_0} \right \} ] \\
&\leq \E_{\bA, \bx}[\frac{|\bD|}{n}  \cdot \mathbbm{1} \left \{ \ol{\calE_2} \wedge \ol{\calE_0} \right \} ] \\
&\leq 0.02 
\end{align*}}
\end{proof}

\subsection{Putting Everything Together: Proof of Theorem~\ref{thm:main-lb-intro}}

For a fixed $\eps$ to be at least a large constant-factor of $1/\sqrt{n}$, we consider the distribution over input matrix $\bA \in \{-1,1\}^{n\times n}$ given by sampling $\bx \sim \{-1,1\}^n$ and then letting $\bA \sim \calD_1(\eps, \bx)$ from Definition~\ref{def:hard-distribution}. If, for $\eps_0 = \eps / 100$, there exists a randomized algorithm which outputs a vector whose Rayleigh quotient is within $\eps_0 n$ of maximum using at most $q$ queries, we may consider a deterministic algorithm (in the sense of Definition~\ref{def:det-alg}) satisfying (\ref{eq:guarantee-lb}) and using at most the same query complexity $q$. Then, the reduction in Figure~\ref{fig:reduction} in Subsection~\ref{sec:reduction} gives a procedure that given access to $n$ draws $\bc_1,\dots, \bc_n \sim \calC(\eps \bb)$ for $\bb \sim \{-1,1\}$, can output a random variable $\hat{\bb}$ which is guaranteed to agree with $\bb$ with probability at least $0.96$.

The crucial observation is the following:
\begin{itemize}
    \item In Figure~\ref{fig:reduction}, the query complexity to the row/column $\bell$ is guaranteed to be at most $200q/n$. 
    \item Furthermore, only entries of $\bA$ in row/column $\bell$ are, by the construction of $\bA$ (see Definition~\ref{def:coin-distribution}), dependent on $\bc_1,\dots, \bc_n$.
\end{itemize}
Thus, even though Figure~\ref{fig:reduction} has access to $n$ draws $\bc_1,\dots, \bc_n \sim \calC(\eps \bb)$, it only reads at most $200 q/n$ distinct values of $\bc_1,\dots, \bc_n$ on any execution. Since $\bc_1,\dots, \bc_n$ are all identically distributed, it suffices for the algorithm to read the first $s = 200q/n$ values and maintain a mapping $\pi \colon [s] \to [n]$ which re-assigns the $k$-th coin to $\bc_{\pi(k)}$. This results in a (randomized) function $\boldf \colon \{-1,1\}^{s} \to \{-1,1\}$ given by the output of Figure~\ref{fig:reduction} for solving Problem~\ref{prob:coin}, and therefore,
\[ \Omega(1/\eps^2) = s = \frac{200 q}{n} \qquad\Longrightarrow \qquad q = \Omega(n/\eps^2),\]
where $q$ is the query complexity of outputting an vector whose Rayleigh quotient is at least $\eps_0 n = \eps n /100$ from maximum, and thus, $q = \Omega(n/\eps_0^2)$.

\ignore{
\subsection{Learning $\hat{\bb}$ in few samples}
As described in \ref{alg:lb-reduction}, the reduction takes in a sequence of random variables $\bc_1,\dots \bc_n$, but instead one could equivalently considered the following random experiment, where the reduction receives a sample from $\calC(\eps \bb)$ when $\Alg(\bA,\eps)$ only when the queried a coordinate in $\bell$-th row or column.
\begin{framed}
\textbf{Experiment:}   Given accuracy parameter $\eps \in (\frac{1200}{\sqrt{n}},1)$, oracle $C$ which samples a $\bc \sim \calC(\eps \bb)$, and deterministic algorithm $\textsc{ALG}$ obtained from Lemma \ref{lem:yaos-lemma} 

\begin{enumerate}
\item Draw $\bell \sim [n]$ and $\bz\sim\{-1,1\}^{[n]\setminus \{\ell\}}$

    \item Run $\textsc{ALG}(\eps,\bA)$ while only drawing one entry of $\bA$ at a time. The algorithm also counts the number of queries made to the $\bell$-th row / column, and if it is ever greater than $1/(10\eps^2)$, output $\hat{\bb} \sim\{-1,1\}$.
    \begin{itemize}
        \item If $\bA_{i,i}$ is the query, set $\bA_{i,i}=1$ for any $i \in [n]$.
        \item Else if $\bA_{i,j}$ for $i=\bell$ or $j=\bell$ is the $k$-th query to the $\bell$-th row / column, use $C$ to sample $\bc_k\sim \calC(\eps\bb)$ and set $\bA_{\bell,j} = \bc_k \bz_j$ (with $\bA_{\bell,\bell} =1$).
        \item Else $i,j\neq \bell$ and draw $\bA_{i,j} =  \begin{cases}
            \bz_i \bz_j &\text{Else with } \Pr(2\eps)\\
        \bB_{i,j} &\text{Else with Pr. } \Pr(1-2\eps)\\
        \end{cases} $
    \end{itemize}

    \item Select $\sigma \in \{-1,1\}$ such that $\sum_{i\in [n]\setminus \{\bell\}} \sigma \by_i \cdot \bz_i >0 $
    
    \item  Output $\hat{\bb} = \sgn(\sigma \by_{\bell})$ as the estimate of $b$.
\end{enumerate}
\end{framed}

 Under this view, it's clear that this experiment makes at most $\frac{1}{10\eps^2}$ many draws from $\calC(\eps,\bb)$. If the distribution of $\bA$ is the same in this experiment and in the reduction \ref{alg:lb-reduction}, then it's clear one could have learned $\bb$ in fewer than $\frac{1}{10\eps^2}$ many draws from $\calC(\eps,\bb)$ violating the well known lower bound for learning the bias of a coin.
 
 \begin{lemma}\label{lem:low-query-complexity} For the output $\hat{\bb}$ of the experiment, $\Pr_{\substack{\bb, \bz, \bell, 
    \bA}}[\hat{\bb} = \bb ] \geq 0.9$     
 \end{lemma}
\begin{proof} Formally, we will prove the lemma by analyzing the distribution of $\bA$ generated by the experiment and show by the principle of deferred decisions, that it is the same as the distribution $\calD_2(\eps,\bb,\bell,\bz,\bc_1,\dots,\bc_n)$ for $\bc_1,\dots\bc_n \sim \calC(\eps \bb)$. Consider a fixed setting of $\bb \in {-1,1}$, $\bell \in [n]$, $\bz \in {-1,1}^{[n]\setminus \{\bell\}}$. Let $f_{t,k-1}(\Alg,\bell,\bb,\bz,\bc_1,\dots\bc_k)$ be the index of $t$-th query (of which $k-1$ are in the $\bell$-th row or column) made by $\Alg$. Now consider the execution of the experiment up to the $t$-th query, the randomness up to time $t$ is fixed, fixing the query made by the $\Alg$, we analyze the distribution of the queried entry of $\bA$
\begin{align*}
&\ \sum_{\substack{ \bell = \ell,\bz = z,\bb = b\\\bc_1,\dots,\bc_{k-1} = c_1,\dots,c_{k-1}}}\Pr\left [\bA_{\bell,i} = \bb z_i \ \wedge \ \substack{ f_{t,k-1}(\Alg,\bell,\bb,\bz,\bc_1,\dots \bc_{k-1}) = (\bell,i) \\ \bell = \ell,\bz = z,\bb = b\\\bc_1,\dots,\bc_{k-1} = c_1,\dots,c_{k-1}}\right ] \\
&= \sum_{\substack{ \bell = \ell,\bz = z,\bb = b\\\bc_1,\dots,\bc_{k-1} = c_1,\dots,c_{k-1}}}\Pr\left [\bc_k z_i = \bb z_i \ \wedge \ \substack{ f_{t,k-1}(\Alg,\bell,\bb,\bz,\bc_1,\dots \bc_{k-1}) = (\bell,i) \\ \bell = \ell,\bz = z,\bb = b\\\bc_1,\dots,\bc_{k-1} = c_1,\dots,c_{k-1}}\right ] \\
&=\sum_{\substack{ \bell = \ell,\bz = z,\bb = b\\\bc_1,\dots,\bc_{k-1} = c_1,\dots,c_{k-1}}}\Pr_{\bc_k}\left [\bc_k  = \bb  \big\vert \substack{ f_{t,k-1}(\Alg,\bell,\bb,\bz,\bc_1,\dots \bc_{k-1}) = (\bell,i) \\ \bell = \ell,\bz = z,\bb = b\\\bc_1,\dots,\bc_{k-1} = c_1,\dots,c_{k-1}}\right ] \cdot \Pr\left [\substack{ f_{t,k-1}(\Alg,\bell,\bb,\bz,\bc_1,\dots \bc_{k-1}) = (\bell,i) \\ \bell = \ell,\bz = z,\bb = b\\\bc_1,\dots,\bc_{k-1} = c_1,\dots,c_{k-1}}\right ] 
\end{align*}
Now using the independence of the sample $\bc_k\sim \calC(\eps,\bb)$ from the execution of the algorithm up to time step $t$, we get .
\begin{align*}
&=\sum_{\substack{ \bell = \ell,\bz = z,\bb = b\\\bc_1,\dots,\bc_{k-1} = c_1,\dots,c_{k-1}}}\Pr_{\bc_k}\left [\bc_k  = \bb \right ] \cdot \Pr\left [\substack{ f_{t,k-1}(\Alg,\bell,\bb,\bz,\bc_1,\dots \bc_{k-1}) = (\bell,i) \\ \bell = \ell,\bz = z,\bb = b\\\bc_1,\dots,\bc_{k-1} = c_1,\dots,c_{k-1}}\right ] \\
&=\sum_{\substack{ \bell = \ell,\bz = z,\bb = b\\\bc_1,\dots,\bc_{k-1} = c_1,\dots,c_{k-1}}} \left ( \frac{1}{2}+\eps \right )\cdot \Pr\left [\substack{ f_{t,k-1}(\Alg,\bell,\bb,\bz,\bc_1,\dots \bc_{k-1}) = (\bell,i) \\ \bell = \ell,\bz = z,\bb = b\\\bc_1,\dots,\bc_{k-1} = c_1,\dots,c_{k-1}}\right ] \\
&= \left (\frac{1}{2}+\eps \right ) \\
&= \Pr_{\substack{\bell,\bb,\bz,\bA \\ \bc_1,\dots\bc_n}}\left [\bA_{\bell,i} = \bb \bz_i\right ]
\end{align*}
We repeat the argument for the case where $f_{t,k-1}(\Alg,\bell,\bb,\bz,\bc_1,\dots \bc_{k-1})$ queries $\bA_{i,j}$ for $i,j\neq \bell$.
\begin{align*}
&\ \sum_{\substack{ \bell = \ell,\bz = z,\bb = b\\\bc_1,\dots,\bc_{k-1} = c_1,\dots,c_{k-1}}}\Pr\left [\bA_{i,j} = z_i z_j \ \wedge \ \substack{ f_{t,k-1}(\Alg,\bell,\bb,\bz,\bc_1,\dots \bc_{k-1}) = (i,j) \\ \bell = \ell,\bz = z,\bb = b\\\bc_1,\dots,\bc_{k-1} = c_1,\dots,c_{k-1}}\right ] \\
&=\sum_{\substack{ \bell = \ell,\bz = z,\bb = b\\\bc_1,\dots,\bc_{k-1} = c_1,\dots,c_{k-1}}}\Pr_{\bB_{i,j}}\left [\bA_{i,j}  = z_i z_j  \big\vert \substack{ f_{t,k-1}(\Alg,\bell,\bb,\bz,\bc_1,\dots \bc_{k-1}) = (i,j) \\ \bell = \ell,\bz = z,\bb = b\\\bc_1,\dots,\bc_{k-1} = c_1,\dots,c_{k-1}}\right ] \cdot \Pr\left [\substack{ f_{t,k-1}(\Alg,\bell,\bb,\bz,\bc_1,\dots \bc_{k-1}) = (i,j) \\ \bell = \ell,\bz = z,\bb = b\\\bc_1,\dots,\bc_{k-1} = c_1,\dots,c_{k-1}}\right ] \\
&=\sum_{\substack{ \bell = \ell,\bz = z,\bb = b\\\bc_1,\dots,\bc_{k-1} = c_1,\dots,c_{k-1}}}\left (2\eps+ (1-2\eps)\Pr_{\bB_{i,j}}\left [\bB_{i,j}  = z_i z_j  \vert \substack{ f_{t,k-1}(\Alg,\bell,\bb,\bz,\bc_1,\dots \bc_{k-1}) = (i,j) \\ \bell = \ell,\bz = z,\bb = b\\\bc_1,\dots,\bc_{k-1} = c_1,\dots,c_{k-1}}\right ] \right )\cdot \Pr\left [\substack{ f_{t,k-1}(\Alg,\bell,\bb,\bz,\bc_1,\dots \bc_{k-1}) = (i,j) \\ \bell = \ell,\bz = z,\bb = b\\\bc_1,\dots,\bc_{k-1} = c_1,\dots,c_{k-1}}\right ] \\
&=\sum_{\substack{ \bell = \ell,\bz = z,\bb = b\\\bc_1,\dots,\bc_{k-1} = c_1,\dots,c_{k-1}}} \left ( \frac{1}{2}+\eps \right )\cdot \Pr\left [\substack{ f_{t,k-1}(\Alg,\bell,\bb,\bz,\bc_1,\dots \bc_{k-1}) = (i,j) \\ \bell = \ell,\bz = z,\bb = b\\\bc_1,\dots,\bc_{k-1} = c_1,\dots,c_{k-1}}\right ] \\
&= \left (\frac{1}{2}+\eps \right ) = \Pr_{\substack{\bell,\bb,\bz,\bA \\ \bc_1,\dots\bc_n}}\left [\bA_{i,j} = \bz_i \bz_j\right ]
\end{align*}
The final case of $f_{t,k-1}(\Alg,\bell,\bb,\bz,\bc_1,\dots \bc_{k-1}) = (i,i)$ results in a fixed draw of $\bA_{i,i}=1$. Since $\bA_{i,j}$ is supported on $\{-1,1\}$, we can conclude that the distribution of $\bA_{i,j}$ generated by the experiment at step $t$ conditioned on the prior draw of $\bc_1,\dots\bc_{k-1}$ is equivalent the distribution of $\bA_{i,j}$ generated by $\calD_2(\eps,\bb,\bell,\bz,\bc_1,\dots,\bc_n)$ for $\bc_1,\dots\bc_n \sim \calC(\eps \bb)$. By the distributional equivalence of the two random processes and lemma \ref{lem:lb-main-reduction} that the probability that $\hat{\bb}=\bb$ is at least $0.9$, the proof of the lemma is concluded:
\begin{align*}
    \Pr_{\bb, \bz, \bell, 
    \bA}[\hat{\bb} = \bb] =  \Pr_{\substack{\bb, \bz, \bell, 
    \bA\\ \bc_1,\dots\bc_n}}[\hat{\bb} = \bb] \geq 0.9
\end{align*}
\end{proof}

By construction, the experiment makes no more than $1/10\eps^2$ queries to the $\bell$-th row or column of $\bA$, meaning it uses at most $1/10\eps^2$ draws from $\calC(\eps,\bb)$. This contradicts the well known lower bound for learning the bias of a random coin, concluding the proof of Theorem \ref{thm:general-lower-bound}  
}

\section{Local Computation Algorithm for Sparsest Cut}\label{sec:sparse-cut}

This section gives an application of Theorem~\ref{thm:main-intro} to providing local access to sparse cuts in dense graphs. We do this by giving a local computation algorithm which can execute the constructive proof of Cheeger's inequality~\cite{AM85}. In particular, we prove theorem quoted in the introduction (restated below -- see Definition~\ref{def:sparsity}).

\localcheeger*

\subsection{Reduction to $(\alpha, \beta)$-Dense Graphs, for $\alpha, \beta = \poly(\eps)$: Proof of Theorem~\ref{thm:cheeger-intro}}

In this section, we show that it suffices to construct a local computation algorithm for certain dense graphs (i.e., graphs where most of the vertices have a high degree). This will be important for our use of Theorem~\ref{thm:main-intro}, since we will need the normalized adjacency matrix to be a bounded-entry matrix. For this section, consider a fixed input graph $G = ([n], E)$ which we provide query access to, as well as a fixed accuracy parameter $\eps > 0$ (notice, we may assume that $G$ contains at least $\eps n^2$ edges, as otherwise, $\phi_{\eps}(G)$ is trivially $1$). The goal will be to obtain the conclusions of Theorem~\ref{thm:main-intro} while assuming a local computation algorithm for certain ``dense'' graphs (Definition~\ref{def:dense-graph} below). 

\begin{definition}\label{def:dense-graph}
For $\alpha, \beta \in (0, 1)$, a graph is $(\alpha,\beta)$-dense if at most $\beta$-fraction of vertices are connected to fewer than $\alpha$-fraction of vertices.
\end{definition}

\begin{lemma}\label{lem:noise-addition}
For any $\alpha, \beta \in (\eps^{c}, \eps / c)$ for large enough constant $c$. Let $\bN \sim \calG(n, 4\alpha)$ be an Erd\"{o}s-Reyni graph and let $\bH = G \cup \bN$. The following occurs with high probability over $\bN$:
\begin{itemize}
\item For any $T \subset [n]$ where $\vol_{\bH}(T) \leq \eps n^2/2$, let $\bH'$ be the subgraph of $\bH$ induced by $\ol{T}$. Then, $\bH'$ is $(\alpha, \beta)$-dense and satisfies $\phi_{\eps/2}(\bH') \leq 2\phi_{\eps}(G) + O(\alpha / \eps)$.
\item Any $S \subset [n]$ with $\eps n^2/8 \leq \vol_{\bH}(S) \leq \vol_{\bH}(\ol{S})$ will satisfy $\phi_{G}(S) \leq \phi_{\bH}(S) + O(\alpha/\eps)$.
\end{itemize}
\end{lemma}

\begin{proof}
Notice that the number of edges in $\bH$ is at least $\eps n^2$ (it contains all edges of $G$, which is size at least $\eps n^2$), and since there are at most $\eps n^2/2$ edges incident on $T$ in $\bH$, $\vol_{\bH}(\ol{T}) \geq \eps n^2 / 2$. Degrees are at most $n$, so $\ol{T}$ is of size at least $\eps n/2$. In order to show that $\bH'$ is $(\alpha,\beta)$-dense, it suffices to show that any induced subgraph of $\bN$ containing at least $\eps n / 2$ vertices will be $(\alpha, \beta)$-dense (as $\bH'$ can only include more edges than any such subgraph). We show this as a standard application of Chernoff bounds: consider any fixed subset $L \subset [n]$ of size $t \geq \eps n /2$ and any fixed subset $B \subset L$ of size $\beta t$. We first upper bound the probability that the set $B$ consists of vertices of degree at most $\alpha t$ within the subgraph induced by $L$:
\begin{align*}
\Prx_{\bG'}\left[ \forall j \in B : \sum_{i \in L \setminus \{ j \}} \ind\left\{ (i,j) \in \bN \right\} \leq \alpha t \right] \leq \left( \Prx\left[ \sum_{i=1}^{t-1} \bX_i \leq \alpha t\right] \right)^{|B|} \leq e^{-\Omega(\alpha t |B|)},
\end{align*}
where since $\alpha t$ is a constant-factor less than half of the expected number of edges incident on any particular vertex $j \in L$ in the induced subgraph of $\bN$.\footnote{We take half of the edges as a convenience, to de-correlate the degrees of $B$.} We take a union bound over at most $2^{2n}$ choices of $L$ and $B$, since $\alpha t |B| \geq \alpha \beta t^2 \geq \alpha \beta \eps^2 n^2$ is significantly larger than $n$ for our setting of $\eps$. This allows us to conclude the union bound.

In order to upper bound $\phi_{\eps/2}(\bH')$, consider the subset $S \subset [n]$ which realizes the bound on $\phi_{\eps}(G)$ and satisfies $\eps n^2 \leq \vol_G(S) \leq \vol_{G}(\ol{S})$, and let $S' = S \cap \ol{T}$. We upper bound $\phi_{\bH'}(S')$ in terms of $\phi_{G}(S)$. The fact $\vol_{\bH}(T) \leq \eps n^2/2$ implies 
\[ \vol_{\bH'}(S') \geq \vol_{\bH}(S) - \vol_{\bH}(T) \geq \vol_{G}(S) - \eps n^2/2 \geq \vol_{G}(S)/2, \]
which is at least $\eps n^2/2$, and this lower bounds the denominator of $\phi_{\bH'}(S')$. Then, the fact $S'$ is a subset of $S$ and $\ol{T} \setminus S'$ is a subset of $\ol{S}$ implies that $E_{\bH'}(S', \ol{T} \setminus S')$ (which is the numerator of $\phi_{\bH'}(S')$) is at most $E_{\bH}(S, \ol{S}) \leq E_{G}(S, \ol{S}) + 8\alpha n^2$, with high probability (since $\bN$ contains at most $8\alpha n^2$ edges with high probability). This implies $\phi_{\eps/2}(\bH') \leq 2\phi_{\eps}(G) + 16\alpha /\eps$.

For the second item, any cut $S \subset [n]$ with $\eps n^2/8 \leq \vol_{\bH}(S) \leq \vol_{\bH}(\ol{S})$ must necessarily satisfy $\min\{ \vol_{G}(S), \vol_G(\ol{S}) \} \geq (\eps/8 - 8\alpha) n^2$ (once, more since $\bN$ contains at most $8\alpha n^2$ edges with high probability). Thus, 
\begin{align*}
\phi_G(S) = \dfrac{E_{G}(S, \ol{S})}{\min\{ \vol_{G}(S), \vol_{G}(\ol{S}) \}} \leq \dfrac{E_{\bH}(S, \ol{S})}{\min\{ \vol_{G}(S), \vol_{G}(\ol{S})\}} &\leq \dfrac{E_{\bH}(S, \ol{S})}{(1 - 64\alpha /\eps) \min\{ \vol_{\bH}(S), \vol_{\bH}(\ol{S}) \}} \\
	&\leq \phi_{\bH}(S) + O(\alpha /\eps)
\end{align*}
as desired.
\end{proof}

\newcommand{\EstimateVol}{\texttt{Estimate-Vol}}

We will assume there exists a local computation algorithm $\calA$ which finds a sparse cut assuming the underlying graph is $(\alpha,\beta)$-dense (Theorem~\ref{thm:cheeger-dense}). Using Lemma~\ref{lem:noise-addition}, we analyze the algorithm in Figure~\ref{fig:prepro-cheeger} to prove Theorem~\ref{thm:cheeger-intro}, which proceeds by iteratively cutting out sparse cuts and analyzing the induced subgraph of remaining vertices. Figure~\ref{fig:prepro-cheeger} assumes access to a subroutine $\EstimateVol(S, H, \xi, \delta)$ which receives query access of the form ``$i \in S$'', as well as query access to the adjacency matrix of a graph $H$ on vertex set $[n]$ and outputs an estimate $\hat{\boldeta}$ of $\vol_{H}(S)$ which is off by at most $\xi n^2$ with probability $1-\delta$. Such an algorithm $\EstimateVol(S, H, \xi, \delta)$ simply works via random sampling, and has query complexity which is a factor of $O(\log(1/\delta) / \xi^2)$ overhead on the query complexity of determining ``$i \in S$,'' and at most $O(\log(1/\delta)/\xi^2)$ queries to the adjacency matrix of $H$.

\begin{theorem}\label{thm:cheeger-dense}
There exists a (randomized) local computation algorithm $\calA$ with query access to the adjacency matrix $G = (V, E)$, an accuracy parameter $\eps > 0$, and a failure probability $\delta > 0$. 
\begin{itemize}
\item $\calA$ draws a subset $\bQ \subset V$ of size $\poly(\log(1/\delta) / (\alpha \beta \eps))$ uniformly at random, and queries the induced subgraph of $\bQ$. On query $i \in V$, $\calA$ queries entries $(i, j)$ for all $j \in \bQ$.  
\item Whenever $G$ is $(\alpha,\beta)$-dense for $\alpha, \beta \leq \eps^c$ for sufficiently large $c$, letting $\bS \subset V$ be the underlying cut that $\calA$ generates,
\begin{align*}
\Prx_{\calA}\left[ \begin{array}{c} \phi_{G}(\bS) \leq O(\sqrt{\phi_{\eps}(G)}) + \eps \\ \alpha^3 n^2 \lsim \vol_{G}(\bS) \leq \vol_G(\ol{\bS}) \end{array} \right] \geq 1-\delta.
\end{align*}
\end{itemize}
\end{theorem}

\begin{figure}[ht!]
\begin{framed}
\textbf{Local Computation Algorithm for Sparsest Cut.} The algorithm receives query access to the adjacency matrix of a graph $G = ([n], E)$, an accuracy parameter $\eps > 0$, and an index $i \in [n]$.
\begin{enumerate}
\item Use public randomness to generate $\bN \sim \calG(n, 4\alpha)$, and use query access to the adjacency matrix of $G$ to provide query access to the adjacency matrix of $\bH_0 = G \cup \bN$.
\item Initialize the following objects:
\begin{itemize}
\item A collection $\sfA$ (initially empty) of at most $t$ local computation algorithms ($\calA_1, \dots, \calA_t$) which provide query access to subsets of $[n]$ ($\bS$ will be the union of the first $\ell$ sets, for some $\ell \leq t$).
\item Using public randomness, $t$ independent draws $\bQ_1^{(0)}, \dots, \bQ_t^{(0)} \subset [n]$ where $j \in \bQ_{\ell}^{(0)}$ i.i.d. with probability $2\gamma$  for $\gamma \geq \poly(1/\eps)/n$ \footnotemark.
\end{itemize}
\item\label{ln:vol-check} For $\ell=1, \dots t$, repeat the following:
\begin{enumerate}
\item\label{ln:estimate-vol} Let $\bT_{\ell-1} \subset [n]$ denote the (implicit) set given by\footnotemark
\[ \bT_{\ell-1} = \left\{ j \in [n] : \exists \ell' \leq \ell - 1 \text{ s.t. $\calA_{\ell'}$ declares $j$ lies in the set } \right\}. \]
Execute $\EstimateVol(\bT_{\ell-1}, \bH_0, \eps/8, 1/(400 t))$ and let $\hat{\boldeta}$ be the output. If $\hat{\boldeta} \geq \eps n^2/4$, break out of the loop to Line~\ref{ln:default}. 
\item Otherwise, $\hat{\boldeta} < \eps n^2/4$ and continue. Let $\bV_{\ell} = [n] \setminus \bT_{\ell-1}$, which is the (implicit) vertex set obtained after removing the union of sets defined by $\calA_{1},\dots, \calA_{\ell-1}$. 
\item For every $\ell' \leq \ell-1$, query the local computation algorithm $\calA_{\ell'}$ with $j$ for $j \in \bQ^{(0)}_{\ell}$ in order to write down $\bQ_{\ell} = \bQ^{(0)}_{\ell} \cap \bV_{\ell}$. Note, $\bQ_{\ell}$ includes each $j \in \bV_{\ell}$ i.i.d. w.p $2\gamma$.
\item\label{ln:declare-output} Let $\bH_{\ell}$ be the subgraph of $\bH_0$ induced by $\bV_{\ell}$, and let $\calA_{\ell}$ be the local computation algorithm from Theorem~\ref{thm:cheeger-dense} with parameter $\eps' = \eps/2$, which uses query set $\bQ_{\ell} \subset \bV_{\ell}$. Query $\calA_{\ell}$ with $i$ (which must lie in $\bV_{\ell}$), and if $\calA_{\ell}$ declares that $i$ lies in the set, output ``$i \in \bS$'' and terminate.
\end{enumerate}
\item\label{ln:default} If the procedure has not produced an output yet, output ``$i \notin \bS$.''
\end{enumerate}
\end{framed}
\caption{Local Computation Algorithm for Sparsest Cut}\label{fig:prepro-cheeger}
\end{figure}

\paragraph{Correctness Guarantees.} We first establish the correctness guarantees of Theorem~\ref{thm:cheeger-intro}, where we show that the local computation algorithm specified in Figure~\ref{fig:prepro-cheeger} provides query access to a cut $\bS \subset [n]$ which satisfies $\phi_G(\bS)$ is at most $O(\sqrt{\phi_{\eps}(G)}) + \eps$ with probability at least $0.99$. By setting $\delta > \frac{1}{400 t}$, we assume that all executions of $\EstimateVol$, and uses of Lemma~\ref{lem:noise-addition} and Theorem~\ref{thm:cheeger-dense} result in conclusions which succeed; each event occurs with high enough probability, so by a union bound, all events occur with probability of at least $0.99$.  In what follows, we assume this is the case. Notice that Figure~\ref{fig:prepro-cheeger} defines a local computation algorithm which, on query $i \in [n]$ will also produce an output (either that ``$i \in \bS$'' in some execution of Line~\ref{ln:estimate-vol}, or that ``$i \notin \bS$'' in Line~\ref{ln:default}). Hence, we may define
\[ \bS = \left\{ i \in [n] : \text{ Figure~\ref{fig:prepro-cheeger} on input $G$, $\eps$, and $i \in [n]$ outputs ``$i \in \bS$''}\right\}, \]
and notice that $\bS$ naturally decomposes as a union $\bS_1 \cup \dots \cup \bS_{\bell_0-1}$, where $\bS_{\ell}$ consists of the indices $i \in [n]$ where $\calA_{\ell}$ in Line~\ref{ln:declare-output} outputs ``$i \in \bS$'' and $\ell$ is the smallest to do so. Note, $\bell_0$ represents the iteration of Line~\ref{ln:estimate-vol} which triggered the ``break'' condition, and we default $\bell_0 = t+1$ if Line~\ref{ln:estimate-vol} is never triggered (we will show this does not happen). Note, the following inductive definitions are a consequence of the iterative algorithm in Figure~\ref{fig:prepro-cheeger}:
\begin{itemize}
    \item $\bT_0 = \emptyset$.
    \item For $\ell=1, \dots, \bell_0-1$, $\bS_{\ell}$ is the (implicit) subset defined by $\calA_{\ell}$ executed on input graph $\bH_{\ell}$, which is the subgraph of $\bH_0$ induced by $\bV_{\ell} = [n] \setminus \bT_{\ell-1}$.
    \item For $\ell=2, \dots,\bell_0$, we have $\bT_{\ell-1} = \bS_1 \cup \dots \cup \bS_{\ell-1}$.
\end{itemize}
Notice that, for every $\ell = 1, \dots, \bell_0-1$ executed in Line~\ref{ln:vol-check} will have $\bT_{\ell-1}$ with $\vol_{\bH_0}(\bT_{\ell-1}) \leq \eps n^2/2$. This is because $\ell < \bell_0$, Line~\ref{ln:estimate-vol} results in $\hat{\boldeta} \leq \eps n^2 / 4$, and the fact that $\EstimateVol(\bT_{\ell-1}, \bH_0, \eps/8,\delta)$ is an $\eps n^2 / 8$-error approximation of $\vol_{\bH_0}(\bT_{\ell-1})$. Hence, we use the first item of Lemma~\ref{lem:noise-addition} and conclude the subgraph $\bH_{\ell}$ of $\bH_0$ induced by vertex set $\bV_{\ell}$ is $(\alpha, \beta)$-dense, and $\phi_{\eps/2}(\bH_{\ell}) \leq 2 \phi_{\eps}(G) + O(\alpha/\eps)$. Similarly to the proof of Lemma~\ref{lem:noise-addition}, $\bV_{\ell}$ must have size at least $(\eps/2)^{1/2} n$, since $\vol_{\bH_0}(\bV_{\ell}) \geq \eps n^2 - \vol_{\bH_0}(\bT_{\ell-1}) \geq \eps n^2/2$ (recall, $G$ contains at least $\eps n^2$ edges). As a result, the set $\bQ_{\ell} \subset \bV_{\ell}$, which includes each index $j \in \bV_{\ell}$ i.i.d. with probability $\gamma'$ satisfies $\gamma' \geq \poly(1/\eps)/n \geq \poly(1/\eps) / |\bV_{\ell}|$, and we may Theorem~\ref{thm:cheeger-dense}. The local computation algorithm $\calA_{\ell}$ in Line~\ref{ln:declare-output} will satisfy,
\begin{align} 
\phi_{\bH_{\ell}}(\bS_{\ell}) \leq O\left(\sqrt{\phi_{\eps/2}(\bH_{\ell})}\right) + \eps, \qquad \text{and}\qquad \alpha n^2 \leq \vol_{\bH_{\ell}}(\bS_{\ell}) \leq \vol_{\bH_{\ell}}(\ol{\bS}_{\ell}).  \label{eq:thm-dense-bounds}
\end{align}
\addtocounter{footnote}{-1}\footnotetext{Note that Theorem \ref{thm:main} samples a subset of a fixed size, while other sub-routines such as in Theorem \ref{thm:cheeger-dense} or $\EstimateVol(S, H, \xi, \delta)$ sample entries independently with fixed probability. These sampling methods are interchangeable, so it suffices to define query sets $\bQ_1,\dots \bQ_t$ in terms of either sampling model.}
\addtocounter{footnote}{+1}\footnotetext{We let $\bT_{0} = \emptyset$, and note that $\bT_{\ell-1}$ is implicit, since we can query $\calA_{1}, \dots, \calA_{\ell-1}$ in order to know whether $i \in \bT_{\ell-1}$.}

Therefore, we count the edges of $E_{\bH_0}(\bS,\ol{\bS})$ using the decomposition of $\bS = \bS_1 \cup \dots \cup \bS_{\bell_0-1}$ and applying the above bounds,
\begin{align*}
    E_{\bH_0}(\bS, \ol{\bS}) \leq \sum_{\ell=1}^{\bell_0-1} E_{\bH_{\ell}}(\bS_{\ell}, \bV_{\ell} \setminus \bS_{\ell}) &\leq \sum_{\ell=1}^{\bell_0-1}\left(O(\sqrt{\phi_{\eps/2}(\bH_{\ell}}) + \eps \right) \cdot \vol_{\bH_{\ell}}(\bS_{\ell}) \\
    &\leq \left( O\left(\sqrt{\phi_{\eps}(G)} \right) + O(\alpha/\eps) + \eps \right) \cdot \vol_{\bH_0}(\bS),
\end{align*}
where the first inequality uses the fact $\bH_{\ell}$ consist of induced subgraphs which iteratively remove $\bS_{\ell}$, the second inequality uses the bound on $\phi_{\bH_{\ell}}(\bS_{\ell})$, and the third inequality uses the conclusion of Lemma~\ref{lem:noise-addition} to relate to $\phi_{\eps}(G)$, and the fact volumes with respect to induced subgraphs are at most that of the entire graph. 

We now relate $\phi_{\bH_0}(\bS)$ to $\phi_G(\bS)$, by using the second condition of Theorem~\ref{lem:noise-addition}. For that, we need a lower bound on $\vol_{\bH_0}(\bS)$. For that, notice that if $\bell_0 < t$ (and the condition of Line~\ref{ln:estimate-vol} was triggered), then it was because $\hat{\boldeta} \geq \eps n^2 / 4$, which implies $\vol_{\bH_0}(\bT_{\bell_0-1}) = \vol_{\bH_0}(\bS)$ is at least $\eps n^2/8$. If such was the case, we apply the second item of Lemma~\ref{lem:noise-addition} as desired. In order to show that $\bell_0 < t$, and the condition of Line~\ref{ln:estimate-vol} eventually gets triggered, we use the fact the right-most inequality of (\ref{eq:thm-dense-bounds}), which implies that $\vol_{\bH_0}(\bS)$ is at least $(\bell_0 - 1) \cdot \alpha n^2$; once $\bell_0 - 1 \geq 1 / \alpha$, $\vol_{\bH_0}(\bS)$ would have constant volume and trigger Line~\ref{ln:estimate-vol}.

\paragraph{Query Complexity.} We now bound the query complexity of Figure~\ref{fig:prepro-cheeger}. First, note that a single query to the adjacency matrix of $\bH_0$ uses at most one query to the adjacency matrix of $G$. For any particular local computation algorithm $\calA_{\ell}$ initialized in Line~\ref{ln:declare-output}, Figure~\ref{fig:prepro-cheeger} makes:
\begin{itemize}
    \item Determine which $\bQ_{\ell}^{(0)} \subset [n]$ are part of the induced subgraph $\bH_{\ell}$, which requires making $|\bQ_{\ell}^{(0)}|$ queries to the (previous) local computation algorithms $\calA_1, \dots, \calA_{\ell-1}$. This is used to determine the set $\bQ_{\ell}$.
    \item Once we know $\bQ_{\ell}$, $\calA_{\ell}$ makes $|\bQ_{\ell}|^2$ queries during the preprocessing time, and $|\bQ_{\ell}|$ queries during any particular query. 
\end{itemize}
As a result, the total number of probes needed is:
\[ \sum_{\ell=1}^t |\bQ_{\ell}|^2 + \sum_{\ell=1}^t \sum_{\ell' < \ell} |\bQ_{\ell}^{(0)}| \cdot |\bQ_{\ell'}|, \]
where the first term constitutes all of the preprocessing queries (done after determining $\bQ_{\ell}$ for each $\ell$), plus the $|\bQ_{\ell}^{(0)}|$ queries which are done to $\calA_{\ell'}$ for $\ell' < \ell$, each of which uses $|\bQ_{\ell'}|$ probes to the adjacency matrix of $G$. The fact that all $\bQ_{\ell} \subset \bQ_{\ell}^{(0)}$ and all are of size $\poly(1/\eps)$, and the fact $t$ is at most $O(1/\alpha)$ gives the desired upper bound.

\subsection{Sparsest Cut on $(\alpha, \beta)$-dense Graphs with $\alpha,\beta = \poly(\eps)$: Proof of Theorem~\ref{thm:cheeger-dense}}

Fix an input graph $G = (V, E)$ which is $(\alpha,\beta)$-dense and has at least $\eps n^2$ edges.  The algorithm of Theorem~\ref{thm:cheeger-dense} will proceed by providing query access to a deflated normalized adjacency matrix $M$ (see Definition~\ref{def:deflated} below). As we discuss, we will execute our top eigenvector local computation algorithm from Theorem~\ref{thm:main-intro} on an input matrix which is an approximation to a deflated normalized adjacency matrix, in order to access a vector which is approximates the second-top eigenvector of the normalized adjacency.

\begin{definition}\label{def:deflated}
    For a graph $G = ([n], E)$, the deflated normalized adjacency matrix $M$ is the $n \times n$ matrix whose  $(i,j)$-entry is 
    \[ M_{ij} = \dfrac{\ind\{ (i,j) \in E \}}{\sqrt{\deg_G(i) \cdot \deg_G(j)}} - \dfrac{\sqrt{\deg_G(i) \cdot \deg_G(j)}}{2m}. \]
\end{definition}

The matrix $M$ is an ``exact'' deflated normalized adjacency (where we take the normalized adjacency matrix and remove the top eigenvector). Note, a local computation algorithm (whose query complexity is $\poly(1/\eps)$) cannot provide exact query access to $M$, since knowing $m$ (the number of edges), and $\deg_G(i)$ for any vertex exactly would require too many queries. Hence, we provide access to an approximation of $M$ in the following sense, where we will let $\xi$ be a parameter indicating the sampling error which will be set to a small enough polynomial of $\eps$:
\begin{itemize}
    \item \textbf{Estimate of Edges}: The algorithm executes a sampling-based estimation, to obtain an estimate $\hat{m}$ of $m$ which satisfies $m \leq \hat{m} \leq (1+\xi) m$. We use the fact $m \geq \eps n^2$ to obtain a multiplicative error estimate, using $\poly(1/(\eps \xi))$ queries to the adjacency matrix of $G$.
    \item \textbf{Degree Estimation}: Similarly to above, a sampling-based estimation of the degrees $d(i)$ for $\deg_G(i)$. The fact $G$ is $(\alpha,\beta)$-dense allows us to, (i) enforce $\alpha n \leq d(i)$ for all $i \in [n]$, and there will be at most $\beta n$ vertices where $\deg_G(i) \leq \alpha n$, and (ii) except for at most $\xi n$ vertices, the estimate $d(i)$ satisfies $\max\{ \alpha n, \deg_G(i)\} \leq d(i) \leq (1+\xi) \max\{\deg_G(i) , \alpha n\}$.\footnote{Importantly, we cannot hope to estimate the degrees of all vertices correctly, since that would require an ``union bound'' and incur $O(\log n)$-dependencies.}
\end{itemize}
While technically random variables, $\hat{m}$ and $d(i)$ are kept un-bolded, and we assume they satisfy above properties. Furthermore, note the above sampling is made non-adaptive (requiring uniform random draws from $[n]$), and incurs $\poly(1/(\xi\eps\alpha))$-additive overhead on the size of $\bQ$ (as specified in Theorem~\ref{thm:cheeger-dense}). With that, we assume the algorithm operates by receiving query access to the entries of the matrix $\hat{M}$:
\begin{align}
    \hat{M}_{ij} = \dfrac{\ind\{ (i,j) \in E \}}{\sqrt{d(i) \cdot d(j)}} - \dfrac{\sqrt{d(i) \cdot d(j)}}{2\hat{m}}, \label{eq:noisy-deflate}
\end{align}
which we will refer to as a noisy deflated normalized adjacency of $G$. 

\begin{lemma}\label{lem:feasibility}
    Let $\hat{M}$ be a noisy deflated normalized adjacency of an $(\alpha,\beta)$-dense graph $G = ([n], E)$ (as in (\ref{eq:noisy-deflate})). Then, $\| \hat{M} \|_{\infty} \leq 2/(\alpha n)$ and $\lambda_{\max}(\hat{M}) \geq 1 - 2\phi_{\eps}(G) - O((\xi + \alpha \beta)/\eps)$.
\end{lemma}

\begin{proof}
    For the first item, notice that every $i \in [n]$ satisfies $\alpha n \leq d(i) \leq n$, where the lower bound is by definition of (\ref{eq:noisy-deflate}), and we always have $\hat{m} \geq \eps n^2$ (since we know $G$ contains at least $\eps n^2$ edges). Hence, each entry is at most $1/(\alpha n) + 1/(\eps n)$, which is at most $2/(\alpha n)$ by our setting of $\alpha$. For the lower bound on $\lambda_{\max}(\hat{M})$, consider a cut $S \subset [n]$ which realizes the bound on $\phi_{\eps}(G)$. We let $v(S) = \sum_{i \in S} d(i)$ and $v(\ol{S}) = \sum_{i \in \ol{S}} d(i)$ (which we note would be the volumes under $G$ if all degree estimates were perfect), and $y \in \R^n$ be the vector where $y_i = v(\ol{S})$ if $i \in S$ and $y_i = -v(S)$ if $i \in \ol{S}$. Consider the vector $x \in \R^n$ which sets $x_i = y_i \cdot \sqrt{d(i)}$. Then,
    \begin{align*}
        \| x\|_2^2 &= \sum_{i\in S} v(\ol{S})^2 \cdot d(i) + \sum_{i \in \ol{S}} v(S)^2 \cdot d(i) = v(S) v(\ol{S}) \left( v(S) + v(\ol{S})\right).
    \end{align*}    
    In order to lower bound $\langle x, \hat{M} x\rangle$, we note that the setting of $x$ implies that the second term of $\hat{M}$ (that which ``deflates'' the normalized adjacency) satisfies $1/(2\hat{m})\sum_{i,j\in [n]} \sqrt{d(i) \cdot d(j)} \cdot x_i x_j = 1/(2\hat{m}) (\sum_{i\in[n]} d(i) y_i)^2 = 0$. Hence, we have
    \begin{align*}
        \langle x, \hat{M} x \rangle &= \sum_{i=1}^n \sum_{j=1}^n \ind\left\{ (i,j) \in E \right\} \cdot y_i y_j \\
        &= v(\ol{S})^2 \sum_{i,j \in S} \ind\{ (i,j) \in E\} + v(S)^2 \sum_{i,j \in \ol{S}} \ind\{ (i,j) \in E\} - 2 v(S) v(\ol{S}) \sum_{i \in S, j \in \ol{S}} \ind\{ (i,j) \in E\}\\
            &= v(\ol{S})^2 \cdot \vol_G(S) + v(S)^2 \cdot \vol_G(\ol{S}) - E_G(S, \ol{S}) (v(S) + v(\ol{S}))^2            
    \end{align*}
    Hence, the Rayleigh quotient $\langle x, \hat{M} x \rangle / \|x\|_2^2$ is at least:
    \begin{align*}
        \dfrac{v(\ol{S}) \cdot \vol_G(S)}{v(S) (v(S) + v(\ol{S}))} + \dfrac{v(S) \cdot \vol_G(\ol{S})}{v(\ol{S}) (v(S) + v(\ol{S}))} - \dfrac{E_G(S, \ol{S}) \cdot (v(S) + v(\ol{S}))}{v(S) v(\ol{S})}.
    \end{align*}   
    It remains to re-write the above in terms of the conductance $\phi_{\eps}(G)$, using the fact that $v(S)$ and $v(\ol{S})$ will be close enough to $\vol_G(S)$ and $\vol_G(\ol{S})$. Specifically, recall that all but $\xi n$ vertices have $d(i)$ being $(1+\xi)$-multiplicative approximations to $\max\{ \alpha n, \deg_G(i)\}$, and at most $\beta n$ are less than $\alpha n$, and furthermore, $\vol_G(\ol{S}) \geq\vol_G(S) \geq \eps n^2$ by assumption. By upper bounding $v(S)$ and $v(\ol{S})$ in terms of $\vol_G(S)$ and $\vol_G(\ol{S})$, respectively, both $\vol_G(S) / v(S)$ and $\vol_G(\ol{S}) / v(\ol{S})$ are at least
    \[ \frac{1}{(1+\xi) (1 + \beta \alpha / \eps) + \xi / \eps} \geq 1 - O((\xi + \alpha \beta)/\eps),\]
    which implies that our Rayleigh quotient is at least
    \begin{align*}
        1 - O((\xi + \alpha \beta)/\eps) - \phi_{\eps}(G) \cdot \left(\frac{v(S) + v(\ol{S})}{v(\ol{S})} \right) \cdot \left( \frac{\vol_{G}(S)}{v(S)}\right) \geq 1 - 2\phi_{\eps}(G) - O((\xi + \alpha \beta)/\eps).
    \end{align*}
\end{proof}

Lemma~\ref{lem:feasibility} provides the ``feasibility'' condition, and shows that the easy direction of Cheeger's inequality can tolerate some error in the deflated normalized adjacency matrix. We will use the local computation algorithm for the top eigenvector (Theorem~\ref{thm:main}) with input matrix $(\alpha n/2) \cdot \hat{M}$ will result in entry access to $\bx \in \R^n$ which is an approximately top eigenvector to $\hat{M}$. In what follows, we will apply Theorem~\ref{thm:main} to the input matrix $(\alpha n / 2) \cdot \hat{M}$ and accuracy parameter $\eps_0$ (which will be a polynomial of $\eps$ which we set later). In order to simplify the notation, it will be helpful to define the parameter $\zeta$ which will be the error on the returned Rayleigh quotient of the vector $\bx$ in Lemma~\ref{lem:output-vec} below
\begin{align}
\zeta \leq 2 \phi_{\eps}(G) + O((\xi + \alpha \beta)/\eps + \eps_0 / \alpha). \label{eq:zeta-def}
\end{align}

\begin{lemma}\label{lem:output-vec}
    Let $\bx \in \R^n$ be the output vector of Theorem~\ref{thm:main} with input $(\alpha n/2) \hat{M}$ and accuracy $\eps_0 \in (0,1)$. With high probability over randomness of Theorem~\ref{thm:main}, 
    \begin{align*} 
    \frac{\langle \bx, \hat{M} \bx\rangle}{\|\bx\|_2^2} \geq 1 - \zeta \qquad \|\bx\|_{\infty} \lsim 1 / (\alpha \sqrt{n}) \qquad \frac{1}{2\hat{m}} \left( \sum_{i=1}^n \sqrt{d(i)} \cdot \bx_i \right)^2 \leq (\zeta + O(\xi/\alpha^3)) \cdot \| \bx\|_2^2.
    \end{align*}
\end{lemma}

\begin{proof}
    From Theorem~\ref{thm:main}, the Rayleigh quotient of $\bx$ with $(\alpha n /2) \hat{M}$ will be at least $\lambda_{\max}((\alpha n/2)\cdot \hat{M}) - \eps_0 n$ (with high probability), which is at least $(\alpha n/2)(1 - \zeta)$ by Lemma~\ref{lem:feasibility}. Dividing by $(\alpha n/2)$ gives the desired bound on the Rayleigh quotient of $\hat{M}$. The upper bound on $\|\bx\|_{\infty}$ follows from the description of the algorithm in Theorem~\ref{thm:main} (namely, Figure~\ref{fig:preprocess}), where $\bx = A \bT^{\intercal} \by / \hat{\blambda}_{\max}$ for $A = (\alpha n/2) \hat{M}$. Note, $|\bx_i|$ is by Cauchy-Schwarz at most $1/\hat{\blambda}_{\max} \cdot (\sqrt{\bt} / \sqrt{q}) \cdot \|\by\|_2$, which is at most $O(1/(\alpha \sqrt{n}))$ with high probability (since $\bt = O(nq)$ and $\phi_{\eps}(G)$ is at most a small enough constant).  Finally, the Rayleigh quotient bound implies
    \begin{align*}
        (1 - \zeta) \|\bx\|_2^2 \leq \sum_{i,j \in [n]} \frac{\ind\{(i,j) \in E\}}{\sqrt{d(i) d(j)}} \cdot \bx_i \bx_j - \frac{1}{2\hat{m}} \left( \sum_{i=1}^n \sqrt{d(i)} \bx_i \right)^2,
    \end{align*}
    and because $\deg_G(i) \leq d(i)$ and $\deg_G(j) \leq d(j)$ for all but $\xi n$ vertices (where $d(i) \geq \alpha n$), the right-most quantity above is at most
    \begin{align*}
        - (1-\zeta) \|\bx\|_2^2 + \sum_{i, j \in [n]} \frac{\ind\{(i,j) \in E\}}{\sqrt{\deg_G(i) \deg_G(j)}} \cdot |\bx_i| \cdot |\bx_j| + \frac{2\xi n}{\alpha} \cdot \|\bx\|_{\infty}^2 \leq \zeta \| \bx\|_2^2 + O(\xi / \alpha^3)
    \end{align*}
    since the normalized adjacency matrix has all eigenvalues of magnitude at most $1$, and the bound on $\| \bx\|_{\infty}$. The final substitution follows from the fact $\|\bx\|_2^2$ will be at least a constant (from Lemma~\ref{lem:2}).
\end{proof}

In the next lemma, we show that the any vector satisfying the conclusions of Lemma~\ref{lem:output-vec} satisfies the necessary conditions which we need to execute the ``rounding algorithm'' (from a vector to a cut) in the hard direction of Cheeger's inequality. In the expressions below, the expectations over $(\bi, \bj)$ represent draws uniformly at random from the set of edges $E$ of the graph, and draws $\bi$ from the vertex set consist of the marginal distribution on vertices (note, this is not uniform over $[n]$, but rather proportional to $\deg_G(i)$).

\begin{lemma}\label{lem:clean-guaran}
    Any vector $x \in \R^n$ which satisfies the conclusions of Lemma~\ref{lem:output-vec} satisfies,
 as long as $\xi$ is smaller than $\poly(\alpha \eps)$, $\alpha, \beta$ smaller than $\poly(\eps)$, and $\zeta$ smaller than a fixed constant, 
    \begin{align*}
        \Ex_{(\bi, \bj)}\left[\left(\frac{x_{\bi}}{\sqrt{d(\bi)}} - \frac{x_{\bj}}{\sqrt{d(\bj)}}\right)^2 \right] &\leq \frac{2\|x\|_2^2}{m} \cdot \zeta \qquad \text{and} \qquad  \Ex_{\bi}\left[\left( \frac{x_{\bi}}{\sqrt{d(\bi)}} - \Ex_{\bj}\left[\frac{x_{\bj}}{\sqrt{d(\bj)}}\right] \right)^2\right] \geq \frac{\|x\|_2^2}{4m} 
    \end{align*}
\end{lemma}

\begin{proof}
    Consider setting $\gamma_1$ to $\zeta + O(\xi / \alpha^3)$ to simplify the final condition of Lemma~\ref{lem:output-vec} and $\gamma_2$ to $1/\alpha^2$ so $n \|x\|_{\infty}^2 \leq \gamma_2 \|x\|_2^2$ (recall, $\|x\|_2^2$ will $\Theta(1)$). For the first inequality, we use the fact that $d(i) \geq \deg_G(i)$ for all but at most $\xi n$ vertices, and that $d(i) \geq \alpha n$. In particular, expanding the square we obtain
    \begin{align*}
        \Ex_{(\bi,\bj)}\left[ \left(\frac{x_{\bi}}{\sqrt{d(\bi)}} - \frac{x_{\bj}}{\sqrt{d(\bj)}} \right)^2 \right] \leq \frac{2}{m} \sum_{i=1}^n \frac{\deg_G(i)}{d(i)} \cdot x_i^2 - \frac{2}{m} \langle x, \hat{M} x \rangle \leq \frac{2}{m} \left( \|x\|_2^2 - \langle x, \hat{M} x\rangle + (\xi/\alpha) n\|x\|_{\infty}^2\right),
    \end{align*}
    which is at most the desired bound. For the second inequality, we first note that letting $\mu = \Ex_{\bi}[x_{\bi} / \sqrt{d(\bi)}]$, any $\tau \in \R$ satisfies
    \begin{align}
        \Ex_{\bi}\left[ \left(\frac{x_{\bi}}{\sqrt{d(\bi)}} - \tau \right)^2 \right] \geq \Ex_{\bi}\left[ \frac{x_{\bi}^2}{d(\bi)} \right] - \mu^2 &= \frac{1}{2m} \sum_{i=1}^n \frac{\deg_G(i)}{d(i)} \cdot x_i^2 - \mu^2 \geq \frac{\|x\|_2^2}{2m} \left( \dfrac{1 - (\beta + \xi) \gamma_2}{1+\xi} \right) - \mu^2. \label{eq:bound-denom}
    \end{align}
    Finally,  we upper bound the magnitude of $\mu$ via the triangle inequality and Cauchy-Schwarz
    \begin{align*}
        |\mu| &\leq \frac{1}{2m} \left| \sum_{i=1}^n \frac{\deg_G(i) - d(i)}{\sqrt{d(i)}}\cdot x_i \right| + \frac{1}{2m} \left| \sum_{i=1}^n \sqrt{d(i)} \cdot x_i \right| \\   
        		&\leq \frac{\|x\|_2}{2m} \left( \sum_{i=1}^n \frac{(\deg_G(i) - d(i))^2}{d(i)} \right)^{1/2} + \frac{\sqrt{\gamma_1 \hat{m}} \cdot \|x\|_2}{2m} 
    \end{align*}
    where the final summation inside the square-root (from the fact $d(i)$ is oftentimes an approximation of $\deg_G(i)$) is at most $(\xi^2 + 4 \xi / \alpha + 4 \alpha \beta) n^2$. In particular, this implies that (\ref{eq:bound-denom}) is at least
    \[ \frac{\|x\|_2^2}{2m} \left(\dfrac{1 - (\beta + \xi) \gamma_2}{1 + \xi} - \frac{2(\xi^2 + 4 \xi / \alpha + 4 \alpha \beta)}{\eps} - 2\gamma_1(1+\xi) \right), \]
    where in the final line, we use the fact that $m \geq \eps n^2$ and that $\hat{m} \leq (1+\xi) m$, and this implies the desired inequality.
\end{proof}

Finally, we will use the proof of the hard direction of Cheeger's inequality, which gives a rounding algorithm from a vector to a cut. 
\begin{theorem}[Hard Direction of Cheeger's Inequality]\label{thm:cheeger}
    Fix any graph $G = ([n], E)$ as well as a vector $y \in \R^n$ satisfying
    \[ \Ex_{(\bi,\bj)}\left[ (y_{\bi} - y_{\bj})^2 \right] \leq a \qquad \text{and}\qquad \Ex_{\bi}\left[ \left(y_{\bi} - \Ex_{\bi}[y_{\bi}] \right)^2 \right] \geq b. \]
    There exists a cut $S \subset [n]$ of the form $S = \left\{ i \in [n] : \sigma y_i \geq \tau \right\}$ for $\tau \in \R$ and $\sigma \in \{-1,1\}$ which satisfies
    \begin{align*}
        \vol_G(\ol{S}) \geq \vol_G(S) \geq \Omega(m b / \| y\|_{\infty}^2) \qquad \text{and}\qquad \phi_G(S) \lsim \sqrt{ a / b}.
    \end{align*}
\end{theorem}

\begin{remark}[Volume of $S$] 
Standard proofs of Theorem \ref{thm:cheeger} may not explicitly compute a bound on $\vol_{G}(\bS)$, but such bounds follows directly from standard rounding schemes. The vector $y$ will be rounded by centering the vector around the median $\tau \in \R$ and then sampling a threshold $\gamma$ uniformly from $0$ to $\max_{i} (y_i-\tau)^2$. The setting of $\tau$ ensures that half the elements are positive and half negative. Therefore $\bS$ contains fewer vertices and has smaller volume than $\ol{\bS}$. The expectation can be computed explicitly, which is simply the degrees times the probability that half the elements $(y_i -\tau)^2$ are larger than $\bgamma$. 
\begin{align*}
\E[\min(\vol_{G}(\bS),\vol_{G}(\ol{\bS}))] = \E[\vol_{G}(\bS)] = \frac{1}{2}\sum_{i} d_i \Pr[i \in \bS] = \frac{\sum_{i} d_i (y_i -\tau)^2}{2\max_{i} (y_i-\tau)^2}
\end{align*} 
The final bound follows from the fact that $(y_i -\tau)^2 \leq \|y\|_{\infty}^2$ and the observation that the initial condition of Lemma \ref{lem:trevisan-rounding}, $\Ex_{\bi}\left[ \left(y_{\bi} - \Ex_{\bi}[y_{\bi}] \right)^2 \right]$ can be no greater than $\sum_{i} d_i (y_i - \tau)^2/m$, as $\Ex_{\bi}[y_{\bi}]$ is the minimizing setting of $\tau$.
\end{remark}

\paragraph{Putting Everything Together: Proof of Theorem~\ref{thm:cheeger-dense}.} First, note that, as discussed prior to (\ref{eq:noisy-deflate}), we assume the local computation algorithm receives query access to the entries of an $n \times n$ matrix $\hat{M}$ in (\ref{eq:noisy-deflate}), by using an additive $\poly(1/(\xi \eps))$-additive overhead on the size of the query set $\bQ$; in particular, for any $(i,j) \in [n]\times [n]$, we may estimate $\deg_G(i)$ and $\deg_G(j)$ via determining the number of edges to $\bQ$, and obtain $d(i)$ and $d(j)$, as well as the number of edges in $\bQ \times \bQ$ in order to let $\hat{m}$ be the estimate of $m$.

From Lemma~\ref{lem:feasibility}, the matrix $(\alpha n/2)\hat{M}$ is a bounded-entry matrix, whose top eigenvalue is is at least $(\alpha n/2) (1 - 2\phi_{\eps}(G) - O((\xi + \alpha \beta) / \eps))$, and as a result of Theorem~\ref{thm:main} with accuracy $\eps_0$, there is a (randomized) local computation algorithm $\calA$ which receives query access to the entries of a vector $\bx \in \R^n$. With high probability over the randomness of $\calA$, the output vector $\bx$ satisfies the three conditions of Lemma~\ref{lem:output-vec}. We consider the vector $\by \in \R^n$ where each $\by_i$ is given by rounding $\bx_i / \sqrt{d(i)}$ to the nearest multiple of $\poly(\eps) / n$. One can see that the vector $\by$ satisfies the conditions of Theorem~\ref{thm:cheeger} with $a = \zeta / m$ and $b = 1/(4m)$, and in addition, each coordinate is one of at most $\poly(1/\eps)$ many values. In other words, there exists a collection of at most $\poly(1/\eps)$ many cuts obtained by thresholding $\by$ such that at least one of them satisfies $\phi_G(S) \lsim \sqrt{\zeta}$ and $\vol_G(\ol{S}) \geq \vol_G(S)$ is at least $\Omega(1/\|\by\|_{\infty}^2)$, which is at least $\Omega(\alpha^3 n^2)$. This guarantees the existence of a cut $S \subset [n]$ among a family of at most $\poly(1/\eps)$ cuts; using $\poly(\log(1/\eps)/\alpha)$ additional queries, we may iterate through all possible cuts $S$, and estimate $\vol_G(S)$ and $E_G(S, \ol{S})$ to find the desired cut. After setting the parameters $\eps_0 \leq \poly(\eps \alpha)$, $\xi \leq \poly(\alpha \beta)$ and by the assumption of Theorem \ref{thm:cheeger-dense} that $\alpha,\beta \leq \poly(\eps)$, we conclude that the set $S$ returned satisfies $\phi_G(S) \lesssim \sqrt{\phi_{\eps}(G)} + \eps$.

\section{Local Computation Algorithm for Max-Cut} \label{sec:max-cut}
Similar to the local computation algorithm for sparsest cut, this section outlines the application of Theorem \ref{thm:main-intro} to provide local access to an approximate maximum cut of a dense and nearly bipartite graph. In particular, we prove the following theorem. 

\localtrevisan*

The algorithm is structurally quite similar to the LCA for the sparsest cut problem, in particular it reduces the problem of computing the maximum cut in approximately bipartite graphs to repeatedly computing close to bipartite cuts in a certain class of dense graphs. We define the following measure of distance to bipartiteness (also defined in \cite{L12}) which measures for a given partition $L,R$ of a subset $S$ of $V$, the fraction of edges incident to $S$ with one endpoint in $L$ and the other not in $R$ (and vice-versa).

\begin{definition}[Bipartiteness Ratio of a Graph]\label{def:bipartitness}
    For a fixed graph $G = ([n], E)$ and a subset $S \subset [n]$ split into two disjoint partitions $(L,R)$ such that $S = L \cup R$, we let
    \begin{itemize}
        \item The bipartiteness ratio of the partition is $\calB_G(L,R) = (2E(L,L)+ 2E(R,R)+E(S,\ol{S}))/\vol_G(L \cup R)$.
        \item The bipartiteness ratio of the entire graph $G$ is defined by 
    \[ \calB(G) = \min_{S \subseteq V} \calB_G(S,\ol{S}). \]
    \end{itemize}
    
\end{definition}

\subsection{Reduction to $(\alpha, \beta)$-Dense Graphs, for $\alpha, \beta = \poly(\eps)$: Proof of Theorem~\ref{thm:trevisan}}

As in the proof of Theorem \ref{thm:cheeger-intro}, we show that it suffices to construct a local computation algorithm for $(\alpha,\beta)$-dense graphs (see Definition \ref{def:dense-graph}). For this section, consider a fixed input graph $G = ([n], E)$ which we provide query access to, as well as a fixed accuracy parameter $\eps > 0$ (notice, we may assume that $G$ contains a max-cut of at least $\eps n^2$ edges, as otherwise, the $\eps n^2$ error  allows any cut $S$ to satisfy the approximation ratio of the Theorem).

\begin{lemma}\label{lem:trevisan-noise-addition}
For $\alpha, \beta \in (\eps^c,\eps/c)$ for some constant $c$, let $\bN \sim \calG(n, 4\alpha)$ be an Erd\"{o}s-Reyni graph and let $\bH = G \cup \bN$. The following occurs with high probability over $\bN$. 
\begin{itemize}
\item For any $T \subset [n]$ which satisfies $\vol_{\bH}(\ol{T}) \geq \eps n^2 / 4$. 
\item Let $\bH'$ be the subgraph of $\bH$ induced by $\ol{T}$ and $\rho(\ol{T}) = \vol_{\bH'}(\ol{T}) / \vol_{\bH}(\bH)$. Then, $\bH'$ is $(\alpha, \beta)$-dense and satisfies $\calB(\bH') \leq \calB(G)/\rho(\ol{T}) + O(\alpha / (\eps  \rho(\ol{T})) )$.
\end{itemize}
\end{lemma}
\begin{proof}
The proof of $\bH'$ being $(\alpha,\beta)$-dense follows from the proof of Lemma \ref{lem:noise-addition}. We simply need to upper bound $\calB(\bH')$. As before, consider the subset $S \subset [n]$ which realizes the bound on $\calB(G)$, and let $S' = S \cap \ol{T}$ and $\ol{S}' = \ol{S} \cap \ol{T}$. We upper bound $\calB_{\bH'}(S',\ol{S}')$ in terms of $\calB_{G}(S,\ol{S})$. We have that by definition $\vol_{\bH'}(S') + \vol_{\bH'}(\ol{S}') = \rho(\ol{T})\cdot \vol_{\bH}(\bH)$ bounding the denominator of $\calB_{\bH'}(S',\ol{S}')$. Then, we have the with high probability over the draw of $\bN$ (since $\bN$ has at most $8\alpha n^2$ edges), that  $E_{\bH'}(S',S') \leq E_{G}(S, S) + 8\alpha n^2$ and $E_{\bH'}(\ol{S}',\ol{S}') \leq E_{G}(\ol{S}, \ol{S}) + 8\alpha n^2$. Since $\vol(G) \geq \eps n^2$, this implies $\calB(\bH') \leq \calB(G)/\rho(\ol{T}) + 16\alpha /(\eps\rho(\ol{T}))$.
\end{proof}

Similarly to the sparsest cut LCA, we will assume there exists a local computation algorithm $\calA$ which finds a cut with nearly optimal bipartiteness ratio assuming the underlying graph is $(\alpha,\beta)$-dense (Theorem~\ref{thm:trevisan-dense}). Using Lemma~\ref{lem:trevisan-noise-addition}, we analyze the algorithm in Figure~\ref{fig:prepro-trevisan} to prove Theorem~\ref{thm:trevisan}, which proceeds by iteratively cutting out approximately bipartite cuts and repeating on the induced subgraph of remaining vertices. The main difference from the sparsest cut LCA is that the volume of the final cut needs to be close to $2\vert E \vert$, ensuring it cuts at least $(1-\eps)\vert E \vert$ many edges. Any remaining edges have their vertices assigned randomly. In order to determine when sufficently many edges have been cut, we will assume access to two subroutines, $\EstimateVol(S, H, \xi, \delta)$ defined in Theorem \ref{thm:cheeger-intro}, which receives query access of the form ``$i \in S$'', as well as query access to the adjacency matrix of a graph $H$ on vertex set $[n]$ and outputs an estimate $\hat{\boldeta}$ of $\vol_{H}(S)$ which is off by at most $\xi n^2$ with probability $1-\delta$.

\begin{theorem}\label{thm:trevisan-dense}
There exists a (randomized) local computation algorithm $\calA$ with query access to the adjacency matrix $G = (V, E)$, an accuracy parameter $\eps > 0$, and a failure probability $\delta > 0$. 
\begin{itemize}
\item For any $\gamma \geq \poly(\log(1/\delta)/\eps)/|V|$, $\calA$ draws a subset $\bQ \subset V$ where $j \in \bQ$ i.i.d. with probability $\gamma$, and queries the induced subgraph of $\bQ$. On query $i \in V$, $\calA$ queries entries $(i, j)$ for all $j \in \bQ$.  
\item Whenever $G$ is $(\alpha,\beta)$-dense, letting $\bL,\bR \subset [n]$ be the cut that $\calA$ generates,
\begin{align*}
\Prx_{\calA}\left[ \begin{array}{c} \calB_{G}(\bL,\bR) \leq O(\sqrt{\calB(G)}) + \eps \\ \alpha^3 n^2 \leq \vol_{G}(\bL \cup \bR)  \end{array} \right] \geq 1-\delta.
\end{align*}
\end{itemize}
\end{theorem}

\begin{figure}[ht!]
\begin{framed}
\textbf{Local Computation Algorithm for Max Cut.} The algorithm receives query access to the adjacency matrix of a graph $G = ([n], E)$, an accuracy parameter $\eps > 0$, and an index $i \in [n]$.
\begin{enumerate}
\item Use public randomness to generate $\bN \sim \calG(n, 4\alpha)$, and use query access to the adjacency matrix of $G$ to provide query access to the adjacency matrix of $\bH_0 = G \cup \bN$.
\item Initialize the following objects:
\begin{itemize}
\item A collection $\sfA$ (initially empty) of at most $t$ local computation algorithms ($\calA_1, \dots, \calA_t$) which provide query access to subsets of $[n]$ ($\bL$ and $\bR$, which are the union of the first $\ell$ assignments, for some $\ell \leq t$).
\item Using public randomness, $t$ independent draws $\bQ_1^{(0)}, \dots, \bQ_t^{(0)} \subset [n]$ where $j \in \bQ_{\ell}^{(0)}$ i.i.d. with probability $2\gamma$ for $\gamma \geq \poly(1/\eps)/n$ \footnotemark.
\item Execute $\EstimateVol([n], \bH_0, \eps/8, 1/400)$ and let $\hat{\bv}$ be the output.
\end{itemize}
\item\label{ln:trev-vol-check} For $\ell=1, \dots t$, repeat the following:
\begin{enumerate}
\item\label{ln:trev-estimate-vol} Let $\bT_{\ell-1} \subset [n]$ denote the set given by
\begin{align*}
    \bT_{\ell-1} = \left\{ j \in [n] : \exists \ell' \leq \ell - 1 \text{ s.t. $\calA_{\ell'}$ declares $j$ lies in $\bL_{\ell'}$ or $\bR_{\ell'}$ } \right \}
\end{align*} 
Execute $\EstimateVol(\bT_{\ell-1}, \bH_0, \eps/8, 1/(400 t))$ and let $\hat{\boldeta}$ be the output. If $\hat{\boldeta} \geq \hat{\bv}-\eps n^2/2$, break out of the loop to Line~\ref{ln:trev-default}. 
\item Otherwise, $\hat{\boldeta} < \hat{\bv}-\eps n^2/2$ and continue. Let $\bV_{\ell} = [n] \setminus \bT_{\ell-1}$, which is the (implicit) vertex set obtained after removing the union of sets defined by $\calA_{1},\dots, \calA_{\ell-1}$. 
\item For every $\ell' \leq \ell-1$, query the local computation algorithm $\calA_{\ell'}$ with $j$ for $j \in \bQ^{(0)}_{\ell}$ in order to write down $\bQ_{\ell} = \bQ^{(0)}_{\ell} \cap \bV_{\ell}$. Note, $\bQ_{\ell}$ includes each $j \in \bV_{\ell}$ i.i.d. w.p $2\gamma$.
\item\label{ln:trev-declare-output} Let $\bH_{\ell}$ be the subgraph of $\bH_0$ induced by $\bV_{\ell}$, and let $\calA_{\ell}$ be the local computation algorithm from Theorem~\ref{thm:trevisan-dense} with parameter $\eps' = \eps/2$, which uses query set $\bQ_{\ell} \subset \bV_{\ell}$. Query $\calA_{\ell}$ with $i$ (which must lie in $\bV_{\ell}$), and if $\calA_{\ell}$ declares that $i$ lies in the partial cut, output ``$i \in \bL$'' or ``$i \in \bR$'' and terminate.
\end{enumerate}
\item\label{ln:trev-default} If the procedure has not produced an output yet, using public randomness output ``$i \in \bL$'' with probability $1/2$ and otherwise ``$i \in \bR$''.
\end{enumerate}
\end{framed}
\caption{Local Computation Algorithm for Max Cut}\label{fig:prepro-trevisan}
\end{figure}
\footnotetext{Note that Theorem \ref{thm:main} samples a subset of a fixed size, while other sub-routines such as in Theorem \ref{thm:trevisan-dense} or $\EstimateVol(S, H, \xi, \delta)$ sample entries independently with fixed probability. These sampling methods are interchangeable, so it suffices to define query sets $\bQ_1,\dots \bQ_t$ in terms of either sampling model.}

\paragraph{Correctness Guarantees.} We first establish the correctness guarantees of Theorem~\ref{thm:trevisan}, where we show that the local computation algorithm specified in Figure~\ref{fig:prepro-trevisan} provides query access to a cut $\bS \subset [n]$ which satisfies $E(\bS,\ol{\bS}) \geq (1-\sqrt{\varphi})\vert E \vert - \eps^2 n$ with probability at least $0.99$. Observe that by the assumptions of Theorem~\ref{thm:trevisan}, we know that $G$ contains a cut of volume $2\vert E \vert$, which has bipartiteness ratio at most $\varphi$. Therefore, our goal will be to output a cut of volume at least $(1-\eps)\vert E \vert$ with bipartiteness ratio at most $\sqrt{\varphi} + \eps$. 

By setting $\delta > \frac{1}{400 t}$, we assume that all executions of $\EstimateVol$, and uses of Lemma~\ref{lem:trevisan-noise-addition} and Theorem~\ref{thm:trevisan-dense} result in conclusions which succeed; each event occurs with high enough probability, so by a union bound, all events occur with probability at least $0.99$. In the remainder of the proof, we assume this is the case. Notice that Figure~\ref{fig:prepro-trevisan} defines a local computation algorithm which, on query $i \in [n]$ will also produce an output (either that ``$i \in \bL$'' or ``$i \in \bR$'' in some execution of Line~\ref{ln:trev-declare-output}, or that ``$i \in \bL$'' or ``$i \in \bR$'' in Line~\ref{ln:trev-default})\footnote{Either $\bL$ or $\bR$ can be the cut $\bS$ specified in Theorem \ref{thm:trevisan}, as all vertices are contained in $\bL \cup \bR$}. Hence, we may define
\[ \bL = \left\{ i \in [n] : \text{ Figure~\ref{fig:prepro-trevisan} on input $G$, $\eps$, and $i \in [n]$ outputs ``$i \in \bL$''}\right\}, \]
and $\bR$ analogously, notice that $\bL$ naturally decomposes as a union $\bL_1 \cup \dots \cup \bL_{\bell_0-1} \cup \bL_{\bell_0}$, where for $\ell \in \{1, \dots, \bell_0 - 1\}$, $\bL_{\ell}$ consists of the indices $i \in [n]$ where $\calA_{\ell}$ in Line~\ref{ln:trev-declare-output} outputs ``$i \in \bL$'' and $\ell$ is the smallest to do so, and $\bL_{\ell_0}$ consists of all vertices which reach Line~\ref{ln:trev-default} and declare ``$i \in \bL$''. The analogous decomposition also defines $\bR$ as $\bR_1 \cup \dots \cup \bR_{\ell_0-1} \cup \bR_{\ell_0}$. Note, $\bell_0$ represents the iteration of Line~\ref{ln:trev-estimate-vol} which triggered the ``break'' condition, and we default $\bell_0 = t+1$ if Line~\ref{ln:trev-estimate-vol} is never triggered (we will show this does not happen). Note, the following inductive definitions are a consequence of the iterative algorithm in Figure~\ref{fig:prepro-trevisan}:
\begin{itemize}
    \item $\bT_0 = \emptyset$.
    \item For $\ell=1, \dots, \bell_0-1$, $\bL_{\ell},\bR_{\ell}$ are the (implicit) subsets defined by $\calA_{\ell}$ executed on input graph $\bH_{\ell}$, which is the subgraph of $\bH_0$ induced by $\bV_{\ell} = [n] \setminus \bT_{\ell-1}$.
    \item Let $\bL_{\ell_0},\bR_{\ell_0}$ be the subsets of $\bH_{\ell_0}$ induced by $\bV_{\ell_0} = [n]\setminus \bT_{\ell_0 -1}$ which are output by Line \ref{ln:trev-default}. 
    \item For $\ell=2, \dots,\bell_0$, we have $\bT_{\ell-1} = \bL_1 \cup \bR_1 \cup \dots \cup \bL_{\ell-1} \cup \bR_{\ell-1}$.
\end{itemize}
Notice that, for every $\ell = 1, \dots, \bell_0-1$ executed in Line~\ref{ln:trev-vol-check} will have $\bT_{\ell-1}$ with $\vol_{\bH_0}(\bT_{\ell-1}) \leq 2\vert E \vert-\eps n^2/4$. This is because $\hat{\bv} \geq 2\vert E \vert-\eps n^2/8$ and for $\ell < \bell_0$, Line~\ref{ln:trev-estimate-vol} results in $\hat{\boldeta} \leq \hat{\bv} -\eps n^2/2$, and the fact that $\EstimateVol(\bT_{\ell-1}, \bH_0, \eps/8,\delta)$ is an $\eps n^2 / 8$-error approximation of $\vol_{\bH_0}(\bT_{\ell-1})$. Hence, define $\rho_{\ell}$ to be the fraction of the total volume which survives inside the current subgraph, 
\[ \rho_{\ell} \eqdef \vol_{\bH_{\ell}}(\bV_{\ell}) / \vol_{\bH_0}(\bH_0), \] 
and notice that $\vol_{\bH_0}(\bV_{\ell}) = 2|E| - \vol_{\bH_0}(\bT_{\ell-1}) \geq \eps n^2/ 4$ for any $\ell < \bell_0$, so we can use Lemma~\ref{lem:trevisan-noise-addition} and conclude the subgraph $\bH_{\ell}$ of $\bH_0$ induced by vertex set $\bV_{\ell}$ is $(\alpha, \beta)$-dense, and $\calB(\bH_{\ell}) \leq \calB(G)/\rho_{\ell} + O(\alpha/(\eps\rho_{\ell}))$.  We have that $\rho_{\ell_0}$ is at most $3\eps n^2/(4 \vol_{\bH_0}(\bH_0))$, since $\rho_{\ell_0} \vol_{\bH_0}(\bH_0) = \vol_{\bH_{\ell_0}}(\bV_{\ell_0}) \leq \vol_{\bH_0}(\bV_{\ell_0}) \leq 3\eps n^2/4$. We also have that $\vert \bV_{\ell}\vert \geq \eps n /4$, since the degrees of $\bV_{\ell}$ in $\bH_0$ are at most $n$.  As a result, the set $\bQ_{\ell} \subset \bV_{\ell}$, which includes each index $j \in \bV_{\ell}$ i.i.d. with probability $\gamma'$ satisfies $\gamma' \geq \poly(1/\eps)/n \geq \poly(1/\eps) / |\bV_{\ell}|$, and we may apply Theorem~\ref{thm:trevisan-dense}. The local computation algorithm $\calA_{\ell}$ in Line~\ref{ln:trev-declare-output} will satisfy,
\begin{align} 
\calB_{\bH_{\ell}}(\bL_{\ell}, \bR_{\ell}) \leq O\left(\sqrt{\calB(\bH_{\ell})}\right) + \eps, \qquad \text{and}\qquad \alpha^3 \vert \bV_{\ell} \vert^2 \leq \vol_{\bH_{\ell}}(\bL_{\ell} \cup \bR_{\ell}).  \label{eq:thm-trev-dense-bounds}
\end{align}
We aim to bound the total number of edges not cut by $\bL$ and $\bR$. For a graph $G_0$ on a vertex set $V$ and disjoint subsets $A, B \subset V$ which are not necessarily a partition of $V$, let $U_{G_0}(A,B)$ be sum of edges in $G_0$ incident on $A \cup B$, which are not in $A \times B$. We can upper bound $U_{G_0}(A,B)$ by $\calB_{G_0}(A, B) \cdot \vol_{G_0}(A \cup B)$. Note, for every edge $(u, v) \in \bH_0$ which does not cross $\bL, \bR$, must either be fully contained in $\bL$ or fully contained in $\bR$ (since $\bL \cup \bR = [n])$. Suppose $u, v \in \bL$, then assume without loss of generality that $u \in \bL_{\ell}$ and that $v \in \bL_{\ell'}$ for $\ell' \geq \ell$, so $(u, v)$ is counted in $U_{\bH_{\ell}}(\bL_{\ell}, \bR_{\ell})$. The similar argument holds for $\bR$, and it allows us to upper bound $U_{\bH_0}(\bL, \bR)$ by looking at each of the individual subgraphs found:
\begin{align*}
    U_{\bH_0}(\bL, \bR) \leq \sum_{\ell=1}^{\bell_0} U_{\bH_{\ell}}(\bL_{\ell},\bR_{\ell}) &\leq \sum_{\ell=1}^{\bell_0-1}\left(O(\sqrt{\calB(\bH_{\ell}}) + \eps \right) \cdot \vol_{\bH_{\ell}}(\bL_{\ell} \cup \bR_{\ell}) + \vol_{\bH_{\ell_0}}(\bH_{\ell_0})\\
    &\leq \sum_{\ell=1}^{\bell_0-1}\left(O(\sqrt{\calB(G)/\rho_{\ell} + \alpha/(\eps\rho_{\ell})}) + \eps \right) \cdot \vol_{\bH_{\ell}}(\bL_{\ell} \cup \bR_{\ell}) + \rho_{\ell_0} \vol_{\bH_0}(\bH_0) 
\end{align*}
where the second inequality uses the bound on $\calB_{\bH_{\ell}}(\bL_{\ell},\bR_{\ell})$ and the third inequality uses Lemma~\ref{lem:trevisan-noise-addition}. We can factor out the terms involving $\eps$, to get $\eps \sum_{\ell=1}^{\bell_0-1} \vol_{\bH_{\ell}}(\bL_{\ell} \cup \bR_{\ell}) \leq \eps \cdot \vol_{\bH_0}(\bH_0)$. We now repeat the proof of Theorem 4 in \cite{L12} to bound the entire sum.  Consider the $\ell$-th cut $\bL_{\ell}, \bR_{\ell}$ for any $\ell < \bell_0$. Every edge of $\bH_{\ell}$ which is incident on $\bL_{\ell} \cup \bR_{\ell}$ is absent from $\bH_{\ell+1}$, and each such edge contributes at most $2$ to $\vol_{\bH_{\ell}}(\bL_{\ell} \cup \bR_{\ell})$, so 
\[ \vol_{\bH_{\ell}}(\bL_{\ell} \cup \bR_{\ell}) = \vol_{\bH_{\ell}}(\bH_{\ell}) - \vol_{\bH_{\ell}}(\ol{\bL_{\ell} \cup \bR_{\ell}}) \leq \vol_{\bH_{\ell}}(\bH_{\ell}) - \vol_{\bH_{\ell+1}}(\bH_{\ell+1}) = (\rho_{\ell} - \rho_{\ell+1}) \vol_{\bH_0}(\bH_0), \]
Recall that $\rho_{\ell_0} \cdot \vol_{\bH_0}(\bH_0) \leq \vol_{\bH_0}(\bV_{\ell_0}) \leq O(\eps n^2)$ as a consequence of Line \ref{ln:trev-estimate-vol}. In order to combine the terms together smoothly, we can bound the sum over $\ell=1$ to $\ell_0-1$ of $(\rho_{\ell} - \rho_{\ell+1})\frac{1}{\sqrt{\rho_{\ell}}}$ by the integral $ \int_{\rho_{\ell+1}}^{\rho_{\ell}}1/\sqrt{x} \ dx$. We can thus apply the following upper bound 
\begin{align*}
    (\rho_{\ell} - \rho_{\ell+1}) \vol_{\bH_0}(\bH_0) \cdot O(\sqrt{\calB(G)/\rho_{\ell} + \alpha/(\eps \rho_{\ell})}) \leq O(\vol_{\bH_0}(\bH_0)) \cdot \int_{\rho_{\ell+1}}^{\rho_{\ell}} (\sqrt{\frac{\calB(G)}{x} + \frac{\alpha}{\eps x}}) dx
\end{align*}
Combining the integrals together and integrating yields a final bound of 
\begin{align*}
    U_{\bH_0}(\bL, \bR) \leq O(\vol_{\bH_0}(\bH_0) ) \cdot \int_{0}^{1} \left(\sqrt{\frac{\calB(G)}{x} + \frac{\alpha}{\eps x}}\right) dx + O(\eps n^2) \leq O(\sqrt{\calB(G) + \frac{\alpha }{\eps}}) \cdot\vol_{\bH_0}(\bH_0) + O(\eps n^2).
\end{align*}
We now have that  $U_{\bH_0}(\bL, \bR) \geq U_{G}(\bL, \bR)$, as removing the edges of $\bN$ can only reduce the number of uncut edges. After setting $\alpha,\beta \leq \poly(\eps)$, we can conclude that if the optimal max-cut of $G$ cuts $(1-\varphi)\vert E \vert$ edges, then $\calB(G) = \varphi$, and the cut $\bL,\bR$ returned by Theorem \ref{thm:trevisan} cuts at least $\left ( 1-O(\sqrt{\varphi}) \right )\vert E \vert - \eps n^2$ edges.  In order to show that $\bell_0 < t$, and the condition of Line~\ref{ln:trev-estimate-vol} eventually gets triggered, we use the fact the right-most inequality of (\ref{eq:thm-trev-dense-bounds}), which implies that $\vol_{\bH_0}(\bL \cup \bR)$ is at least $(\bell_0 - 1) \cdot \alpha^3 \vert \bV_{\ell} \vert^2$. Recall we have that $\vert \bV_{\ell} \vert \geq \sqrt{\eps/4} \ n$, so once $\bell_0 - 1 \geq 4 / (\eps \alpha^3)$, $\vol_{\bH_0}(\bL \cup \bR)$ would be large enough to trigger Line~\ref{ln:trev-estimate-vol}.

\paragraph{Query Complexity.} We turn to the query complexity of Figure~\ref{fig:prepro-trevisan}. The proof is analogous to that of Theorem~\ref{thm:cheeger-intro}. A single query to the adjacency matrix of $\bH_0$ uses at most one query to the adjacency matrix of $G$. For any particular local computation algorithm $\calA_{\ell}$ initialized in Line~\ref{ln:trev-declare-output}, Figure~\ref{fig:prepro-trevisan} does the following:
\begin{itemize}
    \item Determine which $\bQ_{\ell}^{(0)} \subset [n]$ are part of the induced subgraph $\bH_{\ell}$, which requires making $|\bQ_{\ell}^{(0)}|$ queries to the (previous) local computation algorithms $\calA_1, \dots, \calA_{\ell-1}$. This is used to determine the set $\bQ_{\ell}$.
    \item Once we know $\bQ_{\ell}$, $\calA_{\ell}$ makes $|\bQ_{\ell}|^2$ queries during the preprocessing time, and $|\bQ_{\ell}|$ queries during any particular query. 
\end{itemize}
As a result, the total number of probes needed is:
\[ \sum_{\ell=1}^t |\bQ_{\ell}|^2 + \sum_{\ell=1}^t \sum_{\ell' < \ell} |\bQ_{\ell}^{(0)}| \cdot |\bQ_{\ell'}|, \]
where the first term constitutes all of the preprocessing queries (done after determining $\bQ_{\ell}$ for each $\ell$), plus the $|\bQ_{\ell}^{(0)}|$ queries which are done to $\calA_{\ell'}$ for $\ell' < \ell$, each of which uses $|\bQ_{\ell'}|$ probes to the adjacency matrix of $G$. The fact that all $\bQ_{\ell} \subset \bQ_{\ell}^{(0)}$ and all are of size $\poly(1/\eps)$, and the fact $t$ is at most $O(\frac{1}{\eps \alpha^3})$ with $\alpha = \poly(\eps)$ gives the desired upper bound.


\subsection{Max Cut on $(\alpha,\beta)$-dense Graphs with $\alpha,\beta = \poly(\eps)$: Proof of Theorem \ref{thm:trevisan-dense}}

Fix an input graph $G = (V, E)$ which is $(\alpha,\beta)$-dense and has at least $\eps n^2$ edges. As in the algorithm of Theorem~\ref{thm:cheeger-dense}, the algorithm of Theorem~\ref{thm:trevisan-dense} will start by providing query access to the negative normalized adjacency matrix $N$ (defined in Definition~\ref{def:negative-normalized} below) on which we will execute our top eigenvector local computation algorithm from Theorem~\ref{thm:main-intro}, in order to access a approximation of the top vector.

\begin{definition}\label{def:negative-normalized}
    For a graph $G = ([n], E)$, the negative normalized adjacency matrix $N$ is the $n \times n$ matrix whose  $(i,j)$-entry is 
    \[ N_{ij} = \dfrac{-\ind\{ (i,j) \in E \}}{\sqrt{\deg_G(i) \cdot \deg_G(j)}}. \]
\end{definition}

As with definition \ref{def:deflated} in Theorem \ref{thm:cheeger-dense}, we  cannot provide exact access to the matrix $N$, and instead we provide access to an approximation of $N$ using the same degree estimates as in (\ref{eq:noisy-deflate}). The sampling to estimate $\deg(u)$ is non-adaptive (requiring uniform random draws from $[n]$), and incurs $\poly(1/(\xi\eps))$-additive overhead on the size of $\bQ$ (as specified in Theorem~\ref{thm:trevisan-dense}). Hence the algorithm receives query access to the matrix $\hat{N}$:
\begin{align}
    \hat{N}_{ij} = \dfrac{-\ind\{ (i,j) \in E \}}{\sqrt{d(i) \cdot d(j)}}, \label{eq:noisy-negative}
\end{align}
which we will refer to as a noisy negative normalized adjacency of $G$.

\begin{lemma}[Feasibility]\label{lem:trevisan-feasibility} Let $\hat{N}$ be a noisy negative normalized adjacency matrix of an $(\alpha,\beta)$-dense graph $G=([n],E)$. Then, $\|\hat{N}\|_{\infty} \leq 1/(\alpha n)$ and $\lambda_{\max}(\hat{N}) \geq 1-2\calB(G) - O((\xi + \alpha \beta)/\eps)$ 
\end{lemma}
\begin{proof}
    The proof is quite similar to that of Lemma \ref{lem:feasibility}. We have that by definition (\ref{eq:noisy-negative}) for all $i \in [n]$ satisfies $\alpha n \leq d(i) \leq n$. Hence, each entry is at most $1/(\alpha n)$. For the lower bound on $\lambda_{\max}(\hat{N})$, consider the partition $L,R \subset [n]$ which realizes the bound on $\calB(G)$. Recall $L \cup R = [n]$, so all edges are explicitly cut.  Consider the vector $x \in \R^n$ which sets $x_i = \sqrt{d(i)}$ if $ i \in L$ and $x_i = -\sqrt{d(i)}$ if $ i \in R$. 
    To bound the denominator, recall that $\xi n$ vertices have $d(i)$ within a  $(1+\xi)$-multiplicative approximations of $\max\{ \alpha n, \deg_G(i)\}$, and at most $\beta n$ are less than $\alpha n$. Since $\vert E \vert \geq \eps n^2$, the sum of $d(i)$ is within a $(1+\xi)(1+\beta\alpha/\eps) + \xi/\eps$ of the sum of the true degrees, we have
    \begin{align*}
        \|x\|_2^2 \leq \sum_{i=1}^n d(i) &\leq \left ( (1+\xi)(1+\beta\alpha/\eps) + \xi/\eps \right )\vol(G)
    \end{align*}
    In order to lower bound $\langle x, \hat{N} x\rangle$, we can rewrite the sum over edges crossing the cut as the volume minus the sum of internal edges.
    \begin{align*}
        \langle x, \hat{N} x \rangle &= 2\sum_{i \in L} \sum_{j \in R} \ind\left\{ (i,j) \in E \right\} - 2\sum_{i \in L} \sum_{j \in L} \ind\left\{ (i,j) \in E \right\} - 2\sum_{i \in R} \sum_{j \in R} \ind\left\{ (i,j) \in E \right\}  \\
        &= \vol(G) - 4 \left ( E(L,L) + E(R,R)\right )
    \end{align*}
    Observe that since all vertices are partitioned into either $L$ or $R$, the bipartiteness ratio is exactly the sum of internal edges. Hence, the Rayleigh quotient $\langle x, \hat{N} x \rangle / \|x\|_2^2$ is at least:
    \begin{align*}
        \dfrac{ \vol(G) - 4 \left ( E(L,L) + E(R,R)\right )}{\left ( (1+\xi)(1+\beta\alpha/\eps) + \xi/\eps \right )\vol(G)} \geq \frac{1 - 2 \calB(G)}{\left  (1+\xi)(1+\beta\alpha/\eps) + \xi/\eps \right )} \geq 1-2\calB(G) - O((\xi+\alpha\beta)/\eps).
    \end{align*}   
\end{proof}
Lemma \ref{lem:trevisan-feasibility} shows that the noisy version of the negative adjacency matrix still preserves the relation between the bipartiteness of the graph and the top eigenvalue.  Hence we can apply Theorem \ref{thm:main} with accuracy parameter $\eps_0\in (0,1)$ on the matrix $\alpha n \cdot \hat{N}$ to get local access to an approximate top eigenvector $\bx \in \R^n$. As in the proof of Theorem \ref{thm:cheeger-dense}, we will define an error parameter $\zeta$ which captures the error in the Rayleigh quotient of the vector $\bx$ with $\hat{N}$.
\begin{align}
\zeta \leq 2\calB(G) + O((\xi + \alpha \beta)/\eps + \eps_0/\alpha) \label{eq:trevisan-error}
\end{align}
\begin{lemma}\label{lem:trevisan-output-vec} Let $\bx \in \R^n$ be the output of Theorem \ref{thm:main} on the matrix $\alpha n \cdot \hat{N}$ with accuracy parameter $\eps_0 \in (0,1)$. Then with high probability over the randomness of Theorem \ref{thm:main},
\begin{align*}
   \dfrac{\la \bx , \hat{N} \bx\ra }{\|\bx\|_2^2}\geq 1-\zeta \qquad \text{ and } \qquad \| \bx\|_{\infty} \lesssim 1/(\alpha \sqrt{n})
\end{align*}
\end{lemma}
\begin{proof}
By Theorem  \ref{thm:main}, the Rayleigh quotient of $\bx$ on $(\alpha n) \cdot \hat{N}$ is at least $\lambda_{\max}((\alpha n) \hat{N}) - \eps_0 n$, which is a least $\alpha n \cdot (1 - \zeta)$ by \ref{lem:trevisan-feasibility}. Dividing by $\alpha n$ gives the desired bound. Now note that  $\|\bx\|_{\infty}$ can be bounded by rewriting the output of the algorithm in Theorem~\ref{thm:main} (namely, Figure~\ref{fig:preprocess}), as  $\bx = A \bT^{\intercal} \by / \hat{\blambda}_{\max}$ for $A = (\alpha n) \hat{N}$. Therefore by Cauchy-Schwarz  $|\bx_i|$ is at most $1/\hat{\blambda}_{\max} \cdot (\sqrt{\bt} / \sqrt{q}) \cdot \|\by\|_2$, which is at most $O(1/(\alpha \sqrt{n}))$ with high probability.
\end{proof}

\begin{lemma}\label{lem:trevisan-vol-bound}
    Any vector $x \in \R^n$ which satisfies conditions of Lemma \ref{lem:trevisan-output-vec} with $\xi$ smaller than $\poly(\alpha \eps)$, $\alpha,\beta$ smaller than $\poly(\eps)$ and $\zeta$ smaller than some fixed constant, also satisfies
    \begin{align*}
        \Ex_{(\bi, \bj)}\left[\left(\frac{x_{\bi}}{\sqrt{d(\bi)}} + \frac{x_{\bj}}{\sqrt{d(\bj)}}\right)^2 \right] \leq \frac{2\zeta}{m}  \cdot \|x\|_2^2 \qquad \text{ and }\qquad   \Ex_{(\bi)}\left[\left(\frac{x_{\bi}}{\sqrt{d(\bi)}} \right)^2 \right] \geq \frac{\|x\|_2^2}{2m}.
    \end{align*}
where expectation over $(\bi, \bj)$ is over uniform draws from the set of edges $E$, and expectation over $\bi$ to be sampled proportionally to $\deg_G(i)$.
\end{lemma}
\begin{proof}
   The proof proceeds analogously to the proof of Lemma \ref{lem:clean-guaran}. For the first inequality, we use the fact that $d(i) \geq \deg_G(i)$ for all but at most $\xi n$ vertices, and that $d(i) \geq \alpha n$. In particular, expanding the square we obtain
    \begin{align*}
        \Ex_{(\bi,\bj)}\left[ \left(\frac{x_{\bi}}{\sqrt{d(\bi)}} + \frac{x_{\bj}}{\sqrt{d(\bj)}} \right)^2 \right] \leq \frac{2}{m} \sum_{i=1}^n \frac{\deg_G(i)}{d(i)} \cdot x_i^2 - \frac{2}{m} \langle x, \hat{N} x \rangle \leq \frac{2}{m} \left( \|x\|_2^2 - \langle x, \hat{N} x\rangle + (\xi/\alpha) n\|x\|_{\infty}^2\right),
    \end{align*}
    which is at most the desired bound. The second inequality follows from
    \begin{align}
        \Ex_{\bi}\left[ \left(\frac{x_{\bi}}{\sqrt{d(\bi)}} \right)^2 \right] \geq \Ex_{\bi}\left[ \frac{x_{\bi}^2}{d(\bi)} \right] &= \frac{1}{2m} \sum_{i=1}^n \frac{\deg_G(i)}{d(i)} \cdot x_i^2  \geq \frac{\|x\|_2^2}{2m} \left( \dfrac{1 - (\beta + \xi) \gamma_2}{1+\xi} \right). \label{eq:trev-bound-denom}
    \end{align}
\end{proof}

\begin{lemma}[Lemma 3 in \cite{L12}]\label{lem:trevisan-rounding}
    Fix any graph $G = ([n], E)$ as well as a vector $y \in \R^n$ satisfying
    \[ \Ex_{(\bi,\bj)}\left[ (y_{\bi} + y_{\bj})^2 \right] \leq a  \qquad \text{and}\qquad \Ex_{\bi}\left[ y_{\bi}^2 \right] \geq b. \]
    There exists disjoint subsets $L,R \subset [n]$ of the form $L = \left\{ i \in [n] : y_i \leq -\tau \right\}$ and $R = \left\{ i \in [n] :  y_i \geq \tau \right\}$ for $\tau \in \R$ which satisfies
    \begin{align*}
        \calB_G(L,R) \leq 4 \sqrt{ a / b} \qquad \text{ and } \qquad \vol_G(L \cup R) \geq \Omega(m b / \| y\|_{\infty}^2).
    \end{align*}
\end{lemma}

\begin{remark}[Volume of $L \cup R$] 
This remark outlines the bound on the volume of $L$ and $R$ in Lemma \ref{lem:trevisan-rounding}, which is not explicitly computed in Lemma 3 of \cite{L12}, but follows easily from the rounding scheme. The vector $y$ is rounded by sampling a threshold $\tau $ uniformly from $0$ to $\|y\|_{\infty}^2$, then setting $\bL = \{i \in [n] : y_i \geq \sqrt{\tau} \}$ and $\bR = \{i \in [n] : y_i \leq -\sqrt{\tau} \}$. The expected volume of $\bL \cup \bR$ can be computed explicitly, which is simply sum over $i \in [n]$ of the $\deg(i)$ times the probability that $y_i^2$ is larger than $\tau$. 
\begin{align*}
\E[\vol_{G}(\bL \cup \bR)] = \sum_{i} d_i \Pr[i \in \bL \cup \bR] = \frac{1}{\|y\|_{\infty}^2}\sum_{i} d_i (y_i)^2
\end{align*} 
The final bound follows from the observation that the expectation $\Ex_{\bi}\left[ y_{\bi}^2 \right]$ evaluates to $\sum_{i} d_i (y_i)^2/m$.
\end{remark}

\paragraph{Proof of Theorem \ref{thm:trevisan-dense}.} As in the proof of Theorem \ref{thm:cheeger-dense} and as discussed in (\ref{eq:noisy-negative}), we assume that the algorithm receives query access to the $n \times n$ matrix $\hat{N}$ defined in (\ref{eq:noisy-negative}), which can be done using an additive $\poly(1/(\xi \eps))$ queries to the size of the query  set $\bQ$; for any $(i,j) \in [n]\times[n]$ we may estimate $\deg_G(i)$ and $\deg_G(j)$ by checking the edges in $\bQ$ to obtain $d(i)$ and $d(j)$. From lemma \ref{lem:trevisan-feasibility}, the matrix $\alpha n \hat{N}$ is bounded entry and has a top eigenvalue of at least $\alpha n \cdot \left ( 1 - 2\calB(G) - O((\xi + \alpha \beta)/\eps)\right )$, and as a result of running (and boosting) the algorithm of Theorem \ref{thm:main}, we receive a local computation algorithm $\calA$ which with probability $1-\delta$ provides access to a vector $\bx \in \R^n$ which satisfies the conditions of lemma \ref{lem:trevisan-output-vec}. Now we can round the vector $\bx/\sqrt{d(i)}$ into a vector $\by \in \R^n$ where we round $\by$ to the nearest multiple of $\poly(\eps)/n$. A simple calculation confirms that $\by$ satisfies the conditions of Lemma \ref{lem:trevisan-rounding} with $a = 2\zeta/m$ and $b = 1/(2m)$, and each coordinate lies within a discrete set of $\poly(1/\eps)$ possible values. Hence there are at most $\poly(1/\eps)$ many possible thresholds $\tau$ and by Lemma \ref{lem:trevisan-rounding} at least one of them satisfies $\calB_G(L,R) \leq \sqrt{\calB(G)}$ and $\vol_G(L \cup R) \geq \Omega(1/\|\by\|_{\infty}^2)$. Therefore we may use $\poly(1/\eps)$ additional queries to iterate through all possible cuts and estimate $\vol_G(L \cup R)$ as well as $2E_G(L,L)$, $2E_G(R,R)$, and $ E_G(L \cup R, \ol{L \cup R})$ in order to find the desired cut. After setting the parameters $\eps_0 \leq \poly(\eps \alpha)$, $\xi \leq \poly(\alpha \beta)$ and by the assumption of Theorem \ref{thm:cheeger-dense} that $\alpha,\beta = \poly(\eps)$, we conclude that the partial cut $L,R$ returned satisfies $\calB_G(L,R) \lesssim \sqrt{\calB{\eps}(G)} + \eps$.

\section*{Acknowledgements}

The authors would like to thank Aaron Sidford for suggesting we explore local computation algorithms for top eigenvectors. 

\appendix
\section{Spectral Norm Decay}\label{app:spectral-norm-decay}

We give a proof of the main theorem of~\cite{RV07}, as well as the subsequent application in~\cite{T12}. These bounds were used in~\cite{BDDMR24,SW25}, but were quoted with a $\log n$-multiplicative factor, which introduced a dependence on $n$ into the query-complexity of the algorithm---\cite{SW25} then used a clever recursive application of the algorithm to remove this $O(\log n)$-factor. In this appendix, we show that the results of~\cite{RV07,T12} may be used directly. The reason is that, upon closer inspection, the expected spectral norm incurs a factor of $\log(n\delta)$, where $\delta$ is the probability that any row or column is included. In the application, $\delta$ is set to $1/(n\cdot \poly(\eps))$, and therefore the logarithmic factor overhead becomes $\log(n\delta) = \Theta(\log(1/\eps))$. We include a proof which loses a $O(\log(1/\eps))$-factor, which will suffice for our purposes, with the benefit that it is simpler.

\begin{theorem}[Spectral Norm Decay~\cite{RV07}]\label{thm:rv}
Let $A$ be a matrix with $n$ columns, and let $\bQ$ be a random $n\times n$ matrix with i.i.d. diagonal entries $\bdelta_i \sim \Ber(\delta)$. Then,
\[ \Ex_{\bQ}\left[ \left\| A \bQ \right\|_2\right] \lsim \left( \delta \|AA^{\intercal}\|_2 + \Ex_{\bQ}\left[ \log|\bQ| \cdot \| A \bQ \|_{1\to 2}^2\right] \right)^{1/2},\]
where $\|A\bQ\|_{1\to 2}$ is the maximum column of $A \bQ$.
\end{theorem}

The proof begins by expressing the matrix $A$ as a sum of outer-products over its columns. This will make it easy to evaluate what the sampling experiment is, and how to evaluate the expected spectral norm. In particular, let $x_1, \dots, x_n$ be the columns of $A$, and write
\[ A = \sum_{i=1}^n x_i e_i^{\intercal} \qquad \text{and}\qquad A\bQ = \sum_{i=1}^n \bdelta_i\cdot x_i e_i^{\intercal}. \]
The key here is that the spectral norm of a matrix can be computed by computing the square-root of the spectral norm of applying the matrix twice---namely, that $\|A\|_2 = \| AA^{\intercal} \|_2^{1/2} = \|A^{\intercal} A \|_2^{1/2}$, for any matrix $A$. This means, in particular, that we may compute the spectral norm of the matrices $AA^{\intercal}$ and $(A\bQ)(A\bQ)^{\intercal}$, which may be expressed as:
\begin{align*}
A A^{\intercal} = \left( \sum_{i=1}^n x_i e_i^{\intercal} \right) \left( \sum_{j=1}^n x_j e_j^{\intercal}\right)^{\intercal} = \left( \sum_{i=1}^n x_i e_i^{\intercal} \right) \left( \sum_{j=1}^n e_j x_j^{\intercal} \right) = \sum_{i=1}^n \sum_{j=1}^n \langle e_i , e_j \rangle \cdot x_i x_j^{\intercal} = \sum_{i=1}^n x_i x_i^{\intercal}, \\
(A \bQ)(\bQ A)^{\intercal} = \left( \sum_{i=1}^n \bdelta_i x_i e_i^{\intercal} \right) \left( \sum_{j=1}^n \bdelta_j x_j e_j^{\intercal}\right)^{\intercal} = \left( \sum_{i=1}^n \bdelta_i x_i e_i^{\intercal} \right) \left( \sum_{j=1}^n \bdelta_j e_j x_j^{\intercal} \right) = \sum_{i=1}^n \bdelta_i x_i x_i^{\intercal}.
\end{align*}
Thus, the goal of the theorem is to understand to what extent we may compare
\begin{align*}
\Ex_{\bdelta_1,\dots, \bdelta_n}\left[ \left\| \sum_{i=1}^n \bdelta_i \cdot x_i x_i^{\intercal} \right\|_2^{1/2} \right] \qquad \text{vs.} \qquad \sqrt{\delta} \left\| \sum_{i=1}^n x_i x_i^{\intercal} \right\|_2^{1/2}.
\end{align*}
As we will see, it will become easier for the symmetrization argument to, instead of upper bounding the expected square-root of the spectral norm, to instead compare:
\begin{align}
\Ex_{\bdelta_1,\dots, \bdelta_n}\left[ \left\| \sum_{i=1}^n \bdelta_i \cdot x_i x_i^{\intercal} \right\|_2 \right] \qquad \text{vs.} \qquad \delta \left\| \sum_{i=1}^n x_i x_i^{\intercal} \right\|_2. \label{eq:compare}
\end{align}
By Jensen's inequality, for any real-valued random variable, the square-root of the $2$nd moment is always larger than the expected magnitude of the random variable---thinking of the spectral norm of $A \bQ$ as the random variable, we compare the spectral norm squared to the spectral norm squared of $A$, and will then take square roots.

\paragraph{Matrix Chernoff.} We relate the spectral norm of $\sum \bdelta_i \cdot x_i x_i^{\intercal}$ via matrix concentration inequalities---namely, the matrix we are interested in is a sum of i.i.d rank-$1$ psd matrices ($x_i x_i^{\intercal}$) whose expectation is $\delta AA^{\intercal}$. A direct application of the matrix Chernoff bound (Theorem 1.1 in~\cite{T12}) gives:
\begin{align*}
\Prx_{\bdelta_1,\dots, \bdelta_n}\left[ \left\| \sum_{i=1}^n \bdelta_i \cdot x_i x_i^{\intercal} \right\|_2 \geq (1+\chi) \mu_{\max} \right] \leq d \cdot \left(\frac{e^{\chi}}{(1+\chi)^{1+\chi}} \right)^{\mu_{\max} / R},
\end{align*}
where:
\begin{itemize}
\item $\mu_{\max} = \| \Ex[\sum_{i=1}^n \bdelta_i x_i x_i^{\intercal}] \|_2 = \|\delta A A^{\intercal} \|_2$, giving a multiplicative error guarantee on the spectral norm.
\item $d$ represents the dimensionality of the vectors $x_1, \dots, x_n \in \R^d$, so the rank-$1$ matrices $x_i x_i^{\intercal}$ are $d\times d$ matrices.  
\item The parameter $R$ is a uniform bound on $\| x_i x_i^{\intercal}\|_2$. In our setting, we have $R = \max_i \|x_i\|_2^2$.
\end{itemize}
The main issue with using this bound is the multiplicative factor of $d$ in the above inequality. Indeed, we'd like to apply the above inequality for $\chi$ which is a large constant factor; however, any non-trivial bound would incur an additional $\log d$ factor in order to overcome the factor of $d$. Note that, in our case, we will always be sampling very few columns, and do not want to incur the additional $\log d$ factor---to see that this $\log d$ factor should not be necessary, it simply follows from the triangle inequality that the spectral norm should always be at most $|\bQ| R$, which the above arguments cannot capture, since they lose this factor of $d$. A useful aspect of the argument in~\cite{RV07} is that it makes it clear where one should not be losing the factor of $d$.

\paragraph{Symmetrization.} We first perform the symmetrization argument, where we will consider the discussion which upper bounds the expected squared spectral norm of $A\bQ$. In other words, that which appears on the left-hand side of (\ref{eq:compare}). We first apply the triangle inequality (using that the spectral norm is a norm), as well as Jensen's inequality (since all norms induce convex functions):
\begin{align*}
\left\| \sum_{i=1}^n \bdelta_i \cdot x_i x_i^{\intercal} \right\|_2 &\leq \left\| \sum_{i=1}^n (\bdelta_i - \delta) x_i x_i^{\intercal} \right\|_2 + \left\| \delta \sum_{i=1}^n x_i x_i^{\intercal} \right\|_2\\
			&= \left\| \Ex_{\bdelta_1',\dots, \bdelta_n'} \left[ \sum_{i=1}^n (\bdelta_i - \bdelta_i') x_i x_i^{\intercal} \right] \right\|_2 + \delta \| A A^{\intercal} \|_2 \\
			&\leq \Ex_{\bdelta_1', \dots, \bdelta_n'}\left[ \left\| \sum_{i=1}^n (\bdelta_i - \bdelta_i') x_i x_i^{\intercal} \right\|_2\right] + \delta \| A A^{\intercal} \|_2.
\end{align*}
Hence, when we compute the expectation over $\bdelta_1, \dots, \bdelta_n$, it suffices to upper bound:
\begin{align*}
\Ex_{\bdelta_1,\dots, \bdelta_n}\left[ \left\| \sum_{i=1}^n \bdelta_i \cdot x_i x_i^{\intercal} \right\|_2 \right] - \delta \| A A^{\intercal} \|_2 &\leq \Ex_{\substack{\bdelta_1, \dots, \bdelta_n \\ \bdelta_1', \dots, \bdelta_n'}}\left[ \left\| \sum_{i=1}^n (\bdelta_i - \bdelta_i') \cdot x_i x_i^{\intercal} \right\|_2 \right] \\
&= \Ex_{\substack{\bdelta_1, \dots, \bdelta_n \\ \bdelta_1', \dots, \bdelta_n' \\ \beps_1,\dots, \beps_n}}\left[ \left\| \sum_{i=1}^n \beps_i (\bdelta_i - \bdelta_i') \cdot x_i x_i^{\intercal} \right\|_2 \right],
\end{align*}
where $\beps_1, \dots, \beps_n$ are i.i.d and uniform $\{-1,1\}$ variables, which we may add since $\bdelta_i - \bdelta_i'$ is symmetric around the origin (because $\bdelta_i, \bdelta_i'$ are identically distributed). Finally, we may apply the triangle inequality once more and have:
\begin{align*}
\Ex_{\substack{\bdelta_1, \dots, \bdelta_n \\ \bdelta_1', \dots, \bdelta_n' \\ \beps_1,\dots, \beps_n}}\left[ \left\| \sum_{i=1}^n \beps_i (\bdelta_i - \bdelta_i') \cdot x_i x_i^{\intercal} \right\|_2 \right] &\leq \Ex_{\substack{\bdelta_1, \dots, \bdelta_n \\ \bdelta_1', \dots, \bdelta_n' \\ \beps_1,\dots, \beps_n}}\left[ \left\| \sum_{i=1}^n \beps_i \bdelta_i \cdot x_i x_i^{\intercal} \right\|_2 + \left\| \sum_{i=1}^n \beps_i \bdelta_i' \cdot x_i x_i^{\intercal} \right\|_2 \right] \\
	&= 2 \Ex_{\substack{\bdelta_1,\dots, \bdelta_n \\ \beps_1, \dots, \beps_n}}\left[ \left\| \sum_{i=1}^n \beps_i \bdelta_i \cdot x_i x_i^{\intercal} \right\|_2\right].
\end{align*}

\paragraph{Reducing Dimension and Applying Matrix Rademacher Bounds.} Consider any fixed setting of $\bdelta_1, \dots, \bdelta_n$, so we remove the unbolded variables and refer to $\delta_1,\dots, \delta_n$, and we let $Q$ denote the set of selected indices. Then, we are left with upper bounding the expected spectral norm of under the Rademacher random variables
\begin{align}
\Ex_{\beps_i : i \in Q}\left[ \left\| \sum_{i \in Q} \beps_i \cdot x_i x_i^{\intercal} \right\|_2\right].  \label{eq:spect}
\end{align}
Importantly, there are only $Q$ of them. Thus, consider the $d \times r$ matrix $E$ whose columns form an orthonormal basis of $\Span(x_i : i \in Q)$. Furthermore, we can write each $x_i = E y_i$, where $y_i \in \R^{r}$, and crucially important is that $r \leq |Q|$. In this language, we have (\ref{eq:spect}) is exactly equivalent to:
\begin{align*}
\Ex_{\beps_i : i \in Q}\left[ \left\| \sum_{i\in Q} \beps_i \cdot (E y_i) (E y_i)^{\intercal}\right\|_2 \right] = \Ex_{\beps_i : i \in Q}\left[ \left\| E \left(\sum_{i\in Q} \beps_i \cdot y_i y_i^{\intercal}\right) E^{\intercal}\right\|_2 \right] = \Ex_{\beps_i : i \in Q}\left[ \left\| \sum_{i\in Q} \beps_i \cdot y_i y_i^{\intercal} \right\|_2 \right] ,
\end{align*}
since $E$ is orthonormal, and hence does not affect the spectral norm. We now apply the matrix Rademacher series bound (Theorem~1.2 in~\cite{T12}), and bound for any $u \geq 0$,
\begin{align*}
\Ex_{\beps_i : i \in Q} \left[ \left\| \sum_{i \in Q} \beps_i \cdot y_i y_i^{\intercal} \right\|_2 \right] &= \int_{t:0}^{\infty} \Prx_{\beps_i : i \in Q}\left[ \left\| \sum_{i\in Q} \beps_i \cdot y_i y_i^{\intercal} \right\|_2 \geq t \right] dt \leq u + \int_{t:u}^{\infty} r e^{-t^2 / (2\sigma^2)} dt,
\end{align*}
since probabilities are always at most $1$ (which allows us to add $u$ up to beginning the integral at $u$), and we have $\sigma^2$ given by
\begin{align*}
\sigma^2 = \left\| \sum_{i \in Q} \| y_i\|_2^2 \cdot y_i y_i^{\intercal}  \right\|_2 \leq \max_{i \in Q} \| y_i \|_2^2 \cdot \left\| \sum_{i \in Q} y_iy_i^{\intercal} \right\|_2 = \max_{i \in Q} \| x_i\|_2^2 \cdot \left\| \sum_{i \in Q} x_i x_i^{\intercal} \right\|_2.
 \end{align*}
Note, in the above, the first inequality says that if we add up positive semi-definite matrices with positive scalar multiples (namely, the p.s.d matrices are $y_iy_i^{\intercal}$ and their positive scalar multiples $\|y_i\|_2^2$), their spectral norm is at most the sum of these p.s.d matrices times the largest scalar multiple. The inequality can be seen immediately from considering the Lowner order on positive semi-definite matrices, which says $A \preceq B$ if $B - A$ is p.s.d; then, we have $B + \sum A_i \preceq (1+\eps) B + \sum A_i$, and the spectral norm is monotone in the Lowner order. 

We will upper bound the integral by summing over all intervals of the form $[u + (j-1) \sigma, u + j \sigma]$ over all $j \in \N$, and we obtain the following simplification of the integral:
 \begin{align*}
 \Ex_{\beps_i : i \in Q} \left[ \left\| \sum_{i \in Q} \beps_i \cdot x_i x_i^{\intercal} \right\|_2 \right] &\leq u + r \int_{t:u}^{\infty} \exp\left( \frac{-t^2}{2\sigma^2} \right) dt \leq u + r \sigma \sum_{j=1}^{\infty} \exp\left( -\frac{(u + (j-1) \sigma)^2}{2\sigma^2} \right) \\
 			&\leq u + r \sigma \cdot \exp\left(-\frac{u^2}{2\sigma^2}\right) \sum_{j=1}^{\infty} \exp\left(-(j-1)^2 \right) \lsim \sigma \sqrt{\log r},
 \end{align*}
 by setting $u = O(\sigma \sqrt{\log(r)})$. 
 
 \paragraph{The Final Trick of~\cite{RV07}.} We are now in good shape, since may consider back the random choice of $\bdelta_1, \dots, \bdelta_n$, and obtain the upper bound:
 \begin{align*}
 \Ex_{\bdelta_1,\dots, \bdelta_n}\left[ \left\| \sum_{i=1}^n \bdelta_i \cdot x_i x_i^{\intercal} \right\|_2 \right] - \delta \| A A^{\intercal} \|_2 &\lsim  \Ex_{\bdelta_1,\dots, \bdelta_n}\left[ \sqrt{\log |\bQ|} \cdot \max_{i \in \bQ} \| x_i\|_2 \cdot \left\| \sum_{i=1}^n \bdelta_i \cdot x_i x_i^{\intercal}  \right\|_2^{1/2} \right] \\
 	&\lsim \left( \Ex_{\bdelta_1,\dots, \bdelta_n}\left[\log |\bQ| \cdot \max_{i \in \bQ} \|x_i \|_2^2 \right]\right)^{1/2}\left( \Ex_{\bdelta_1,\dots, \bdelta_n}\left[ \left\|\sum_{i=1}^n \bdelta_i \cdot x_i x_i^{\intercal} \right\|_2 \right] \right)^{1/2},
 \end{align*}
where the second inequality is an application of Cauchy-Schwarz. As~\cite{RV07} write, we have given an expression of the form $E \lsim a + \sqrt{b E}$, which can only happen if $E$ is bounded (otherwise, the \emph{linear-in-}$E$ term on the left-hand side would grow faster than the \emph{square-root-in-}$E$ term on the right-hand side. This can only happen if $E \lsim a + b^2$, so we obtain:
 \begin{align*}
 \Ex_{\bdelta_1,\dots, \bdelta_n}\left[  \left\| \sum_{i=1}^n \bdelta_i \cdot x_i x_i^{\intercal} \right\|_2\right] \lsim \delta \| A A^{\intercal} \|_2 + \Ex_{\bdelta_1,\dots, \bdelta_n}\left[ \log |\bQ| \cdot \max_{i \in \bQ} \| x_i\|_2^2 \right].
 \end{align*}

Now given Theorem \ref{thm:rv}, we can analyze the case of random principle sub-matrices following~\cite{Tro08}.
\begin{theorem}[Spectral Norm of Random Principle Submatrices~\cite{Tro08}]\label{thm:tropp}
Let $A$ be a matrix and $\bQ \subset [n]$ be a random subset which includes each element i.i.d with probability $\delta$. Then,
\begin{align*}
\Ex_{\bQ}\left[ \left\| \bQ A \bQ \right\|_2 \right] \lsim \delta \| A\|_2 + \sqrt{\log(n\delta)\delta} \cdot \| A \|_{1\to2} + \log(n\delta) \cdot \| A \|_{\infty}.
\end{align*}
\end{theorem}

As with the proof of Theorem~\ref{thm:rv}, we will instead consider $\bQ$ as an $n\times n$ diagonal matrix whose entries are i.i.d. Bernoulli r.v.'s e.g. a matrix where $\bQ_{i,i} = \bdelta_i \sim \Ber(\delta)$. The proof will proceed by first ``decoupling" the result into two i.i.d. sampling matrices, then by applying Theorem~\ref{thm:rv} twice.

\begin{proposition}[Decoupling Prop 2.1 in \cite{Tro08}]\label{prop:decoupling} Let $H$ be a symmetric matrix with zero-diagonal. Then for uniform diagonal sampling matrix $\bQ$,
\[\Ex_{\bQ}\left[ \left\|\bQ H \bQ \right\|_2 \right] \leq 2 \Ex_{\bQ,\bQ'}\left[ \left\| \bQ H \bQ' \right\|_2 \right] \]
where $\bQ$ and $\bQ'$ are independent and identically distributed. 
\end{proposition}

\ignore{\begin{proposition}[Recoupling Prop 2.2 in \cite{Tro08}]\label{prop:recoupling} Let $A$ be a symmetric matrix with zero-diagonal. Then for a uniform diagonal sampling matrix $\bQ$,
\[\Ex_{\bQ,\bQ'}\left[ \left\|\bQ A \bQ' \right\|_{\infty} \right] \leq 4 \Ex_{\bQ}\left[ \left\| \bQ A \bQ \right\|_{\infty} \right] \]
where $\bQ$ and $\bQ'$ are independent and identically distributed. 
\end{proposition}}

\textbf{Applying Rudelson-Vershynin Twice.} The proof begins by splitting up the matrix $A$ into the diagonal $D$ and off diagonal part $H$, then applying a ``decoupling" lemma that allows us to relate the sampled principle submatrix to a randomly sampled submatrix (where the rows and columns are sampled independently) which will allow us to apply Theorem~\ref{thm:rv} to the off-diagonal part of the matrix. We are now left with the expected spectral norm of the \emph{row-sampled} matrix and an expected column norm bound in terms of entire submatrix.
\begin{align}
    \Ex_{\bQ}\left[ \left\|\bQ A \bQ \right\|_2 \right] 
    &\leq \Ex_{\bQ}\left[ \left\| \bQ H \bQ \right\|_2 \right]\ +  \Ex_{\bQ}\left[ \left\| \bQ D \bQ \right\|_2 \right] \leq 2 \Ex_{\bQ'} \Ex_{\bQ}\left[ \left\| \bQ H \bQ' \right\|_2 \right] +  \|D\|_{\infty} \nonumber \\
    &\lsim \Ex_{\bQ}\left[ \left( \delta \|\bQ H(\bQ H)^{\intercal}\|_2 + \Ex_{\bQ'}\left[ \log|\bQ'| \cdot \| \bQ H \bQ' \|_{1\to 2}^2\right] \right)^{1/2} \right] + \|D\|_{\infty} \nonumber \\
    &\leq \sqrt{\delta} \Ex_{\bQ}\left[ \|\bQ H\|_2 \right] + \Ex_{\bQ}\left[ \left( \Ex_{\bQ'} \left[\log|\bQ'| \cdot \| \bQ H \bQ' \|_{1\to 2}^2\right]\right)^{1/2}\right]  + \|D\|_{\infty}. \label{eq:first-rv}
\end{align}
We will deal with each term of (\ref{eq:first-rv}) separately. We start with $ \Ex_{\bQ}\left[\| \bQ H\|_2\right]$. Using the invariance of the spectral norm under transpose, and the fact that $\bQ$ and $H$ are symmetric, we can apply Theorem~\ref{thm:rv}, and we obtain
\begin{align*}
\sqrt{\delta} \Ex_{\bQ}\left[\| \bQ H\|_2\right] = \sqrt{\delta} \Ex_{\bQ}\left[\| H \bQ\|_2 \right] \lsim  \delta \| H\|_2 + \sqrt{\delta\log(n\delta)} \| H \|_{1\to2},
\end{align*}
and note that by the triangle inequality, $\| H \|_{2} \leq \|A\|_2 + \|D \|_{2} = \| A \|_2 + \|A\|_{\infty}$.

\textbf{Bounding the Column Norm.} Now we turn to the second term of (\ref{eq:first-rv}). By Jensen's inequality and the definition of the column norm bound, 
\begin{align}
\Ex_{\bQ}\left[ \left( \Ex_{\bQ'} \left[\log|\bQ'| \cdot \| \bQ H \bQ' \|_{1\to 2}^2\right]\right)^{1/2}\right] 
&\leq \left(  \Ex_{\bQ'}\left[ \log|\bQ'| \Ex_{\bdelta_1,\dots,\bdelta_n} \left[\max_{i \in \bQ'} \sum_{j=1}^n \bdelta_j \left (H_{j,i} \right )^2 \right] \right] \right)^{1/2}. \label{eq:val-2}
\end{align}
For any fixed $\bQ'$, an analogous symmetrization argument gives the following expression:
\begin{align*}
    \Ex_{\bdelta_1, \dots, \bdelta_n}\left[ \max_{i \in \bQ'} \sum_{j=1}^n \bdelta_j H_{j,i}^2 \right] &\leq \delta \| H \|_{1 \to 2}^2 + 2 \Ex_{\substack{\bdelta_1,\dots, \bdelta_n \\ \beps_1,\dots, \beps_n}}\left[ \max_{i \in \bQ'} \sum_{j=1}^n \beps_j \bdelta_j H_{j,i}^2 \right] \\ &\leq \delta \| H\|_{1\to 2}^2 + 2\sqrt{\log|\bQ'|} \cdot \| H \|_{\infty} \left(\Ex_{\bdelta_1,\dots, \bdelta_n}\left[ \max_{i \in \bQ'} \sum_{j=1}^n \bdelta_j H_{j,i}^2\right] \right)^{1/2},
\end{align*}
which again implies (via the same ``square-root-$E$'' upper bounds ``linear-in-$E$'' trick),
\[ \Ex_{\bdelta_1,\dots, \bdelta_n}\left[ \max_{i \in \bQ'} \sum_{j=1}^n \bdelta_j H_{j,i}^2 \right] \lsim \delta \| H\|_{1\to 2}^2 + \log|\bQ'| \cdot \| H \|_{\infty}^2. \]
Substituting into (\ref{eq:val-2}), the left-hand side of (\ref{eq:val-2}), and using Jensen's inequality to take expectation over $|\bQ'|$, we obtain a bound of (up to a constant factor)
\begin{align*}
\sqrt{\log(n\delta) \cdot \delta} \cdot \| H \|_{1\to 2} + \log(n\delta) \cdot \| H \|_{\infty}.
\end{align*}
Finally, note $\|H\|_{1\to 2} \leq \|A \|_{1\to2}$, and $\| H \|_{\infty} \leq \|A\|_{\infty}$.

\ignore{
\textbf{Symmetrization.} We consider a fixed setting of $\bQ'$ (and henceforth unbold the random variable to consider a set $Q'$). We seek to upper bound:
\begin{align*}
    \Ex_{\bdelta_1, \dots, \bdelta_n}\left[ \max_{i \in Q'} \sum_{j=1}^n \bdelta_j H_{j,i}^2 \right] &\leq \Ex_{\substack{\bdelta_1, \dots, \bdelta_n \\ \bdelta_1', \dots, \bdelta_n' \\ \beps_1,\dots, \beps_n}}\left[ \max_{i \in Q'} \sum_{j=1}^n \beps_j\left(\bdelta_j - \bdelta_j'\right) H_{j,i}^2 \right] + \delta \| H \|_{1\to 2}^2  \\
    &\leq 2 \Ex_{\substack{\bdelta_1, \dots, \bdelta_n \\ \beps_1,\dots, \beps_n}}\left[ \max_{i \in Q'} \sum_{j=1}^n \beps_j\bdelta_j H_{j,i}^2 \right] + \delta \| H \|_{1\to 2}^2 .
\end{align*}
Consider a fixed setting of $\bdelta_1, \dots, \bdelta_n$ (and hence unbold these), and we can upper bound using Bernstein's inequality,
\begin{align*}
\Ex_{\beps_1, \dots, \beps_n}\left[ \max_{i \in Q'} \sum_{j=1}^n \beps_j \delta_j H_{j,i}^2 \right] &\leq \int_{t:0}^{\infty} \sum_{i \in Q'} \Prx_{\beps_1, \dots, \beps_n}\left[\sum_{j=1}^n \beps_j \delta_j H_{j,i}^2 \geq t\right] dt \\
        &\leq \int_{t:0}^{\infty} \sum_{i \in Q'}\exp\left(-\dfrac{t^2/2}{ \sum_{j=1}^n \delta_j H_{j,i}^4 + \| H\|_{\infty} t/3 }\right) dt \\
        &\lsim \log|Q'| \cdot \| H \|_{\infty} + \sqrt{\log|Q'|} \cdot \max_{i \in Q'} \sum_{j=}   
\end{align*}

Now we perform the symmetrization argument, where we we first center the random variables, then use the equivalence of the expectation to insert i.i.d. copies of the random variables, before using the fact that the random variables have a new distribution in order to insert i.i.d. uniform $\{-1,1\}$ random variables. The last step just applies triangle inequality and removed the i.i.d. copies of $\bdelta_1,\dots,\bdelta_n$.
\begin{align*}
     &\Ex_{\bQ'} \left[\log|\bQ'| \cdot \Ex_{\bdelta_1,\dots,\bdelta_n} \left[ \max_{i \in \bQ'} \sum_j \bdelta_j \left (H_{j,i} \right )^2 \right] \right] \\ 
     &\qquad\qquad\qquad= \Ex_{\bdelta_1,\dots,\bdelta_n} \Ex_{\bQ'} \left[\log|\bQ'| \cdot  \max_{i \in \bQ'} \sum_j (\bdelta_j -\delta) \left (H_{j,i} \right )^2 \right] + \delta \Ex_{\bQ'} \left[\log{|\bQ'|}\cdot\|H\bQ'\|_{1\rightarrow 2}^2\right]\\
     &\qquad\qquad\qquad\leq \Ex_{\substack{\bdelta_1, \dots, \bdelta_n \\ \bdelta_1', \dots, \bdelta_n' \\ \beps_1,\dots,\beps_n}} \Ex_{\bQ'} \left[\log|\bQ'| \cdot  \max_{i \in \bQ'} \sum_j \beps_j(\bdelta_j -\bdelta_j') \left (H_{j,i} \right )^2 \right] + \delta\Ex_{\bQ'}\left[\log{|\bQ'|}\cdot\|H\bQ'\|_{1\rightarrow 2}^2\right]\\
     &\qquad\qquad\qquad\leq 2 \Ex_{\bQ'} \left[\log|\bQ'| \Ex_{\substack{\bdelta_1, \dots, \bdelta_n \\ \beps_1,\dots,\beps_n}}\left[  \max_{i \in \bQ'} \sum_j \beps_j \bdelta_j \left (H_{j,i} \right )^2 \right] \right]+ \delta\Ex_{\bQ'} \left[\log{|\bQ'|}\cdot\|H\bQ'\|_{1\rightarrow 2}^2 \right]\\
\end{align*}
We ignore the rightmost term and instead turn to the leftmost term. Consider any fixed setting of $\bQ'$ (call it $Q'$, since it is no longer a random variable). We first use an $\ell_s$-norm (for appropriately chosen $s$) instead of the ``max'', then applying Jensen's inequality, and Khintchine's inequality. 
\begin{align*}
    \Ex_{\substack{\bdelta_1, \dots, \bdelta_n \\ \beps_1,\dots,\beps_n}} \left[\max_{i \in Q'} \sum_j \beps_j \bdelta_j \left (H_{j,i} \right )^2 \right] 
    &\leq  \Ex_{\substack{\bdelta_1, \dots, \bdelta_n \\ \beps_1,\dots,\beps_n}} \left[ \left( \sum_{i \in Q'} \left( \sum_j \beps_j \bdelta_j \left (H_{j,i} \right )^2 \right)^s\right)^{1/s} \right] \\
    &\leq  \Ex_{\substack{\bdelta_1, \dots, \bdelta_n}} \left[\left( \sum_{i \in Q'} \Ex_{\beps_1,\dots,\beps_n} \left [ \left( \sum_j \beps_j \bdelta_j \left (H_{j,i} \right )^2 \right)^s \right]\right)^{1/s} \right]  \\
    &\leq |Q'|^{1/s} \Ex_{\substack{\bdelta_1, \dots, \bdelta_n}} \left[\max_{i \in Q'} \left\{ \Ex_{\beps_1,\dots,\beps_n} \left [ \left( \sum_j \beps_j \bdelta_j \left (H_{j,i} \right )^2 \right)^s  \right]^{1/s} \right\} \right]  \\
    &\lsim \sqrt{s} |Q'|^{1/s}  \Ex_{\substack{\bdelta_1, \dots, \bdelta_n}} \left[  \max_{i \in Q'} \left\{ \left( \sum_j \bdelta_j \left (H_{j,i} \right )^4  \right)^{1/2} \right\} \right].
\end{align*}
Now we set $p = \log(|Q'|)$ and bound $\Ex_{\beps_1, \dots, \beps_n} \max_{i \in Q'} \left( \sum_j \bdelta_j \left (H_{j,i} \right )^4  \right)^{1/2}$ by the entry-wise max in the sub-array we are considering times the maximum column norm. We will also release the randomness of $\bQ'$.

\begin{align*}
    \Ex_{\substack{\bdelta_1, \dots, \bdelta_n}} \left[ \log^2|Q'| \cdot \max_{i \in Q'} \left( \sum_j \bdelta_j \left (H_{j,i} \right )^4  \right)^{1/2} \right] 
    &\leq \Ex_{\substack{\bdelta_1, \dots, \bdelta_n}} \left[ \log^2|Q'| \cdot \max_{i \in Q'} \left( \sum_j \bdelta_j \max_{k,l}(H_{k,l})^2 H_{j,i}^2  \right)^{1/2} \right] \\
    &\leq \Ex_{\substack{\bdelta_1, \dots, \bdelta_n}} \left[ \log^2|Q'| \left ( \max_{\substack{i \in Q'\\ j \in \bQ}} |H_{j,i}| \right  )\cdot \max_{i \in Q'} \left( \sum_j \bdelta_j H_{j,i}^2  \right)^{1/2} \right] \\
    &\leq \Ex_{\bQ,\bQ'} \left[ \log^2|\bQ'| \cdot \|\bQ H\bQ'\|_{\infty} \cdot  \max_{i \in \bQ'} \|\bQ \cdot x_i\| \right] 
\end{align*}
Now we re-introduce the leftmost term and observe that we have found an inequality of the form $E^2 \leq \alpha \cdot E + \beta$, which must mean that $E$ is bounded by at most $\alpha + \sqrt{\beta}$.  Therefore we can bound the terms by 
\begin{align*}
    \Ex_{\bQ}\left (\Ex_{\bQ'} \log|\bQ'| \cdot \left ( \max_{i \in \bQ'} \|\bQ \cdot x_i\|\right )^2 \right )^{1/2}&\leq  \Ex_{\bQ,\bQ'} \left[ \log^2|\bQ'| \cdot \|\bQ H\bQ'\|_{\infty} \right ] + \Ex_{\bQ'} \sqrt{\delta\log{|\bQ'|}}\cdot\|H\bQ'\|_{1\rightarrow 2}
  \end{align*}
Now we can simplify further to bound these terms by 
\[\log^2(n\delta)\|H\|_{\infty} + \sqrt{\delta\log(n\delta)} \| A \|_{(1/\delta)}\]
concluding the proof. }

\section{Negative Eigenvalue Interference}\label{app:negative-interfere}

The goal of this section is to point the reader to exactly where the ``negative eigenvalue interference'' shows up, and we will do by giving a concrete example where the analysis leads to a complexity degradation. Let $u, v$ be two unit vectors with $\langle u, v \rangle = 0$, and consider the matrix $A = \lambda_{\max} u u^{\intercal} - \lambda_{\min} v v^{\intercal}$, where we will later set $\lambda_{\min} = \Theta(n)$ and $\lambda_{\max} = \Theta(\eps n)$. Note that the algorithm outputs $\bx = A \bT^{\intercal} \by / \hat{\blambda}_{\max}$, where $\by$ is the top eigenvector of $\bT A \bT^{\intercal}$. 

First, we note that in this rank-$2$ case, we may compute the Rayleigh quotient of $\bx$ with $A$ as a function of how $\bT$ acts on $u$ and $v$. Let $\tilde{\bu} = \bT u$ and $\tilde{\bv} = \bT v$ and $\delta = \langle \tilde{\bu}, \tilde{\bv}\rangle$. Note that $\bT A \bT^{\intercal} = \lambda_{\max} \tilde{\bu} \tilde{\bu}^{\intercal} + \lambda_{\min} \tilde{\bv} \tilde{\bv}^{\intercal}$, and write the vector $\by = \alpha (\tilde{\bu} + c \tilde{\bv})$, for some scale factor $\alpha$. Note, the fact $\by$ is an eigenvector means that it satisfies the following eigenvalue equation:
\begin{align*}
\bT A \bT^{\intercal} \by &= \alpha \lambda_{\max} (\|\tilde{\bu}\|_2^2 + c \delta) \cdot \tilde{\bu} - \alpha \lambda_{\min} (c\|\tilde{\bv}\|_2^2 + \delta) \cdot\tilde{\bv}  \\
	&= \hat{\blambda}_{\max} \by = \hat{\blambda}_{\max} \alpha \cdot \tilde{\bu} + \hat{\blambda}_{\max}  \alpha c \cdot \tilde{\bv}.
\end{align*}
As long as $\tilde{\bu}$ and $\tilde{\bv}$ span a two-dimensional subspace, the above implies that the individual coefficients on $\tilde{\bu}$ and $\tilde{\bv}$ must be equal, which implies
\begin{align*}
\lambda_{\max} \cdot \|\tilde{\bu}\|_2^2 + \lambda_{\max} \cdot c\delta &= \hat{\blambda}_{\max} \\
-\lambda_{\min} \cdot c \| \tilde{\bv}\|_2^2 - \lambda_{\min} \cdot \delta &= \hat{\blambda}_{\max} \cdot c,
\end{align*}
so the parameter $c$ must satisfy the following quadratic equation,
\begin{align*}
c \left( \lambda_{\max} \| \tilde{\bu}\|_2^2 + \lambda_{\min} \| \tilde{\bv}\|_2^2 \right) + \lambda_{\max} \cdot c^2 \delta = - \lambda_{\min} \cdot \delta,
\end{align*}
and this means that 
\begin{align*}
c = - \dfrac{1}{2\lambda_{\max} \delta} \left( (\lambda_{\max} \| \tilde{\bu}\|_2^2 + \lambda_{\min} \|\tilde{\bv}\|_2^2) + \sigma \left( (\lambda_{\max} \| \tilde{\bu}\|_2^2 + \lambda_{\min} \|\tilde{\bv}\|_2^2)^2 - 4 \lambda_{\max} \lambda_{\min} \delta^2 \right)^{1/2} \right),
\end{align*}
for a sign $\sigma \in \{-1,1\}$. As long as $\delta$ is smaller than a sufficiently small constant, $\hat{\blambda}_{\max} \geq \lambda_{\max} - \eps n$, and $\|\tilde{\bu}\|_2^2$ is close to $1$, it must be that $\sigma = - 1$; otherwise, $c$ becomes $-\Theta((\lambda_{\max}\|\tilde{\bu}\|_2^2 + \lambda_{\min}\|\tilde{\bv}\|_2^2) / (2 \lambda_{\max} \delta))$ and $\hat{\blambda}_{\max}$ would become too small (effectively, this setting leads to the negative eigenvalue close to $\lambda_{\min}$). Thus, taking a Taylor expansion of the inner-most square-root and using the fact $\delta$ is a sufficiently small constant, we obtain the following approximation for $c$, 
\begin{align*}
c &= - \frac{1}{2\lambda_{\max} \delta} \left( \dfrac{4 \lambda_{\max} \lambda_{\min} \delta^2}{2 (\lambda_{\max}\|\tilde{\bu}\|_2^2 + \lambda_{\min}\|\tilde{\bv}\|_2^2)} \right) \pm \Theta\left( \frac{(\lambda_{\max} \lambda_{\min} \delta^2)^2}{\lambda_{\max}\delta (\lambda_{\max}\|\tilde{\bu}\|_2^2 + \lambda_{\min}\|\tilde{\bv}\|_2^2)^3}\right) \\
	&= - \frac{\lambda_{\min} \delta}{(\lambda_{\max}\|\tilde{\bu}\|_2^2 + \lambda_{\min}\|\tilde{\bv}\|_2^2)} \pm \Theta\left( \dfrac{\lambda_{\max} \delta^3}{\lambda_{\min}}\right),
\end{align*}
where we further used the fact that $\| \tilde{\bv}\|_2^2$ will be close to $1$ and $\lambda_{\min} \geq \lambda_{\max}$. Note, it remains to understand what the output Rayleigh quotient will be with the vector $\bx$, which can be expressed by
\begin{align*}
\dfrac{\langle \bx,A \bx \rangle}{\|\bx\|_2^2} &= \lambda_{\max} \cdot (1 - \chi) - \lambda_{\min} \cdot \chi = \lambda_{\max} - \chi \left(\lambda_{\min} - \lambda_{\max} \right),
\end{align*}
where represents the ratio of squared-$\ell_2$-mass on coordinates of $v$ in $\bx$, so $\chi = \langle \bx, v\rangle^2 / (\langle \bx, u\rangle^2 + \langle \bx, v\rangle^2)$ where the inner products are:
\begin{align*}
\langle \bx, v\rangle &= \frac{1}{\hat{\blambda}_{\max}} \cdot \langle A \bT^{\intercal} \by, v \rangle = \frac{\lambda_{\min}}{\hat{\blambda}_{\max}} \cdot \langle \by, \bT v\rangle = \frac{\alpha \cdot \lambda_{\min}}{\hat{\blambda}_{\max}} \cdot (\delta + c \|\tilde{\bv}\|_2^2), \\
\langle \bx, u\rangle &= \frac{1}{\hat{\blambda}_{\max}} \cdot \langle A \bT^{\intercal} \by, u\rangle = \frac{\lambda_{\max}}{\hat{\blambda}_{\max}} \cdot \langle \by, \bT u \rangle = \frac{\alpha \cdot \lambda_{\max}}{\hat{\blambda}_{\max}} \cdot (\|\tilde{\bu}\|_2^2 + c\delta).
\end{align*}
Notice, we must certainly have $\langle \bx, v \rangle \leq \langle \bx, u \rangle$, since otherwise, we'd already be done as the Rayleigh quotient would be negative when $\lambda_{\min} \geq \lambda_{\max}$. Hence, the ratio $\chi$ will be, up to a constant factor, $\langle \bx, v\rangle^2 / \langle \bx, u\rangle^2$. Note, $\hat{\blambda}_{\max}$ and $\alpha$ cancel out, and plugging in $c$ gives us:
\begin{align*}
\dfrac{\langle \bx, v\rangle^2}{\langle \bx, u \rangle^2} \gsim  \dfrac{\lambda_{\min}^2 \cdot \delta^2 (\lambda_{\max}\|\tilde{\bu}\|_2^2)^2 / (\lambda_{\min}\|\tilde{\bv}\|_2^2)^2}{\lambda_{\max}^2 (\|\tilde{\bu}\|_2^2 + c\delta)^2} \gsim \delta^2.
\end{align*}
Recall $\delta = \langle \bT u, \bT v\rangle$, and here is where one can see the weakness in our analysis. The fact that $\bT$ is an $L / \lambda$-subspace embedding for the span of eigenvectors of magnitude at least $\lambda$, means we guarantee an upper bound $|\delta| \leq L / \lambda_{\max}$, but there our analysis cannot rule out $|\delta| = \Theta(L/\lambda_{\max})$, and this leads to a degradation. Namely, the Rayleigh quotient on $\bx$ would become
\[ \lambda_{\max} - \Omega(\delta^2) \cdot \lambda_{\min}.\]
If $\lambda_{\max} = \Theta(\eps n)$ and $\lambda_{\min} = \Theta(n)$, we'd need to set $L = \eps^{3/2} n$ for the error to be smaller than $\eps n$. This leads to query complexity $1/\eps^6$ during the preprocessing and $1/\eps^3$ during the query time. 

\bibliography{waingarten}
\bibliographystyle{alpha}

\end{document}